\documentclass[a4paper,11pt]{amsart}
\usepackage{amssymb}
\usepackage{amscd}
\usepackage{comment}
\usepackage{amsmath,amsthm}
\usepackage[colorlinks=true]{hyperref}
\usepackage{enumerate}
\usepackage{booktabs,multirow}
\usepackage{tikz}
\usetikzlibrary{patterns}

\allowdisplaybreaks[1]
\title[A positive integral property on the ground state of 2BTL Hamiltonian]
{A Positive integral property on the ground state of the two-boundary 
Temperley--Lieb Hamiltonian}
\author[K.~Shigechi]{Keiichi~Shigechi}
\email{k1.shigechi AT gmail.com}
\date{\today}

\newcommand\linkpattern[4]{
\foreach \x/\y in {#1}
\draw[thick](0.8*\x,0)..controls (0.8*\x,-0.6*\y+0.6*\x) 
		and (0.8*\y,-0.6*\y+0.6*\x)..(0.8*\y,0);
\foreach \x/\y in {#2}
\draw[dashed,thick](0.8*\x,0)..controls (0.8*\x,-0.6*\y+0.6*\x) 
		and (0.8*\y,-0.6*\y+0.6*\x)..(0.8*\y,0);
\foreach \x/\y in {#3}
	\draw(0.8*\x,0)--(0.8*\x,-0.5)node[anchor=north]{\rm{\y}};
\foreach \x/\y in {#4}
	\draw(0.8*\x,0)--(0.8*\x,-0.5)
	node[circle,inner sep=1pt,draw,anchor=north]{\rm\y};
}
\newcommand\upa[1]{
\draw(#1,0)--(#1,-0.6)(#1-0.12,-0.12)--(#1,0)--(#1+0.12,-0.12);
}
\newcommand\downa[1]{
\draw(#1,0)--(#1,-0.6)(#1-0.12,-0.6+0.12)--(#1,-0.6)--(#1+0.12,-0.6+0.12);
}
\newcommand\tikzpic[2]{
\raisebox{#1\totalheight}{
\begin{tikzpicture}
#2
\end{tikzpicture}
}}

\newtheorem{theorem}{Theorem}[section]
\newtheorem{example}[theorem]{Example}
\newtheorem{lemma}[theorem]{Lemma}
\newtheorem{defn}[theorem]{Definition}
\newtheorem{prop}[theorem]{Proposition}
\newtheorem{cor}[theorem]{Corollary}

\newtheorem{conj}[theorem]{Conjecture}
\newtheorem{remark}[theorem]{Remark}

\begin{document}
\begin{abstract}
We study the two-boundary Temperley--Lieb $O(n)$ loop model 
on Kazhdan--Lusztig bases of type A and B.
We obtain explicit expressions of the ground state of 
the two-boundary Temperley--Lieb Hamiltonian by means of 
a coideal subalgebra of $U_q(\mathfrak{sl}_2)$.
This ground state possesses a positive integral property.
We conjecture that some components of the ground state 
are directly related to an enumeration of binary or permutation 
matrices.
\end{abstract}

\maketitle

\section{Introduction}
Razumov and Stroganov studied the ground state of the $XXZ$ spin chain 
with the periodic boundary condition at the isotropic parameter 
$\Delta=-1/2$ \cite{RazStr01-1,RazStr04}. 
They observed that the largest component (under a certain normalization)
is related to the total number of combinatorial objects called 
alternating sign matrices (see, {\it e.g.}, \cite{Bres99} and references therein).
They also obtained explicit expressions for some correlation functions.
These striking observations open a new way to study quantum integral
systems \cite{BatdeGNie01,MitNiedeGBat04,RazStr01-2,RazStr05,StoVla13}, quantum Knizhnik--Zamolodchikov 
equations \cite{deGPya10,DiFZJ05-1,DiFZJ05-2,DiFPZJ07,RazStrPZJ07} and  
combinatorics such as alternating sign matrices and plane 
partitions \cite{BehDiFPZJ,deGPyaPZJ09,DiFZJ05-2}.

An alternative description of the Razumov--Stroganov correspondence 
(proven in~\cite{CanSpor10}) is the one by the $O(n)$ loop model \cite{BatdeGNie01}.
This model has the Temperley--Lieb Hamiltonian \cite{Mar90,Sal89,TL71} which acts 
on the so-called link patterns \cite{deG05,deGRit04}. 
At $n=1$, it was observed that the ground state on the link pattern bases 
counts the number of combinatorial objects such as alternating 
sign matrices and fully packed loop models \cite{deGRit04}.
There are several variants of the model by replacing the 
Temperley--Lieb algebra to the one- or two-boundary 
Temperley--Lieb algebras \cite{deGNic09,deGRit04,MarSal94,NicRitdeG05}.
These algebras act on link patterns with boundaries.

In the case of type A, the link pattern bases for the Temperley--Lieb 
algebra is equivalent to the parabolic Kazhdan--Lusztig bases 
of the Hecke algebra \cite{Deo87,KL79} or to the (dual) canonical basis (of 
weight zero) \cite{Kas90,Kas93,Lus90-1,Lus90-2} in the tensor products of the fundamental 
representation of $U_q(\mathfrak{sl}_2)$ \cite{FK97}.
The coincidence of Kazhdan--Lusztig bases and canonical bases 
was shown in \cite{FreKhoKir98}.
The Kazhdan--Lusztig bases for the Hermitian symmetric pair 
$(B_N,A_{N-1})$ are relevant to the representation of the 
one-boundary Temperley--Lieb algebra considered in this paper.
The graphical description of these Kazhdan--Lusztig bases 
is revealed in \cite{Shi14-1}. 
A diagram for a Kazhdan--Lusztig basis has the flavour 
of a link pattern of type A and type B. 
The two-boundary Temperley--Lieb algebra can be regarded as a
quotient algebra of the affine Hecke algebra of type C. 
Although the affine Hecke algebra is infinite dimensional, we 
impose a quotient relation (see Eqn.(\ref{quotient2BTL})) to 
obtain a finite dimensional representation of the two-boundary 
Temperley--Lieb algebra.
Since the (one-boundary) Temperley--Lieb algebra is a subalgebras 
of the two-boundary Temperley--Lieb algebra, 
the representation theory of the former can be applicable to the latter.

In this paper, we investigate the two-boundary Temperley--Lieb 
$O(n)$ loop model on the Kazhdan--Lusztig bases of type A and B 
and on the standard bases.
The Kazhdan--Lusztig bases of type B can be regarded as a 
(dual) canonical basis of a coideal subalgebra 
of $U_q(\mathfrak{sl}_2)$ \cite{Shi14-2}.  
A key is the fact that the generator of the coideal subalgebra commutes 
with the one-boundary Temperley--Lieb Hamiltonian of the system.
Thus an eigenfunction of the generator of the coideal subalgebra
with the multiplicity one is also an eigenfunction of the one-boundary 
Temperley--Lieb Hamiltonian. 

We have four types of Kazhdan--Lusztig bases which the two-boundary
Temperley--Lieb algebra acts on.
We call them type A, BI, BII and BIII respectively 
(see Section~\ref{sec-reps} for definitions).
Type BII corresponds to a link pattern with a boundary with a suitable 
choice of the normalization (compare Type BII with ,{\it e.g.}, \cite{deG05}). 
We first consider the eigensystem of the generator $X$ of a coideal 
subalgebra of $U_q(\mathfrak{sl}_2)$.
From the explicit action of $X$ on the Kazhdan--Lusztig bases, 
one can obtain all the eigenvalues of $X$ and their multiplicities.
For all types, there exists an eigenvector $\Psi$ with 
the multiplicity one.
One of the main results of this paper is explicit expressions 
of $\Psi$ (see Definition~\ref{PsiA}, \ref{PsiBI}, \ref{PsiBII} 
and \ref{PsiBIII}).
By construction, it is obvious that $\Psi$ is a positive vector, 
{\it i.e.}, all the entries of $\Psi$ are positive. 
We also show that this eigenfunction $\Psi$ is an eigenvector 
of the generators of one-boundary Temperley--Lieb algebra 
with the eigenvalue zero.
This implies that $\Psi$ is the ground state of the one-boundary 
Temperley--Lieb Hamiltonian 
with the eigenvalue zero.
Furthermore, this eigenfunction $\Psi$ is the ground state of 
the two-boundary Temperley--Lieb Hamiltonian under an integrable 
condition $q^{N-1}QQ_{0}=1$ where $N$ is the size of the system  
and $q, Q$ and $Q_{0}$ are the Hecke parameters 
(see Section~\ref{sec-TL} for the definitions of parameters).
This integrable condition can be viewed as a compatibility condition
to embed the representation of the two-boundary Temperley--Lieb algebra 
(the affine Hecke algebra of type C) into the one of one-boundary 
Temperley--Lieb algebra (the Hecke algebra of type B).

From the explicit expressions of $\Psi$, we can compute 
correlation functions exactly. 
One can also show
that $\Psi$ is in $\mathbb{N}[q,q^{-1},Q,Q^{-1}]$ for 
type A, BII and BIII and in $\mathbb{N}[q,q^{-1}]$ for
type BI.  
This positive integral property appears not only on the 
Kazhdan--Lusztig bases but also on the standard bases 
(see Definition~\ref{defPsi-SB}).
The transition matrix from the standard bases to the 
Kazhdan--Lusztig bases is written in terms of the 
Kazhdan--Lusztig polynomials which also have another 
positive integral structure.
Thus the origin of the positivity of $\Psi$ may come from 
these two positivities.
Since $\Psi$ has a positive integral property, 
it is natural to ask whether the components of $\Psi$ count 
combinatorial objects along the spirit of the Razumov--Stroganov 
correspondence. 
As a first step, we consider the sum of the components of $\Psi$.
For type A, BI and BIII, the sums are conjectured to be the 
total number of symmetric binary/permutation matrices with 
appropriate conditions.
In fact, some components are conjectured to be a 
$(q,Q)$-enumeration of symmetric binary matrices. 
We expect that these observations are a starting point 
of a connection of Kazhdan--Lusztig bases to enumerative 
combinatorics.

The paper is organized as follows.
In Section~\ref{sec-TL}, we briefly review the two-boundary 
Temperley--Lieb algebra and a coideal subalgebra of 
$U_{q}(\mathfrak{sl}_2)$. 
In Section~\ref{sec-reps}, we introduce a diagrammatic presentation
of the Kazhdan--Lusztig bases. 
We show the action of the two-boundary Temperley--Lieb algebra on
the Kazhdan--Lusztig bases of type A and B.
Section~\ref{sec-eigenX} is devoted to the analysis of the eigensystem
of the generator $X$. 
We define an eigenfunction $\Psi$ of $X$ and show that this 
eigenfunction has the multiplicity one. 
Section~\ref{sec-eigenH} is devoted to the analysis of the action of 
the Hamiltonian on $\Psi$. 
We show that the generators of one-boundary Temperley--Lieb algebra 
acts zero on $\Psi$ and that $\Psi$ is the ground state of the 
two-boundary Temperley--Lieb Hamiltonian with the integrable condition. 
In Section~\ref{sec-sum}, we compute correlation functions, show 
the positive integral property of $\Psi$ and propose several conjectures 
on $\Psi$ as a $(q,Q)$-enumeration of symmetric binary/permutation matrices. 
In Section~\ref{sec-appendix}, we collect technical lemmas 
used in this paper.

\section{Two-boundary Temperley--Lieb algebra}
\label{sec-TL}
\subsection{Two-boundary Temperley--Lieb algebra}
The Temperley--Leib algebra \cite{Mar90,Sal89,TL71} is an associative algebra 
over the ring $\mathbb{Z}[q,q^{-1}]$ and generated by $e_i, 1\le i\le N-1$
with the relations:
\begin{eqnarray}
\label{TL1}
&&e_i^2=-(q+q^{-1})e_i,\qquad 1\le i\le N-1, \\ 
&&e_ie_{i\pm1}e_i=e_i, \label{TLquotient}\\
&&e_ie_j=e_je_i,\qquad |i-j|>1.
\label{TL2}
\end{eqnarray}
The two-boundary Temperley--Lieb algebra \cite{deGNic09,deGPya04} is a generalization of 
the Temperley--Lieb algebra with extra generators $e_n$ and $e_0$. 
The defining relations are relations (\ref{TL1})-(\ref{TL2}) and 
\begin{eqnarray}
&&e_N^2=-(Q+Q^{-1})e_N, \\
&&e_{N-1}e_Ne_{N-1}=(qQ^{-1}+q^{-1}Q)e_{N-1}, \label{1BTLquotient}\\
&&e_ie_{N}=e_Ne_i, \qquad i\neq N-1. \\
&&e_0^2=-(Q_0+Q_0^{-1})e_N, \\
&&e_{1}e_0e_{1}=(qQ_{0}^{-1}+q^{-1}Q_0)e_{1}, \\
&&e_ie_{0}=e_0e_i, \qquad i\neq 1. 
\end{eqnarray}
We call the subalgebra generated by $\{e_i: 1\le i\le N\}$ the 
one-boundary Temperley--Lieb algebra.
Note that the two-boundary Temperley--Lieb algebra is infinite dimensional. 
We impose the following two conditions to make the algebra finite dimensional:
\begin{eqnarray}
\label{quotient2BTL}
I_{N}J_{N}I_{N}=\alpha I_{N}, \qquad J_{N}I_{N}J_{N}=\alpha J_{N}
\end{eqnarray}
where $\alpha$ is a parameter and 
\begin{eqnarray*}
&&I_{2n}:=\prod_{i=0}^{n-1}e_{2i+1}, \quad 
J_{2n}:=e_0\prod_{i=1}^{n-1}e_{2i}\cdot e_{N},\\
&&I_{2n+1}:=e_0\prod_{i=1}^{n}e_{2i}, \quad
J_{2n+1}:=\prod_{i=0}^{n-1}e_{2i+1}\cdot e_{N}. 
\end{eqnarray*}

Let $V_1$ be a two-dimensional $\mathbb{C}$-vector space spanned by
$v_1$ and $v_{-1}$.
We have a representation of the two-boundary Temperley--Lieb algebra 
acting on $V_1^{\otimes N}$.
The matrix representation of the generators are 
\begin{eqnarray*}
e_i&=&\underbrace{\mathbf{1}\otimes\cdots\otimes\mathbf{1}}_{i-1}\otimes
\begin{pmatrix}
0 & 0 & 0 & 0 \\
0 & -q^{-1} & 1 & 0 \\
0 & 1 & -q & 0 \\
0 & 0 & 0 & 0 
\end{pmatrix}
\otimes\underbrace{\mathbf{1}\otimes\cdots\otimes\mathbf{1}}_{N-i-1}, 
\qquad 1\le i\le N-1,\\
e_N&=&\underbrace{\mathbf{1}\otimes\cdots\otimes\mathbf{1}}_{N-1}\otimes
\begin{pmatrix}
-Q^{-1} & 1 \\
1 & -Q
\end{pmatrix}, \\
e_0&=&
\begin{pmatrix}
-Q_{0} & 1 \\
1 & -Q_{0}^{-1}
\end{pmatrix}\otimes
\underbrace{\mathbf{1}\otimes\cdots\otimes\mathbf{1}}_{N-1},
\end{eqnarray*}
where the order of bases is $(v_1,v_{-1})$ for $V_1$ and 
$(v_1\otimes v_1,v_1\otimes v_{-1},v_{-1}\otimes v_{1},v_{-1}\otimes v_{-1})$
for $V_1\otimes V_1$.
A tensor product $v_{\epsilon_1}\otimes\cdots\otimes v_{\epsilon_N}$ 
with $\epsilon_i=1$ or $-1$ for $1\le i\le N$ is called a standard
basis.
In this representation, one can show by a straightforward computation that  
\begin{eqnarray*}
\alpha=\begin{cases}
(Q^{-1}-q^{-1}Q_{0})(Q-qQ_{0}^{-1}), & \text{for } N: \text{even}, \\
(1+Q^{-1}Q_{0})(1+QQ_{0}^{-1}), & \text{for $N$: odd}
\end{cases}
\end{eqnarray*}

We consider the integrable Hamiltonian for the one- and two-boundary 
Temperley--Lieb algebras acting on $V_{1}^{\otimes N}$ 
defined by 
\begin{eqnarray*}
H^{1B}&=&-\sum_{i=1}^{N-1}e_i-a_{N}e_{N}, \\
H^{2B}&=&H^{1B}-a_{0}e_{0},
\end{eqnarray*}
where $a_{0}$ and $a_{N}$ are parameters.

\begin{remark}
The Temperley--Lieb Hamiltonian is rewritten as 
\begin{multline*}
H^{2B}=-\frac{1}{2}\left(
\sum_{i=1}^{N-1}
\left(\sigma^{x}_{i}\sigma^{x}_{i+1}
+\sigma^{y}_{i}\sigma^{y}_{i+1}
+\frac{q+q^{-1}}{2}\sigma^{z}_{i}\sigma^{z}_{i+1}\right)
+\left(\frac{q-q^{-1}}{2}-a_0(Q_0-Q_{0}^{-1})\right)\sigma^{z}_{1}\right. \\
\left.+2a_{0}(\sigma^{+}_{1}+\sigma^{-}_{1})+2a_{N}(\sigma^{+}_{N}+\sigma^{-}_{N})
-\left(\frac{q-q^{-1}}{2}-a_{N}(Q-Q^{-1})\right)
\sigma^{z}_{N}
\right) \\
+\left(\frac{q+q^{-1}}{4}(N-1)+a_{0}\frac{Q_0+Q_0^{-1}}{2}+a_{N}\frac{Q+Q^{-1}}{2}\right).
\end{multline*}
where $\sigma^{x}, \sigma^{y}$ and $\sigma^z$ are the Pauli matrices 
and $\sigma^{\pm}=(\sigma^{x}\pm\sqrt{-1}\sigma^{y})/2$. 
Thus the spectrum of $H$ can be viewed as the one of the $XXZ$ 
spin $1/2$ quantum chain with boundaries.
\end{remark}

\subsection{\texorpdfstring{A coideal subalgebra of $U_q(\mathfrak{sl}_2)$}
{A coideal subalgebra of Uqsl2}}
The quantum group ${\bf U}:=U_q(\mathfrak{sl}_2)$ is 
an associative algebra over $\mathbb{C}(q)$ with generators $E,F,K^{\pm1}$ 
and relations 
\begin{eqnarray*}
&&KK^{-1}=K^{-1}K=1, \\
&&KEK^{-1}=q^{2}E, \\
&&KFK^{-1}=q^{-2}F, \\
&&EF-FE=\frac{K-K^{-1}}{q-q^{-1}}.
\end{eqnarray*}
We introduce  the quantum integer $[n]:=\sum_{i=0}^{n-1}q^{n-1-2i}$, 
the quantum factorial $[n]!:=\prod_{i=1}^{n}[i]$ and 
$q$-analogue of the binomial coefficient 
\begin{eqnarray*}
\genfrac{[}{]}{0pt}{}{n}{m}:=\frac{[n]!}{[n-m]![m]!}.
\end{eqnarray*}
We define 
\begin{eqnarray*}
[Q;n]:=\frac{Qq^{n}-Q^{-1}q^{-n}}{q-q^{-1}}.
\end{eqnarray*}

The comultiplication $\Delta$ is given by 
\begin{eqnarray*}
\Delta(K^{\pm1})&:=&K^{\pm1}\otimes K^{\pm1}, \\
\Delta(E)&=&E\otimes K^{-1}+1\otimes E, \\
\Delta(F)&=&F\otimes 1+K\otimes F.
\end{eqnarray*}
We consider the two-dimensional representation in $V_1$. 
The action is given by 
\begin{eqnarray*}
&&Ev_1=0, \qquad Ev_{-1}=v_{1}, \\
&&Fv_{1}=v_{-1},\qquad Fv_{-1}=0, \\
&&Kv_{\pm1}=q^{\pm1}v_{\pm1}.
\end{eqnarray*}

We consider the Dynkin diagram of type $A_1$ and the identity 
involution. 
By a general theory of quantum symmetric space \cite{Kolb14,Letz99,Letz03}, one can 
obtain a coideal algebra of ${\bf U}$ associated with the involution. 
The coideal subalgebra ${\bf U'}$ of $U_q(\mathfrak{sl}_2)$ is a polynomial 
algebra in $X$, that is, ${\bf U'}:=\mathbb{C}(q)[X]$.
The injective $\mathbb{C}(q)$-algebra homomorphism 
$\iota:{\bf U'}\rightarrow {\bf U}$ is given by 
\begin{eqnarray*}
X\mapsto F+cKE+sK
\end{eqnarray*}
where $c,s$ are indeterminates. 
The comultiplication $\Delta$ is given by 
\begin{eqnarray}
\label{coproduct}
\Delta(X)=K\otimes X +cKE\otimes 1 +F\otimes 1.
\end{eqnarray}
Note that ${\bf U'}$ is left coideal since 
$\Delta(X)\subset{\bf U}\times{\bf U'}$.
In this paper, we consider 
\begin{eqnarray*}
c=q^{-1}, \qquad s=\frac{Q-Q^{-1}}{q-q^{-1}}.
\end{eqnarray*}

\begin{theorem}
\label{theorem-commute}
$[H^{1B},X]=0$.
\end{theorem}
\begin{proof}
It is enough to show that $[e_i,X]=0$ for all $1\le i\le N$.
When $N=1$, we have $[e_{1},X]=0$ by a straightforward calculation.
Since the comultiplication is given by Eqn.(\ref{coproduct}), 
we have $[e_N,X]=0$ in general.
Since the action of the Temperley--Lieb algebra commutes with 
the action of the quantum group $U_q(\mathfrak{sl}_2)$ in the 
tensor product of the fundamental representation, we have 
$[e_i,X]=0$ for $1\le i\le N-1$.
\end{proof}

\section{Representations}
\label{sec-reps}
The actions of the two-boundary Temperley--Lieb algebra on
the standard bases are obvious through the matrix representation 
of the generators (See Section~\ref{sec-TL}).
In this section, we consider the action of the two-boundary 
Temperley--Lieb algebra on Kazhdan--Lusztig bases. 

\subsection{Kazhdan--Lusztig bases}
\label{sec-Rep-KL}
The Hecke algebra of type A is a unital, associative algebra over 
$\mathbb{C}[q,q^{-1}]$ generated by the generators $T_i$, 
$1\le i\le N-1$, satisfying the relations $(T_i-q^{-1})(T_{i}+q)=0$, 
$T_{i}T_{i+1}T_{i}=T_{i+1}T_{i}T_{i+1}$
and $T_iT_j=T_jT_i$ for $|i-j|>1$.
The Temperley--Lieb algebra can be regarded as the Hecke algebra of 
type A with a quotient relation (\ref{TLquotient}) through 
a relation $T_i=e_i+q^{-1}$. 
The representation of the Temperley--Lieb algebra in $V_1^{\otimes N}$ 
corresponds to the maximal parabolically induced representation of 
the Hecke algebra.
The Hecke algebra of type B is generated by $T_{i}$, $1\le i\le N$, with 
the relations of type A, $(T_N-Q^{-1})(T_N+Q)=0$, 
$T_NT_{N-1}T_NT_{N-1}=T_{N-1}T_NT_{N-1}T_{N}$ and 
$T_iT_N=T_{N}T_{i}$ for $i\neq N-1$.
Similarly, the one-boundary Temperley--Lieb algebra can be regarded 
as the Hecke algebra of type B with a quotient relation 
(\ref{1BTLquotient}).
Therefore, one can apply the representation theory of the Hecke algebra 
of type B to the one-boundary Temperley--Lieb algebra.
Since we consider the representation of the one-boundary 
Temperley--Lieb algebra in $V_{1}^{\otimes N}$ as in Section~\ref{sec-TL}, 
the parabolic Kazhdan--Lusztig bases studied 
in~\cite{Deo87} play a central role rather than original ones studied in \cite{KL79}.
More precisely, the Kazhan--Ludztig bases for the Hermitian symmetric 
pair $(B_N,A_{N-1})$ studied in \cite{Boe88,Bre09,Shi14-1} can be regarded as 
bases of the one-boundary Temperley--Lieb algebra in $V_1^{\otimes N}$.
There are two types of parabolic Kazhdan--Lusztig bases according 
to the choice of a projection map (see, {\it e.g.}, Section 2.3 in \cite{Shi14-1}). 
In this paper, we will consider the parabolic Kazhdan--Lusztig bases 
studied as $C^-_{x}$ in~\cite{Shi14-1}.
Hereafter, a Kazhdan--Lusztig basis means this parabolic one.

We have four types of Kazhdan--Lusztig bases associated with 
the Temperley--Lieb algebra and the one-boundary Temperley--Lieb 
algebras.
The first one is the Kazhdan--Lusztig basis of Temperley--Lieb 
algebra of type A, 
the second is the Kazhdan--Lusztig basis of the one-boundary 
Temperley--Lieb algebra for $Q=q^{M}$ with 
$M\in\mathbb{N}_{+}$ and the third and the fourth are the Kazhdan--Lusztig  
bases of the one-boundary Temperley--Lieb algebra for $q\neq Q$
where $q$ and $Q$ are algebraically independent.
Note that we will consider the representation of the two-boundary 
Temperley--Lieb algebra on Kazhdan--Lusztig bases of Type A 
in the first case and of type B in the second, third and fourth cases.
The difference between the third and the fourth bases is 
the total order with respect to $q$ and $Q$.
We call these bases type A, BI, BII and BIII respectively.
We index a Kazhdan--Lusztig basis by a binary string $\{+,-\}^{N}$.  
Given two binary strings $\epsilon:=\epsilon_1\ldots\epsilon_N$ 
and $\epsilon':=\epsilon'_{1}\ldots\epsilon'_N$, 
we denote $\epsilon<\epsilon'$ if $\epsilon_j=\epsilon'_j$ for 
$1\le j\le i-1$ and $\epsilon_i<\epsilon'_{i}$. 
All types of Kazhdan--Lusztig bases are characterized by the following 
two conditions \cite{Deo87,KL79,Lus03}: 
(1) a Kazhdan--Lusztig basis is invariant under the involutive ring automorphism 
known as ``bar involution" where $T_i\rightarrow T_{i}^{-1}$,  
$q\rightarrow q^{-1}$ and $Q\rightarrow Q^{-1}$. 
On the module $V^{\otimes N}$, we define $\overline{v_{\epsilon}}=v_{\epsilon}$ 
where $\epsilon_{i}=1$ for $1\le i\le N$.
(2) The expansion of a Kazhdan--Lusztig basis $w$ indexed by a binary string 
$\epsilon=\epsilon_1\ldots\epsilon_N$ 
in terms of standard basis has the leading term 
$v_{\kappa_1}\otimes\cdots\otimes v_{\kappa_N}$ where $\kappa_i=1$ if 
$\epsilon_i=+$ and $\kappa_i=-1$ if $\epsilon_i=-$.
The vector $w-v_{\kappa_1}\otimes\cdots\otimes v_{\kappa_N}$ is a linear 
combination of $v_{\kappa'_1}\otimes\cdots\otimes v_{\kappa'_N}$, $\kappa<\kappa'$, 
with a coefficient in $\mathbb{Z}(\Gamma^{\mathrm{X}}_-)$ for Type X. 
Here $\Gamma^{A}_-=\Gamma^{\mathrm{BI}}_-=\{q^{-i}|i\in\mathbb{N}_{+}\}$,   
$\Gamma^{BII}_-=\{q^{-i}Q^{j}|i\in\mathbb{N}_+, j\in\mathbb{Z}\}\cup
\{Q^{-i}|i\in\mathbb{N}_+\}$ 
and $\Gamma^{BIII}_{-}=\{q^{i}Q^{-j}|i\in\mathbb{Z}, j\in\mathbb{N}_{+} \} \cup
\{q^{-i}|i\in\mathbb{N}_{+}\}$.

Since $V_1^{\otimes N}$ can be viewed as the tensor products of fundamental 
representation of $U_q(\mathfrak{sl}_2)$, a Kazhdan--Lusztig basis of type 
A is nothing but the dual canonical basis of $U_q(\mathfrak{sl}_2)$ considered 
in~\cite{FK97} (see also \cite{ShiZJ12}).
A Kazhdan--Lusztig basis of type BI is considered in \cite{Shi14-1} and can be 
viewed as the dual canonical basis of a coideal subalgebra of 
$U_q(\mathfrak{sl}_2)$~\cite{Shi14-2}. 
A Kazhdan--Lusztig basis of type BII is studied in \cite{Shi14-1}.
One can easily show that a basis of type BIII satisfies the criteria for
a Kazhdan--Lusztig basis.

We briefly review the graphical presentation of a Kazhdan--Lusztig basis 
following~\cite{FK97,Shi14-1,ShiZJ12}. 
Let $b=b_1\cdots b_N\in\{\pm\}^{N}$ be a binary string.
We place an up arrow (resp. a down arrow) from left to right according 
to $b_i=+$ (resp. $b_i=-$).
We have the following two rules.
\begin{enumerate}[(A)]
\item We make a pair between adjacent down arrow and up arrow in this order.
Then connect this pair into a simple arc. 
\item Repeat the procedure (A) until all the up arrows are to the left of all 
down arrows.
\end{enumerate}

\paragraph{\bf Type A}
The Kazhdan--Lusztig basis of Type A follows rules (A) and (B).

\paragraph{\bf Type BI}
In addition to rules (A) and (B), we have three more rules:
\begin{enumerate}
\item[(C)] Put an integer $p, 2\le p\le M$, on the $(M+1-p)$-th down arrow from right.
\item[(D)] Put a star ($\bigstar$) on the $M$-th down arrow from right if it exists.
\item[(E)] For remaining down arrows, we make a pair of adjacent down arrows from 
right to left. Then connect this pair into a simple dashed arc.
\end{enumerate}
After applying rules (A)-(E), we may have an unpaired down arrow which does not 
form a dashed arc.

\paragraph{\bf Type BII}
After applying rules (A) and (B) to a diagram, we have unpaired up arrows 
and unpaired down arrows.
We call the $(2i-1)$-th (resp. $2i$-th) unpaired down arrow from right 
an o-unpaired (resp. e-unpaired) down arrow. 
We have an additional rule:
\begin{enumerate}
\item[(F)] We put a vertical line with a mark e (resp. o) on an e-unpaired 
(resp. o-unpaired) down arrow.
\end{enumerate}

\paragraph{\bf Type BIII}
We apply rules (A) and (B) to a diagram. 
We have unpaired up arrows and unpaired down arrows.
We enumerate unpaired down arrows from right to left by $1,2,\ldots$.
Then, we have an additional rule:
\begin{enumerate}
\item[(G)] We put a vertical line with a circled integer $i$ on 
the $i$-th unpaired down arrow.
\end{enumerate}

A diagram corresponds to a vector in $V_1\otimes\cdots\otimes V_1$ as follows.
An unpaired up arrow (resp. down arrow) in a diagram is
$v_1$ (resp. $v_{-1}$) in a tensor product.
Each building block (a simple arc, a dashed arc, a down arrow with a star, mark e
or mark o and a down arrow with an integer or with a circled integer) 
is a vector in $V_1$ or $V_1\otimes V_1$:
\begin{eqnarray*}
\raisebox{-0.6\totalheight}{
\begin{tikzpicture}
\draw (0,0)..controls(0,-1)and(1,-1)..(1,0);
\end{tikzpicture}
}
&=&v_{-1}\otimes v_{1}-q^{-1}v_1\otimes v_{-1},\\
\raisebox{-0.6\totalheight}{
\begin{tikzpicture}
\draw[dashed,thick] (0,0)..controls(0,-1)and(1,-1)..(1,0);
\end{tikzpicture}
}
&=&v_{-1}\otimes v_{-1}-q^{-1}v_1\otimes v_{1},\\
\raisebox{-0.6\totalheight}{
\begin{tikzpicture}
\draw (0,0)--(0,-0.8)node{$\bigstar$};
\end{tikzpicture}}
&=&v_{-1}-q^{-1}v_{1}, \\
\raisebox{-0.6\totalheight}{
\begin{tikzpicture}
\draw (0,0)--(0,-0.6)(0,-0.8)node{$p$};
\end{tikzpicture}}
&=&v_{-1}-q^{-p}v_{1}, \quad \text{for }2\le p\le M, \\
\raisebox{-0.6\totalheight}{
\begin{tikzpicture}
\draw (0,0)--(0,-0.7)(0,-0.8)node{o};
\end{tikzpicture}}
&=&v_{-1}-Q^{-1}v_{1}, \\
\raisebox{-0.6\totalheight}{
\begin{tikzpicture}
\draw (0,0)--(0,-0.7)(0,-0.8)node{e};
\end{tikzpicture}}
&=&v_{-1}+q^{-1}Qv_{1}. \\
\raisebox{-0.6\totalheight}{
\begin{tikzpicture}
\draw(0,0)--(0,-0.7)
(0,-0.7)node[circle,inner sep=1pt,draw,anchor=north]{$p$};
\end{tikzpicture}}
&=&v_{-1}-q^{p-1}Q^{-1}v_{1}.
\end{eqnarray*}
An unpaired up (resp. down) arrow corresponds to $v_{1}$ (resp. $v_{-1}$).
A vector in $V^{\otimes N}$ corresponding to a diagram 
is given by a tensor product of a vector corresponding to 
a building block.

\begin{example}
Let $b=+--+--$. 
The Kazhdan--Lusztig basis indexed by $b$ are 
\begin{eqnarray*}
\resizebox{0.14\hsize}{!}{$
\uparrow\downarrow\!\! 
\raisebox{-0.45\height}{
\begin{tikzpicture}
\draw(0,0)..controls(0,-0.45)and(0.45,-0.45)..(0.45,-0);
\end{tikzpicture}}
\downarrow\downarrow
$},\qquad
\resizebox{0.15\hsize}{!}{$
\uparrow\!\! 
\raisebox{-0.65\height}{
\begin{tikzpicture}
\draw(0,0)..controls(0,-0.45)and(0.45,-0.45)..(0.45,-0);
\draw[dashed](-0.225,0)..controls(-0.225,-0.7)and(0.675,-0.7)..(0.675,0);
\end{tikzpicture}}
$}\! 
\raisebox{-0.5\height}{
\begin{tikzpicture}
\draw[thick](0,0)--(0,-0.7)node{${\bigstar}$};
\end{tikzpicture}}
,\qquad
\resizebox{0.021\hsize}{!}{$
\uparrow$}
\raisebox{-0.45\height}{
\begin{tikzpicture}
\draw(0,0)--(0,-0.6)(0,-0.7)node{$\mathrm{o}$};
\end{tikzpicture}}
\raisebox{-0.3\height}{
\begin{tikzpicture}
\draw(0,0)..controls(0,-0.6)and(0.6,-0.6)..(0.6,-0);
\end{tikzpicture}}
\raisebox{-0.45\height}{
\begin{tikzpicture}
\draw(0,0)--(0,-0.6)(0,-0.7)node{$\mathrm{e}$};
\end{tikzpicture}}
\raisebox{-0.45\height}{
\begin{tikzpicture}
\draw(0,0)--(0,-0.6)(0,-0.7)node{$\mathrm{o}$};
\end{tikzpicture}}, \qquad
\tikzpic{-0.45}{
\linkpattern{1.2/2}{}{}{0.8/3,2.3/2,2.9/1};
\upa{0};
}
\end{eqnarray*}
for Type A, BI ($M=1$), BII and BIII respectively. 
The diagram of type A corresponds to 
a vector 
\begin{eqnarray*}
v_{1}\otimes v_{-1}\otimes v_{-1}\otimes v_{1}\otimes v_{-1}
\otimes v_{-1}-q^{-1}v_{1}\otimes v_{-1}\otimes v_{1}\otimes v_{-1}\otimes v_{-1}
\otimes v_{-1}
\end{eqnarray*}
in $V_{1}^{\otimes6}$.
\end{example}

\subsection{Action of the Temperley--Lieb algebra on Kazhdan--Lusztig bases}
Recall that a Kazhdan--Lusztig basis is expressed as a tensor product of building 
blocks. 
Thus it is enough to consider the action of Temperley--Lieb 
algebra on partial diagrams. 
We list up all the partial diagrams for the action of $e_i$, 
$1\le i\le N-1$.
\begin{eqnarray*}
&&e_i(\downarrow\downarrow)=0, \qquad
e_i\left(
\raisebox{-0.5\totalheight}{
\begin{tikzpicture}
\draw[dashed](0,0)..controls(0,-0.6)and(0.6,-0.6)..(0.6,0);
\draw(0,0)node[anchor=south]{i}(0.6,0)node[anchor=south]{i+1};
\end{tikzpicture}}\ 
\right)=0, \qquad
e_i\left(
\tikzpic{-0.5}{
\draw(0,0.2)node[anchor=north]{$\downarrow$};
\linkpattern{}{0.6/1.4}{}{};
\draw(0,0)node[anchor=south]{i}
(0.6,0)node[anchor=south]{i+1};
}\right)
=\tikzpic{-0.5}{
\linkpattern{0/0.8}{}{}{};
}\uparrow, \\
&&e_i\left(
\tikzpic{-0.5}{
\draw(0,0.2)node[anchor=north]{$\downarrow$};
\linkpattern{0.6/1.4}{}{}{};
\draw(0,0)node[anchor=south]{i}
(0.6,0)node[anchor=south]{i+1};
}\right)
=\tikzpic{-0.5}{
\linkpattern{0/0.8}{}{}{};
}\uparrow, \qquad
e_i\left(
\tikzpic{-0.5}{
\draw(0,0.2)node[anchor=north]{$\downarrow$};
\linkpattern{}{}{0.6/$\bigstar$}{};
\draw(0,0)node[anchor=south]{i}
(0.6,0)node[anchor=south]{i+1};
}\right)
=\tikzpic{-0.5}{
\linkpattern{0/0.8}{}{}{};
},        \\
&&e_i\left(
\raisebox{-0.5\totalheight}{
\begin{tikzpicture}
\draw[dashed](0,0)..controls(0,-0.6)and(0.6,-0.6)..(0.6,0);
\draw[dashed](1.0,0)..controls(1.0,-0.6)and(1.6,-0.6)..(1.6,0);
\draw(0.6,0)node[anchor=south]{i}(0.7,0.25)node[anchor=west]{i+1};
\end{tikzpicture}}\ 
\right)
=
\raisebox{-0.6\totalheight}{
\begin{tikzpicture}
\draw(0,0)..controls(0,-1.2)and(1.6,-1.2)..(1.6,0);
\draw(1.6/3,0)..controls(1.6/3,-0.6)and(3.2/3,-0.6)..(3.2/3,0);
\end{tikzpicture}}, \qquad
e_i\left(
\raisebox{-0.5\totalheight}{
\begin{tikzpicture}
\draw[dashed](0,0)..controls(0,-0.6)and(0.6,-0.6)..(0.6,0);
\draw(1.0,0)..controls(1.0,-0.6)and(1.6,-0.6)..(1.6,0);
\draw(0.6,0)node[anchor=south]{i}(0.7,0.25)node[anchor=west]{i+1};
\end{tikzpicture}}\ 
\right)=
\raisebox{-0.6\totalheight}{
\begin{tikzpicture}
\draw[dashed](0,0)..controls(0,-1.2)and(1.6,-1.2)..(1.6,0);
\draw(1.6/3,0)..controls(1.6/3,-0.6)and(3.2/3,-0.6)..(3.2/3,0);
\end{tikzpicture}}, \\
&&e_i\left(
\raisebox{-0.45\totalheight}{
\begin{tikzpicture}
\draw[dashed](0,0)..controls(0,-1.2)and(1.6,-1.2)..(1.6,0);
\draw(1.6/3,0)..controls(1.6/3,-0.6)and(3.2/3,-0.6)..(3.2/3,0);
\draw(0,0)node[anchor=south]{i}(1.6/3,0)node[anchor=south]{i+1};
\end{tikzpicture}}
\right)
=\raisebox{-0.5\totalheight}{
\begin{tikzpicture}
\draw(0,0)..controls(0,-0.6)and(0.6,-0.6)..(0.6,0);
\draw[dashed](1.0,0)..controls(1.0,-0.6)and(1.6,-0.6)..(1.6,0);
\end{tikzpicture}}, \qquad
e_i\left(
\raisebox{-0.5\totalheight}{
\begin{tikzpicture}
\linkpattern{}{1/2}{2.8/$\bigstar$}{};
\draw(1.6,0)node[anchor=south]{i}(2.2,0)node[anchor=south]{i+1};
\end{tikzpicture}}
\right)
=
\tikzpic{-0.5}{
\linkpattern{1.5/2.7}{}{1/$\bigstar$}{};
},\\
&&e_i\left(
\tikzpic{-0.5}{
\linkpattern{}{}{0/p,0.8/{p+1}}{};
}\right)
=0, \quad
e_i\left(
\tikzpic{-0.5}{
\linkpattern{0.6/1.2}{}{0/p}{};
\draw(0,0)node[anchor=south]{i}
(0.55,0)node[anchor=south]{i+1};
}\right)
=\tikzpic{-0.5}{
\linkpattern{0/0.8}{}{1.2/p}{};
},\\
&&e_i\left(
\tikzpic{-0.5}{
\linkpattern{0.6/1.2}{}{}{0/p};
\draw(0,0)node[anchor=south]{i}
(0.55,0)node[anchor=south]{i+1};
}\right)
=\tikzpic{-0.5}{
\linkpattern{0/0.8}{}{}{1.2/p};
},\quad
e_i\left(
\tikzpic{-0.5}{
\linkpattern{}{}{}{0/{p+1},0.8/{p}};
}\right)=0, \quad
e_i\left(
\tikzpic{-0.5}{
\linkpattern{0.6/1.4}{}{0/e}{};
\draw(0,0)node[anchor=south]{i}
(0.6,0)node[anchor=south]{i+1};
}
\right)
=\tikzpic{-0.5}{
\linkpattern{0/0.8}{}{1.2/e}{};
}, \\
&&e_i\left(
\tikzpic{-0.5}{
\linkpattern{0.6/1.4}{}{0/o}{};
\draw(0,0)node[anchor=south]{i}
(0.6,0)node[anchor=south]{i+1};
}
\right)
=\tikzpic{-0.5}{
\linkpattern{0/0.8}{}{1.2/o}{};
}, \quad
e_i\left(
\tikzpic{-0.5}{
\linkpattern{}{}{0/e,0.8/o}{};
}\right)
=(q^{-1}Q+qQ^{-1})
\tikzpic{-0.5}{
\linkpattern{0/0.8}{}{}{};
}, \\
&&e_i\left(
\tikzpic{-0.5}{
\linkpattern{}{}{0/o,0.8/e}{};
}\right)
=-(Q+Q^{-1})
\tikzpic{-0.5}{
\linkpattern{0/0.8}{}{}{};
}, \quad
e_i\left(
\tikzpic{-0.5}{
\linkpattern{0/0.8}{}{}{};
}\right)=
-[2]\tikzpic{-0.5}{
\linkpattern{0/0.8}{}{}{};
},\quad
e_i\left(
\tikzpic{-0.5}{
\draw(0,0.2)node[anchor=north]{$\uparrow$};
\linkpattern{}{0.6/1.4}{}{};
\draw(0,0)node[anchor=south]{i}
(0.6,0)node[anchor=south]{i+1};
}\right)
=\tikzpic{-0.5}{
\linkpattern{0/0.8}{}{}{};
}\downarrow, \\
&&e_i\left(
\tikzpic{-0.5}{
\linkpattern{0/0.8}{1.2/2}{}{};
\draw(0.8*0.8,0)node[anchor=south]{i}
(1.3*0.8,0)node[anchor=south]{i+1};
}\right)
=\tikzpic{-0.7}{
\linkpattern{0.8/1.6}{0/2.4}{}{};
}, \quad
e_i\left(
\tikzpic{-0.5}{
\draw(0,0.1)node[anchor=north]{$\uparrow$};
\linkpattern{0.7/1.5}{}{}{};
\draw(0,0)node[anchor=south]{i}
(0.7*0.8,0)node[anchor=south]{i+1};
}
\right)
=\tikzpic{-0.5}{
\linkpattern{0/0.8}{}{}{};
}\uparrow, \\
&&e_i\left(
\tikzpic{-0.5}{
\linkpattern{0.8/1.6}{0/2.4}{}{};
\draw(1.6*0.8,0)node[anchor=south]{i}
(2.4*0.8,0)node[anchor=south]{i+1};
}\right)
=\tikzpic{-0.5}{
\linkpattern{1.4/2.2}{0/0.8}{}{};
}, \quad
e_i\left(
\tikzpic{-0.5}{
\linkpattern{0/0.8}{}{1.2/p}{};
\draw(0.8*0.8,0)node[anchor=south]{i}
(1.3*0.8,0)node[anchor=south]{i+1};
}\right)
=\tikzpic{-0.5}{
\linkpattern{0.6/1.4}{}{0/p}{};
},\\
&&e_i\left(
\tikzpic{-0.5}{
\linkpattern{0/0.8}{}{}{1.2/p};
\draw(0.8*0.8,0)node[anchor=south]{i}
(1.3*0.8,0)node[anchor=south]{i+1};
}\right)
=\tikzpic{-0.5}{
\linkpattern{0.6/1.4}{}{}{0/p};
},\quad
e_i\left(\uparrow
\tikzpic{-0.7}{
\linkpattern{}{}{0.2/p}{};
}\right)
=\tikzpic{-0.5}{
\linkpattern{0/0.8}{}{}{};
},\quad 
e_i\left(\uparrow
\tikzpic{-0.7}{
\linkpattern{}{}{}{0.2/p};
}\right)
=\tikzpic{-0.5}{
\linkpattern{0/0.8}{}{}{};
}, \\
&&e_i\left(\uparrow \!
\tikzpic{-0.7}{
\linkpattern{}{}{0/$\bigstar$}{};
}\right)
=\tikzpic{-0.5}{
\linkpattern{0/0.8}{}{}{};
}, \quad
e_i\left(\uparrow
\tikzpic{-0.7}{
\linkpattern{}{}{0.2/e}{};
}\right)
=\tikzpic{-0.5}{
\linkpattern{0/0.8}{}{}{};
},  \quad
e_i\left(\uparrow
\tikzpic{-0.7}{
\linkpattern{}{}{0.2/o}{};
}\right)
=\tikzpic{-0.5}{
\linkpattern{0/0.8}{}{}{};
}, \\
&&e_i\left(
\tikzpic{-0.5}{
\linkpattern{0/0.8}{}{1.3/e}{};
\draw(0.64,0)node[anchor=south]{i}
(1.3*0.8,0)node[anchor=south]{i+1};
}\right)
=\tikzpic{-0.5}{
\linkpattern{0.6/1.4}{}{0/e}{};
}, \quad
e_i\left(
\tikzpic{-0.5}{
\linkpattern{0/0.8}{}{1.3/o}{};
\draw(0.8*0.8,0)node[anchor=south]{i}
(1.3*0.8,0)node[anchor=south]{i+1};
}\right)
=\tikzpic{-0.5}{
\linkpattern{0.6/1.4}{}{0/o}{};
}, \\
&&e_i(\uparrow\downarrow)=
\tikzpic{-0.5}{\linkpattern{0/0.8}{}{}{};}, \quad
e_i\left(
\tikzpic{-0.5}{
\linkpattern{0/0.8}{}{}{};
\draw(0.8*0.8,0)node[anchor=south]{i}
(1.2,0)node[anchor=south]{i+1};
\draw(1.2,-0.2)node{$\downarrow$};
}\right)
=\downarrow\tikzpic{-0.5}{
\linkpattern{0/0.8}{}{}{};
}, \quad
e_i\left(
\tikzpic{-0.5}{
\linkpattern{0/0.8}{}{}{};
\draw(0.64,0)node[anchor=south]{i}
(1.2,0)node[anchor=south]{i+1};
\draw(1.2,-0.2)node{$\uparrow$};
}\right)
=\uparrow\tikzpic{-0.5}{
\linkpattern{0/0.8}{}{}{};
}, \\
&&e_i\left(
\tikzpic{-0.5}{
\linkpattern{0/2.1,0.7/1.4}{}{}{};
\draw(0,0)node[anchor=south]{i}
(0.7*0.8,0)node[anchor=south]{i+1};
}
\right)
=\tikzpic{-0.5}{
\linkpattern{0/0.7,1.0/1.7}{}{}{};
},\qquad
e_i\left(
\tikzpic{-0.5}{
\linkpattern{0/2.1,0.7/1.4}{}{}{};
\draw(1.4*0.8,0)node[anchor=south]{i}
(2.1*0.8,0)node[anchor=south]{i+1};
}
\right)
=\tikzpic{-0.5}{
\linkpattern{0/0.7,1.0/1.7}{}{}{};
},\qquad
e_i(\uparrow\uparrow)=0.
\end{eqnarray*}
\begin{example}
Let $D$ be a diagram of type BI ($M=2$) depicted as 
\begin{eqnarray*}
D=
\tikzpic{-0.6}{
\upa{0};
\linkpattern{0.9/1.7,2.9/3.7}{0.5/2.1}{2.5/$\bigstar$,4.1/2}{};
}.
\end{eqnarray*}
Then, we have 
\begin{eqnarray*}
&&e_{1}(D)=
\tikzpic{-0.6}{
\downa{1.7};	
\linkpattern{-0.3/0.5,0.9/1.7,3.1/3.9}{}{2.7/$\bigstar$,4.3/2}{};
}, \qquad
e_2(D)
=\tikzpic{-0.6}{
\upa{0};
\linkpattern{0.5/1.3,3.3/4.1}{1.7/2.5}{2.9/$\bigstar$,4.4/2}{};
}, \\
&&e_{5}(D)
=\tikzpic{-0.6}{
\upa{0};
\linkpattern{0.9/1.7,2.1/2.9,3.3/4.1}{}{0.5/$\bigstar$,4.5/2}{};
},\qquad
e_{6}(D)
=\tikzpic{-0.6}{
\upa{0};
\linkpattern{0.9/1.7,2.5/3.3}{0.5/2.1}{3.7/$\bigstar$,4.1/2}{};
}
\end{eqnarray*}
and $e_3(D)=e_{7}(D)=-[2]D$.
\end{example}

\subsection{\texorpdfstring{Action of $e_N$ on Kazhdan--Lusztig bases}
{Action of eN on KL bases}}
The action of $e_N$ on a Kazhdan--Lusztig basis 
is given as follows.
\paragraph{\bf Type A}
Let $D$ be a diagram of type A and $N_{\uparrow}$ 
(resp. $N_\downarrow$) be the number of up (resp. down)
arrows.
We have three cases for $D$: 1) the rightmost arrow 
of $D$ is an up arrow, {\it i.e.}, $N_\downarrow=0$, 
2) the rightmost arrow is a down arrow and 
3) the rightmost arrow is an up  
arrow forming an arc.

\paragraph{Case 1}
Let $D'$ be a diagram obtained from $D$ by changing 
the rightmost up arrow to a down arrow. 
Then, the action of $e_N$ on $D$ is given by
\begin{eqnarray*}
e_N(D)=D'-Q^{-1}D.
\end{eqnarray*}

\paragraph{Case 2}
We enumerate down arrows in $D$ from right to left by 
$1,2,\ldots,N_{\downarrow}$.
For each $i$, $1\le i\le N_{\downarrow}-1$, we denote 
by $A_{(i)}(D)$ a diagram obtained from $D$ by connecting 
the $i$-th and the $(i+1)$-th down arrows via an arc.
We denote by $A_{(N_{\downarrow})}(D)$ a diagram obtained 
from $D$ by changing the $N_{\downarrow}$-th arrow 
to an up arrow.
The action of $e_N$ on $D$ is given by
\begin{eqnarray*}
e_N(D)=\sum_{1\le i\le N_{\downarrow}}q^{-(i-1)}A_{(i)}(D)
-QD.
\end{eqnarray*}

\paragraph{Case 3}
Let $\widetilde{D}$ be a diagram obtained from $D$ by changing
the rightmost arc to two down arrows.
We enumerate unpaired down arrows of $\widetilde{D}$ from right to left 
by $1,2,\ldots N_{\downarrow}+2$.
For each $i$, $2\le i\le N_{\downarrow}+1$, we denote by 
$A_{(i)}(\widetilde{D})$ obtained from $\widetilde{D}$ by 
connecting the $i$-th and the $(i+1)$-th down arrows via an arc.
We denote by $A_{(N_{\uparrow}+2)}(\widetilde{D})$ a diagram 
obtained from $\widetilde{D}$ by changing the $(N_{\uparrow}+2)$-th 
down arrow to an up arrow.
For each $1\le i<j\le N_{\downarrow}+1$ with $j-i\ge2$, we denote by 
$B_{(j,i)}(\widetilde{D})$ a diagram obtained from $\widetilde{D}$
by connecting the $i$-th and the $(i+1)$-th down arrows via an arc 
and the $j$-th and the $(j+1)$-th arrows via an arc.
For each $i$, $1\le i\le N_{\downarrow}-1$, we denote by 
$B_{(i+1,i)}(\widetilde{D})$ a diagram obtained from $\widetilde{D}$ 
by connecting the $i$-th and the $(i+3)$-th down arrows via an arc 
and the $(i+1)$-th and the $(i+2)$-th down arrows via an arc.
We denote by $B_{(N_{\downarrow}+1,N_{\downarrow})}(\widetilde{D})$ 
a diagram obtained from $\widetilde{D}$ by connecting the 
$(N_{\downarrow}+2)$-th and the $(N_{\downarrow}+1)$-th arrows 
via an arc and putting an up arrow at the $N_{\downarrow}$-th 
site.
For each $i$, $1\le i\le N_{\downarrow}$, we denote 
by $B_{(N_{\downarrow}+2,i)}$ the diagram obtained from 
$\widetilde{D}$ by changing the $(N_{\downarrow}+2)$-th down arrow
to an up arrow and connecting the $i$-th and the $i+1$-th  down arrows 
via an arc.
We denote by $B_{(N_{\downarrow}+2,N_{\downarrow}+1)}(\widetilde{D})$ 
a diagram obtained from $\widetilde{D}$ by changing the 
$(N_\downarrow+2)$-th and $(N_{\downarrow}+1)$-th down arrows to 
two up arrows.
The action of $e_N$ on an arc is 
\begin{eqnarray*}
e_N(
\raisebox{-0.5\totalheight}{
\begin{tikzpicture}
\linkpattern{0/0.6}{}{}{};
\end{tikzpicture}
})
=\downarrow\downarrow
-Q^{-1}\raisebox{-0.5\totalheight}{
\begin{tikzpicture}
\linkpattern{0/0.6}{}{}{};
\end{tikzpicture}}
+q^{-1}(Q-Q^{-1})\uparrow\downarrow
-q^{-1}\uparrow\uparrow.
\end{eqnarray*}
Then, the action of $e_N$ on $D$ is given by
\begin{eqnarray*}
e_N(D)=\widetilde{D}-Q^{-1}D
+q^{-1}(Q-Q^{-1})\sum_{2\le i\le N_{\downarrow}+2}
q^{-(i-2)}A_{(i)}(\widetilde{D}) 
-q^{-1}\sum_{1\le i<j\le N_{\downarrow}+2}
c_{(j,i)}B_{(j,i)}(\widetilde{D}).
\end{eqnarray*}
where 
\begin{eqnarray}
\label{defc}
c_{(j,i)}:=
\begin{cases}
q^{-j+2}, & \text{for $i=1$}, \\
q^{-2i+2}, & \text{for $j=i+1$}, \\
q^{-i-j+3}(1+q^{2}), & \text{otherwise}.
\end{cases}
\end{eqnarray}

\begin{example}
Let $D$ be a diagram depicted as 
\begin{eqnarray*}
D=\tikzpic{-0.6}{
\upa{0};
\upa{2};
\downa{2.35};
\downa{2.7};
\linkpattern{0.5/2.1,0.9/1.7,3.8/4.6}{}{}{};
}.
\end{eqnarray*}
Then, we have 
\begin{eqnarray*}
&&A_{(2)}(\widetilde{D})=
\tikzpic{-0.6}{
\upa{0};
\upa{2};
\downa{2.35};
\linkpattern{0.5/2.1,0.9/1.7,3.4/4.2}{}{}{};
\downa{3.7};
}, \qquad
A_{(3)}(\widetilde{D})=
\tikzpic{-0.6}{
\upa{0};
\upa{2};
\linkpattern{0.5/2.1,0.9/1.7,3/3.8}{}{}{};
\downa{3.3};
\downa{3.7};
}, \\
&&A_{(4)}(\widetilde{D})=
\tikzpic{-0.6}{
\upa{0};
\upa{2};
\upa{2.35};
\linkpattern{0.5/2.1,0.9/1.7}{}{}{};
\downa{2.9};
\downa{3.3};
\downa{3.7};
}, \qquad 
B_{(3,1)}(\widetilde{D})=
\tikzpic{-0.6}{
\upa{0};
\upa{2};
\linkpattern{0.5/2.1,0.9/1.7,3/3.8,4.2/5}{}{}{};
}, \\
&&B_{(2,1)}(\widetilde{D})=
\tikzpic{-0.6}{
\upa{0};
\upa{2};
\linkpattern{0.5/2.1,0.9/1.7,3/4.6,3.4/4.2}{}{}{};
},\qquad
B_{(3,2)}(\widetilde{D})=
\tikzpic{-0.6}{
\upa{0};
\upa{2};
\linkpattern{0.5/2.1,0.9/1.7,3/3.8}{}{}{};
\upa{3.4};
\downa{3.75};
}, \\
&&B_{(4,1)}(\widetilde{D})=
\tikzpic{-0.6}{
\upa{0};
\upa{2};
\upa{2.35};
\downa{2.7}
\linkpattern{0.5/2.1,0.9/1.7,3.8/4.6}{}{}{};
}, \qquad
B_{(4,2)}(\widetilde{D})=
\tikzpic{-0.6}{
\upa{0};
\upa{2};
\upa{2.35};
\linkpattern{0.5/2.1,0.9/1.7,3.4/4.2}{}{}{};
\downa{3.7};
}, \\
&&B_{(4,3)}(\widetilde{D})=
\tikzpic{-0.6}{
\upa{0};
\upa{2};
\upa{2.35};
\upa{2.7};
\downa{3.05};
\linkpattern{0.5/2.1,0.9/1.7}{}{}{};
\downa{3.4};
}.
\end{eqnarray*}
\end{example}

\paragraph{\bf Type BI}
When the rightmost arrow in $D$ is an up arrow or 
a down arrow with the integer $M$, the action of 
$e_N$ on partial diagrams is given by
\begin{eqnarray*}
e_N(\uparrow)
&=&
\tikzpic{-0.5}{
\linkpattern{}{}{0/M}{};
}, \\
e_N\left(\tikzpic{-0.5}{
\linkpattern{}{}{0/M}{};
}\right)
&=&-(q^{M}+q^{-M})\tikzpic{-0.5}{
\linkpattern{}{}{0/M}{};
}.
\end{eqnarray*}
Below, we consider the case where the rightmost arrow 
of $D$ is an up arrow forming an arc.
The other arcs and up arrows are irrelevant for the action 
of $e_N$ since a Kazhdan--Lusztig basis is tensor products
of building blocks. 
The action of $e_N$ on an arc is given by 
\begin{eqnarray*}
e_{N}\left(
\tikzpic{-0.5}{
\linkpattern{0/0.6}{}{}{};
}
\right)
&=&
\tikzpic{-0.6}{
\linkpattern{}{}{0.6/$\bigstar$}{};
\downa{0};
}+
\tikzpic{-0.6}{
\linkpattern{}{}{0.6/$\bigstar$}{};
\upa{0};
},\quad \text{for } M=1, \\
e_{N}\left(
\tikzpic{-0.5}{
\linkpattern{0/0.6}{}{}{};
}
\right)
&=&
\tikzpic{-0.5}{
\linkpattern{}{}{0/{M-1},0.8/M}{};
}
+\langle M-1\rangle
\tikzpic{-0.5}{
\linkpattern{}{}{0.6/M}{};
\downa{0};
},\quad \text{for } M\ge2,
\end{eqnarray*}
where $\langle k\rangle:=q^{k}+q^{-k}$ for $k\in\mathbb{N}_{+}$.
Then, the action of $e_N$ on partial diagrams is given by (see \cite{Shi14-1})
\begin{eqnarray*}
e_N\left(
\tikzpic{-0.5}{
\linkpattern{1.8/2.4}{}{0/r,1.5/M}{};
\draw(0.2,-0.2)node[anchor=west]{$\cdots$};
}
\right)
&=&
\tikzpic{-0.5}{
\linkpattern{}{}{0/{r-2},1.5/M}{};
\draw(0.2,-0.2)node[anchor=west]{$\cdots$};
}
+
\langle r-2\rangle
\tikzpic{-0.5}{
\linkpattern{}{}{0.6/{r-1},2.1/M}{};
\draw(0,0.18)node[anchor=north]{$\uparrow$};
\draw(0.5,-0.2)node[anchor=west]{$\cdots$};
} \\
&&+
\sum_{r\le k\le M}\langle k-1\rangle
\tikzpic{-0.5}{
\linkpattern{1.8/2.6}{}{0/r,1.5/{k-1},2.8/k,4.3/M}{};
\draw(0.2,-0.2)node[anchor=west]{$\cdots$};
\draw(2.2,-0.2)node[anchor=west]{$\cdots$};
},\quad \text{for $r\ge3$}, \\
e_N\left(
\tikzpic{-0.5}{
\linkpattern{1.8/2.4}{}{0/2,1.5/M}{};
\draw(0.2,-0.2)node[anchor=west]{$\cdots$};
}
\right)
&=&
\tikzpic{-0.5}{
\linkpattern{}{}{0.6/$\bigstar$,2.1/M}{};
\draw(0,0.18)node[anchor=north]{$\uparrow$};
\draw(0.5,-0.2)node[anchor=west]{$\cdots$};
} 
+
\tikzpic{-0.5}{
\linkpattern{}{}{0.6/$\bigstar$,2.1/M}{};
\draw(0,0.18)node[anchor=north]{$\downarrow$};
\draw(0.5,-0.2)node[anchor=west]{$\cdots$};
} \\
&&+
\sum_{2\le k\le M}\langle k-1\rangle
\tikzpic{-0.5}{
\linkpattern{1.8/2.6}{}{0/2,1.5/{k-1},2.8/k,4.3/M}{};
\draw(0.2,-0.2)node[anchor=west]{$\cdots$};
\draw(2.2,-0.2)node[anchor=west]{$\cdots$};
}, \\
e_N\left(
\tikzpic{-0.5}{
\linkpattern{1.8/2.4}{}{0/$\bigstar$,1.5/M}{};
\draw(0.2,-0.2)node[anchor=west]{$\cdots$};
}
\right)
&=&
\tikzpic{-0.5}{
\linkpattern{}{0/0.8}{1.2/$\bigstar$,2.7/M}{};
\draw(1.4,-0.2)node[anchor=west]{$\cdots$};
}
+
\tikzpic{-0.5}{
\linkpattern{0/0.8}{}{1.2/$\bigstar$,2.7/M}{};
\draw(1.4,-0.2)node[anchor=west]{$\cdots$};
} \\
&&+
\sum_{2\le k\le M}\langle k-1\rangle
\tikzpic{-0.5}{
\linkpattern{1.8/2.6}{}{0/$\bigstar$,1.5/{k-1},2.8/k,4.3/M}{};
\draw(0.2,-0.2)node[anchor=west]{$\cdots$};
\draw(2.2,-0.2)node[anchor=west]{$\cdots$};
}.
\end{eqnarray*}

\begin{example}
Let $D$ be a diagram depicted as 
\begin{eqnarray*}
D=\tikzpic{-0.6}{
\upa{0};
\linkpattern{0.9/1.7,2.7/3.5,4.1/5.7,4.5/5.3}
{0.5/2.1}{2.4/$\bigstar$,3.8/2}{};x
}, 
\end{eqnarray*}
where $M=2$. 
Then, the action of $e_{13}$ on $D$ is 
\begin{eqnarray*}
e_{13}(D)
&=&
\tikzpic{-0.6}{
\upa{0};
\linkpattern{0.9/1.7,2.8/3.6,4.6/5.4}
{0.5/2.1,2.4/4}{4.3/$\bigstar$,5.7/2}{};
}
+\tikzpic{-0.6}{
\upa{0};
\linkpattern{0.9/1.7,2.8/3.6,4.6/5.4,2.4/4}
{0.5/2.1}{4.3/$\bigstar$,5.7/2}{};
} \\	
&&+\langle1\rangle\tikzpic{-0.6}{
\upa{0};
\linkpattern{0.9/1.7,2.7/3.5,3.9/4.7,5.1/5.9}{0.5/2.1}{2.4/$\bigstar$,6.2/2}{};
}
\end{eqnarray*}
\end{example}

\paragraph{\bf Type BII}
The actions of $e_N$ on partial diagrams are given by
\begin{eqnarray*}
e_N(\uparrow)
&=&\tikzpic{-0.5}{
\linkpattern{}{}{0/o}{};
},\\
e_N\left(
\tikzpic{-0.5}{
\linkpattern{}{}{0/o}{};
}
\right)
&=&-(Q+Q^{-1})\tikzpic{-0.5}{
\linkpattern{}{}{0/o}{};
}, \\
e_N\left(
\tikzpic{-0.5}{
\linkpattern{0/0.8}{}{}{};
}
\right)
&=&\tikzpic{-0.5}{
\linkpattern{}{}{0/e,0.6/o}{};
}.
\end{eqnarray*}

\paragraph{\bf Type BIII}
We have three cases for the rightmost arrow $a$ of $D$: 
1) an arrow $a$ is an up arrow, 
2) $a$ is a down arrow with the circled integer one,
and 3) $a$ is an up arrow forming an arc. 

In the case 1 and 2, we have 
\begin{eqnarray*}
e_N(\uparrow)&=&
\tikzpic{-0.5}{
\linkpattern{}{}{}{0/1};
}, \\
e_N\left(
\tikzpic{-0.5}{
\linkpattern{}{}{}{0/1};
}\right)
&=&-(Q+Q^{-1})
\tikzpic{-0.5}{
\linkpattern{}{}{}{0/1};
}.
\end{eqnarray*}
In the case 3, the action of $e_N$ on an arc is given by 
\begin{eqnarray*}
e_N\left(
\tikzpic{-0.5}{
\linkpattern{0/0.6}{}{}{};
}
\right)
=
\tikzpic{-0.5}{
\linkpattern{}{}{}{0/$2$,0.8/$1$};
}
+\langle\langle1\rangle\rangle
\tikzpic{-0.5}{
\linkpattern{}{}{}{0.7/$1$};
\upa{0};
},
\end{eqnarray*}
where $\langle\langle k\rangle\rangle:=Qq^{-k}+Q^{-1}q^k$.
The action of $e_N$ on a partial diagram is given by
\begin{eqnarray*}
e_N\left(
\tikzpic{-0.5}{
\linkpattern{1.8/2.4}{}{}{0/r,1.5/1};
\draw(0.2,-0.2)node[anchor=west]{$\cdots$};
}
\right)
&=&
\tikzpic{-0.5}{
\linkpattern{}{}{}{0/{r+2},1.5/1};
\draw(0.2,-0.2)node[anchor=west]{$\cdots$};
}
+
\langle\langle r+1\rangle\rangle
\tikzpic{-0.5}{
\draw(-0.5,-0.2)node{$\uparrow$};
\linkpattern{}{}{}{0/{r+1},1.5/1};
\draw(0.2,-0.2)node[anchor=west]{$\cdots$};
} \\
&&\qquad+\sum_{1\le k\le r}\langle\langle k\rangle\rangle
\tikzpic{-0.5}{
\linkpattern{1.8/2.4}{}{}{0/r,1.5/{k+1},2.6/{k},4.1/1};
\draw(0.2,-0.2)node[anchor=west]{$\cdots$};
\draw(2.2,-0.2)node[anchor=west]{$\cdots$};
}.  
\end{eqnarray*}

\begin{example}
Let $D$ be a diagram of type BIII depicted as 
\begin{eqnarray*}
D=\tikzpic{-0.6}{
\linkpattern{0/0.8,2/2.8,3.6/4.4,5.2/6.8,5.6/6.4}
{}{}{1.6/3,3.2/2,4.8/1};
\upa{0.94};
}.
\end{eqnarray*}
Then we define 
\begin{eqnarray*}
&&D_{1}:=\tikzpic{-0.6}{
\linkpattern{0/0.8,2/2.8,3.6/4.4,5.6/6.4}
{}{}{1.6/5,3.2/4,4.7/3,5.3/2,6.8/1};
\upa{0.94};
},\quad
D_{2}:=\tikzpic{-0.6}{
\linkpattern{0/0.8,2/2.8,3.6/4.4,4.8/5.6,6/6.8}
{}{}{1.6/3,3.2/2,7.2/1};
\upa{0.94};
}, \\
&&D_{3}:=\tikzpic{-0.6}{
\linkpattern{0/0.8,2/2.8,3.2/4.8,3.6/4.4,5.6/6.4}
{}{}{1.6/3,5.2/2,6.8/1};
\upa{0.94};
}, \quad 
D_{4}:=\tikzpic{-0.6}{
\linkpattern{0/0.8,1.6/3.2,2/2.8,3.6/4.4,5.6/6.4}{}{}{4.7/3,5.25/2,6.8/1};
\upa{0.94};
}. \\
&&D_{5}:=\tikzpic{-0.6}{
\linkpattern{0/0.8,1.8/2.6,3.3/4.1,5.6/6.4}{}{}{2.9/4,4.4/3,5/2,7/1};
\upa{0.94};
\upa{1.2};
}.
\end{eqnarray*}
The action of $e_{14}$ on $D$ is 
\begin{eqnarray*}
e_{14}(D)=D_1+\langle\langle1\rangle\rangle D_2+\langle\langle2\rangle\rangle D_3
+\langle\langle3\rangle\rangle D_4+\langle\langle4\rangle\rangle D_5.
\end{eqnarray*}

\end{example}

\subsection{\texorpdfstring{Action of $e_0$ on Kazhdan--Lusztig bases}
{Action of e0 on Kazhdan--Lusztig bases}}
\paragraph{\bf Type A}
We denote by $u$ the bijection of sets $u: D\rightarrow D$ defined 
by reflecting a diagram about a vertical axis and reversing 
orientations of all arrows.
The action of $e_0$ on a diagram $D$ is given by
\begin{eqnarray}
\label{e0A}
e_{0}(D)=u(e_N(u(D))), 
\end{eqnarray}
with a change of the parameter $Q\rightarrow Q_{0}$.

For example, the action of $e_0$ on an arc is 
\begin{eqnarray*}
e_0(
\tikzpic{-0.5}{
\linkpattern{0/0.6}{}{}{};
})
=\uparrow\uparrow
-Q_{0}^{-1}\tikzpic{-0.5}{
\linkpattern{0/0.6}{}{}{};}
+q^{-1}(Q_{0}-Q_{0}^{-1})\uparrow\downarrow
-q^{-1}\downarrow\downarrow.
\end{eqnarray*}

\paragraph{\bf Type BI}
Let $D$ be a diagram of type BI and $N_{\uparrow}$ be the number 
of (unpaired) up arrows.
Let $r$ be a smallest integer attached to down arrows 
with an integer $p$, $1\le p\le M$. 
If $D$ has a down arrow with a star, then we define $r=1$.
If there is no down arrow with an integer in $D$, we define $r=M+1$.
We enumerate (unpaired) up arrows from left to right 
by $1, 2, \ldots, N_{\uparrow}$.
We denote by $A_{(i)}(D)$, $1\le i\le N_{\uparrow}-1$, a diagram 
obtained from $D$ by connecting the $i$-th and $(i+1)$-th up arrows  
via an arc.
We have three cases for $D$: 1) the leftmost arrow is an up arrow,
2) the leftmost arrow is an unpaired down arrow,
3) the leftmost arrow is a down arrow forming an arc.

\paragraph{Case 1}
We have two cases for $D$: a) $D$ does not have an unpaired down 
arrow, and b) $D$ has an unpaired down arrow.

\paragraph{Case 1-a}
We denote by $A_{(N_{\uparrow})}$ a diagram obtained 
from $D$ by changing the $N_{\uparrow}$-th up arrow to 
a down arrow with the integer $r-1$ for $r\ge2$. 
If $r=1$, denote by $A_{(N_{\uparrow})}$ a diagram 
obtained from $D$ by changing the $N_{\uparrow}$-th 
up arrow to an unpaired down arrow.
The action of $e_{0}$ on $D$ is given by 
\begin{eqnarray*}
e_{0}(D)=
\sum_{1\le i\le N_{\uparrow}}q^{-(i-1)}A_{(i)}(D)
+(q^{-(r+N_{\uparrow}-2)}(1-\delta_{1,r})-Q_{0})D,
\end{eqnarray*}
where $\delta_{i,j}$ is the delta function, that is, 
$\delta_{i,j}$ is one if $i=j$ and zero otherwise.

\paragraph{Case 1-b}
We denote by $A_{(N_{\uparrow})}(D)$ a diagram obtained from 
$D$ by connecting the $N_{\uparrow}$-th up arrow and 
the unpaired down arrow via a dashed arc, and 
by $A_{(N_{\uparrow}+1)}(D)$ a diagram obtained from $D$ 
by changing the unpaired down arrow to an up arrow.
The action of $e_{0}$ on $D$ is given by 
\begin{eqnarray*}
e_{0}(D)=
\sum_{1\le i\le N_{\uparrow}+1}
q^{-(i-1)}A_{(i)}-Q_{0}D. 
\end{eqnarray*}

\paragraph{Case 2}
Let $\tilde{D}$ be a diagram obtained from $D$ by 
changing the unpaired down arrow to an up arrow.
The action of $e_0$ on $D$ is given by 
\begin{eqnarray*}
e_{0}(D)=\tilde{D}-Q_{0}^{-1}D.
\end{eqnarray*}

\paragraph{Case 3}
Let $\tilde{D}$ be a diagram obtained from $D$ by 
changing the leftmost arc to two up arrows.
We enumerate (unpaired) up arrows of $\tilde{D}$ 
from left to right by $1,2\ldots, N_{\uparrow}+2$.
We denote by $A_{(N_{\uparrow}+2)}(\tilde{D})$ a diagram 
obtained from $\tilde{D}$ by changing the $(N_{\uparrow}+2)$-th 
up arrow to a down arrow.
For each $1\le i<j\le N_{\uparrow}+1$ with $j-i\ge2$, we denote 
by $B_{(i,j)}(\tilde{D})$ a diagram obtained from $\tilde{D}$ 
by connecting the $i$-th and the $(i+1)$-th up arrows via
an arc and the $j$-th and the $(j+1)$-th up arrows via an arc.
For each $1\le i\le N_{\uparrow}+1$, we denote by $B_{(i,i+1)}$
a diagram obtained from $\tilde{D}$ by connecting the $i$-th 
and the $(i+3)$-th up arrows via an arc and the $(i+1)$-th 
and the $(i+2)$-th up arrows via an arc. 
We denote by $B_{(N_{\uparrow},N_{\uparrow}+1)}(\tilde{D})$ a 
diagram obtained from $\tilde{D}$ by connecting the 
$(N_{\uparrow}+1)$-th and the $(N_{\uparrow}+2)$-th up arrows 
via an arc and putting a down arrow at the $(N_{\uparrow}+2)$-th 
site. 
For each $1\le i\le N_{\uparrow}$, we denote by 
$B_{(i,N_{\uparrow}+2)}(\tilde{D})$ a diagram obtained from 
$\tilde{D}$ by connecting the $i$-th and the $(i+1)$-th up 
arrows via an arc and putting a down arrow at $(N_{\uparrow}+2)$-th 
site.
Finally, we denote by $B_{(N_{\uparrow}+1,N_{\uparrow}+2)}(\tilde{D})$ 
a diagram obtained from $\tilde{D}$ by putting two down arrows at 
the $(N_{\uparrow}+1)$-th and the $(N_{\uparrow}+2)$-th sites.

We have three cases for $D$: 
a) $r=1$ and $D$ does not have an unpaired down arrow,
b) $r\ge2$, and 
c) $D$ has an unpaired down arrow.

\paragraph{Case 3-a}
The action of $e_0$ on $D$ is given by
\begin{eqnarray*}
e_{0}(D)&=&
(1-q^{-2N_{\uparrow}-2})\tilde{D}
-Q_{0}^{-1}D
+
q^{-1}(Q_{0}-Q_{0}^{-1})\sum_{2\le i\le N_{\uparrow}+2}
q^{-(i-2)}A_{(i)}(\tilde{D})  \\
&&\qquad-q^{-1}\sum_{1\le i<j\le N_{\uparrow}+2}
c_{(j,i)}B_{(i,j)}(\tilde{D}),
\end{eqnarray*}
where $c_{(j,i)}$ is defined in Eqn. (\ref{defc}).

\paragraph{Case 3-b}
The action of $e_0$ on $D$ is given by 
\begin{eqnarray*}
e_{0}(D)&=&
(1+q^{-N_{\uparrow}-r}(Q_{0}-Q_{0}^{-1})-q^{-2N_{\uparrow}-2r})\tilde{D}
-(Q_{0}^{-1}+q^{-N_{\uparrow}-r})D \\
&&+
\sum_{2\le i\le N_{\uparrow}+2}\tilde{c}_{(i)}A_{(i)}(\tilde{D})
-q^{-1}\sum_{1\le i<j\le N_{\uparrow}+2}
c_{(j,i)}B_{(i,j)}(\tilde{D})
\end{eqnarray*}
where 
\begin{eqnarray*}
\tilde{c}_{(i)}:=q^{-(i-1)}(Q_{0}-Q_{0}^{-1})
-q^{-N_{\uparrow}-r-i+1}(1+(1-\delta_{2,r}\delta_{i,N_{\uparrow}+2})q^{2}).
\end{eqnarray*}

\paragraph{Case 3-c}
For each $1\le i\le N_{\uparrow}+1$, we denote by $B_{(i,N_{\uparrow}+3)}(\tilde{D})$
a diagram obtained from $\tilde{D}$ by connecting the $i$-th and $(i+1)$-th up 
arrows via an arc and changing the unpaired down arrow to an up arrow.
We also denote by $A'_{(N_{\uparrow}+2)}(\tilde{D})$ a diagram obtained 
from $\tilde{D}$ by connecting the $(N_{\uparrow}+2)$-th up arrow and the 
unpaired down arrow via an arc.
Then, the action of $e_0$ on $D$ is given by 
\begin{eqnarray*}
e_{0}(D)&=&(1-q^{-2N_{\uparrow}-4})\tilde{D}-Q_{0}^{-1}D
+(Q_{0}-Q_{0}^{-1})\sum_{2\le i\le N_{\uparrow}+3}
q^{-(i-1)}A_{(i)}(\tilde{D}) \\
&&-q^{-2N_{\uparrow}-3}A'_{(N_{\uparrow}+2)} 
-q^{-1}\sum_{1\le i<j\le N_{\uparrow}+2}c_{(j,i)}B_{(i,j)}(\tilde{D}) \\
&&-q^{-1}\sum_{1\le i\le N_{\uparrow}+1}
q^{-N_{\uparrow}-i}(1+(1-\delta_{1,i})q^{2})B_{(i,N_{\uparrow}+3)}(\tilde{D}).
\end{eqnarray*}

\paragraph{\bf Type BII}
We have three cases for the leftmost arrow $a$ of $D$: 
1) $a$ is an up arrow,
2) $a$ is an e- or o-unpaired down arrow, and 
3) $a$ is a down arrow forming an arc. 

\paragraph{Case 1}
Let $N_{\uparrow}$ be the number of up arrows of a diagram $D$.
We enumerate up arrows from left to right by 
$1, 2, \ldots, N_{\uparrow}$. 
For each $1\le i\le N_{\uparrow}-1$, we denote by $A_{(i)}(D)$ a 
diagram obtained from $D$ by connecting the $i$-th and $(i+1)$-th 
up arrows via an arc. 
We denote by $A_{(N_{\uparrow})}(D)$ a diagram obtained 
from $D$ by changing the $N_{\uparrow}$-th up arrow to 
an e- or o-unparied down arrow.

Suppose that the leftmost down arrow of $D$ is an o-unpaired
down arrow.
Then, the action of $e_{0}$ on $D$ is given by 
\begin{eqnarray*}
e_{0}(D)=-(Q_{0}+q^{-N_{\uparrow}})D
+
\sum_{1\le i\le N_{\uparrow}}q^{-(i-1)}A_{(i)}.
\end{eqnarray*}
Suppose that the leftmost down arrow of $D$ is an e-unpaired 
down arrow or $D$ does not have a down arrow. 
The action of $e_{0}$ on $D$ is given by 
\begin{eqnarray*}
e_{0}(D)=(q^{-N_{\uparrow}+1}Q-Q_{0})D
+\sum_{1\le i\le N_{\uparrow}}q^{-(i-1)}A_{(i)}.
\end{eqnarray*}

\paragraph{Case 2}
The action of $e_0$ on a partial diagram of $D$ is given by 
\begin{eqnarray*}
e_{0}\left(
\tikzpic{-0.5}{
\linkpattern{}{}{0/e}{};
}\right)
&=&(q^{-1}Q-Q_{0}^{-1})
\tikzpic{-0.5}{\linkpattern{}{}{0/e}{};}
-(q^{-2}Q^2-q^{-1}QQ_{0}^{-1}+q^{-1}QQ_{0}-1)\uparrow, \\
e_{0}\left(
\tikzpic{-0.5}{
\linkpattern{}{}{0/o}{};
}\right)
&=&-(Q^{-1}+Q_{0}^{-1})
\tikzpic{-0.5}{\linkpattern{}{}{0/o}{};}
+(1+Q^{-1}Q_{0}-Q^{-1}Q_{0}^{-1}-Q^{-2})\uparrow,
\end{eqnarray*}

\paragraph{Case 3}
We have two cases for $D$: 
\begin{enumerate}[(a)]
\item The leftmost down arrow is an o-unpaired down arrow. 
The action of $e_0$ on $D$ is given by 
\begin{eqnarray*}
e_{0}\left(
\tikzpic{-0.5}{\linkpattern{0/0.6}{}{}{};}
\right)
&=&
-q^{-1}
\tikzpic{-0.5}{\linkpattern{}{}{0/o,0.6/e}{};}
-(Q_{0}^{-1}-q^{-2}Q)\tikzpic{-0.5}{\linkpattern{0/0.6}{}{}{};}
-(q^{-1}(Q_{0}^{-1}+Q^{-1})-q^{-3}Q-q^{-1}Q_{0})
\tikzpic{-0.5}{\linkpattern{}{}{0.5/e}{};\upa{0};} \\
&&+(1-q^{-4}Q^{2}+q^{-2}Q_{0}^{-1}Q-q^{-2}QQ_{0})\uparrow\uparrow.
\end{eqnarray*}

\item The leftmost down arrow is an e-unpaired down arrow or 
$D$ does not have a down arrow.
The action of $e_0$ on $D$ is given by
\begin{eqnarray*}
e_{0}\left(
\tikzpic{-0.5}{\linkpattern{0/0.6}{}{}{};}
\right)
&=&
-q^{-1}
\tikzpic{-0.5}{\linkpattern{}{}{0/e,0.6/o}{};}
-(Q_{0}^{-1}+q^{-1}Q^{-1})
\tikzpic{-0.5}{\linkpattern{0/0.6}{}{}{};}
-(q^{-1}Q_{0}^{-1}+q^{-2}(Q^{-1}-Q)-q^{-1}Q_{0})
\tikzpic{-0.5}{\linkpattern{}{}{0.5/o}{};\upa{0};} \\
&&+(1-q^{-2}Q^{-2}-q^{-1}Q^{-1}Q_{0}^{-1}+q^{-1}Q^{-1}Q_{0})
\uparrow\uparrow.
\end{eqnarray*}
\end{enumerate}

\paragraph{\bf Type BIII}
Let $N_{\uparrow}$ be the number of up arrows of a diagram 
$D$ and $r$ be the largest integer attached to down arrows 
with a circled integer. 
If there is no down arrow in $D$, we define $r=0$.
We enumerate up arrow from left to right.
For each $1\le i\le N_{\uparrow}-1$, we denote by $A_{(i)}(D)$ 
a diagram obtained from $D$ by connecting the $i$-th and 
$(i+1)$-th arrows via an arc. 
We denote by $A_{(N_{\uparrow})}(D)$ a diagram obtained 
from $D$ by changing the $N_{\uparrow}$-th up arrow 
to a down arrow with a circled integer $r+1$.
We have three cases for $D$: 
1) the leftmost arrow is an up arrow,
2) the leftmost arrow is a down arrow with a circled 
integer $r$, and 
3) the leftmost arrow is a down arrow forming an arc.

\paragraph{Case 1}
The action of $e_0$ on $D$ is given by 
\begin{eqnarray*}
e_{0}(D)
=
(q^{-N_{\uparrow}+r+1}Q^{-1}-Q_{0})D+
\sum_{1\le i\le N_{\uparrow}}q^{-(i-1)}A_{(i)}(D).
\end{eqnarray*}

\paragraph{Case 2}
The action of $e_0$ on $D$ is given by 
\begin{eqnarray*}
e_{0}\left(\tikzpic{-0.6}{\linkpattern{}{}{}{0/r};}\right)
=
(-Q_{0}^{-1}-q^{r-1}Q^{-1})\tikzpic{-0.5}{\linkpattern{}{}{}{0/r};}
+(1+q^{r-1}Q^{-1}(Q_{0}-Q_{0}^{-1})-q^{2r-2}Q^{-2})\uparrow.
\end{eqnarray*}

\paragraph{Case 3}
Let $\tilde{D}$ be a diagram obtained from $D$ by changing the 
leftmost arc to two up arrows.
For each $1\le i<j\le N_{\uparrow}+1$ with $j-i\ge2$, we denote 
by $B_{(i,j)}(\tilde{D})$ a diagram obtained from $\tilde{D}$ 
by connecting the $i$-th and the $(i+1)$-th up arrows via an arc 
and the $j$-th and $(j+1)$-th up arrows via an arc. 
For each $1\le i\le N_{\uparrow}-1$, we denote by 
$B_{(i,i+1)}(\tilde{D})$ a diagram obtained from $\tilde{D}$ by 
connecting the $i$-th and the $(i+3)$-th up arrows via an 
arc and the $(i+1)$-th and $(i+2)$-th up arrows via an arc.
We denote by $B_{(N_{\uparrow},N_{\uparrow}+1)}(\tilde{D})$ a 
diagram obtained from $\tilde{D}$ by connecting the $(N_{\uparrow}+1)$-th
and the $(N_{\uparrow}+2)$-th up arrows via an arc and putting a
down arrow with the circled integer $r+1$ at the $N_{\uparrow}$-th site.
For each $1\le i\le N_{\uparrow}$, we denote by 
$B_{(i,N_{\uparrow}+2)}(\tilde{D})$ a diagram obtained from $\tilde{D}$ 
by connecting the $i$-th and the $(i+1)$-th up arrows via an arc and 
putting a down arrow with the circled integer $r+1$ at the $N_{\uparrow}+2$-th 
site. 
We denote by $B_{(N_{\uparrow}+1,N_{\uparrow}+2)}(\tilde{D})$ a diagram 
obtained from $\tilde{D}$ by putting two down arrows with circled integers 
$r+1$ and $r+2$ at the $(N_{\uparrow}+1)$-th and the $(N_{\uparrow}+2)$-th
sites. 
The action of $e_0$ on $D$ is given by 
\begin{eqnarray*}
e_{0}(D)&=& 
(1+q^{-N_{\uparrow}+r-1}Q^{-1}(Q_{0}-Q_{0}^{-1})
-q^{-2N_{\uparrow}+2r-2}Q^{-2})\tilde{D}
+(-Q_{0}^{-1}-q^{-N_{\uparrow}+r-1}Q^{-1})D \\
&&+\sum_{2\le i\le N_{\uparrow}+2}\tilde{c}_{(i)}A_{(i)}(\tilde{D})
-q^{-1}\sum_{1\le i<j\le N_{\uparrow}+2}c_{(j,i)}B_{(i,j)}(\tilde{D})
\end{eqnarray*}
where $c_{(j,i)}$ is defined in Eqn.(\ref{defc}) and 
\begin{eqnarray*}
\tilde{c}_{(i)}:=q^{-(i-1)}(Q_{0}-Q_{0}^{-1})
-q^{-N_{\uparrow}+r-i}(1+q^{2})Q^{-1}.
\end{eqnarray*}

\section{\texorpdfstring{Eigensystem of $X$}{Eigensystem of X}}
\label{sec-eigenX}
Since $X$ commutes with the Hamiltonian $H^{1B}$ (Theorem~\ref{theorem-commute}), 
an eigenvector of $X$ with the multiplicity one is also an eigenvector
of $H^{1B}$.
We will first find an eigenvector of $X$ with the multiplicity one.

\subsection{Type A}
We consider the action of $X$ on the Kazhdan--Lusztig basis of type A.
Let $D$ be a diagram of type A, $n_\uparrow$ be the number of (unpaired) 
up arrows and $n_\downarrow$ be the number of (unpaired) down arrows. 
We define the weight of $D$ by $\mathrm{wt}(D)=n_{\uparrow}-n_{\downarrow}$.
We enumerate the (unpaired) up arrows from left to right by 
$1,2,\ldots,n_{\uparrow}$.
For each $i$, $1\le i<n_{\uparrow}$, we denote by $E_{(i)}(D)$ a diagram 
obtained from $D$ by connecting the $i$-th and $(i+1)$-th up arrows 
via an arc. 
We denote by $E_{(n_{\uparrow})}$ a diagram obtained from $D$ by changing 
the $n_{\uparrow}$-th up arrow to a down arrow.
Similarly, we enumerate (unpaired) down arrows from right to left by 
$1,2,\ldots,n_{\downarrow}$.
For each $i$, $1\le i<n_{\downarrow}$, we denote by $F_{(i)}(D)$ a diagram
obtained from $D$ by connecting the $i$-th and $(i+1)$-th down arrows 
via an arc.
We denote by $F_{(n_{\downarrow})}$ a diagram obtained from $D$ by changing 
the $n_{\downarrow}$-th down arrow to an up arrow.

We define the action of $X$ by 
\begin{eqnarray}
\label{ActionXA}
X(D)
:=\sum_{1\le i\le n_{\uparrow}}[i] E_{(i)}(D)
+\sum_{1\le i\le n_{\downarrow}}q^{\mathrm{wt}(D)+1}[i] F_{(i)}(D)
+q^{\mathrm{wt}(D)}\frac{Q-Q^{-1}}{q-q^{-1}}D.
\end{eqnarray}

\begin{example}
Let $D$ be a diagram depicted as 
\begin{eqnarray*}
D=\tikzpic{-0.6}{
\draw(-0.5,0)--(-0.5,-0.6)(-0.5-0.12,-0.12)--(-0.5,0)--(-0.5+0.12,-0.12);
\draw(-0.1,0)--(-0.1,-0.6)(-0.1-0.12,-0.12)--(-0.1,0)--(-0.1+0.12,-0.12);
\linkpattern{0.4/2,0.8/1.6,4/4.8}{}{}{};
\draw(2,0)--(2,-0.6)(2-0.12,-0.12)--(2,0)--(2+0.12,-0.12);
\draw(2.4,0)--(2.4,-0.6)(2.4-0.12,-0.6+0.12)--(2.4,-0.6)--(2.4+0.12,-0.6+0.12);
\draw(2.8,0)--(2.8,-0.6)(2.8-0.12,-0.6+0.12)--(2.8,-0.6)--(2.8+0.12,-0.6+0.12);
\draw(4.3,0)--(4.3,-0.6)(4.3-0.12,-0.6+0.12)--(4.3,-0.6)--(4.3+0.12,-0.6+0.12);
}
\end{eqnarray*}
Then, we have $\mathrm{wt}(D)=0$ and
\begin{eqnarray*}
&&E_{(1)}=
\tikzpic{-0.5}{
\linkpattern{0/0.8,1.2/2.8,1.6/2.4,5/5.8}{}{}{};
\upa{2.7};
\downa{3.1};
\downa{3.5};
\downa{5};
},\qquad
E_{(2)}=
\tikzpic{-0.5}{
\linkpattern{0.8/3.2,1.2/2.8,1.6/2.4,5/5.8}{}{}{};
\upa{0.2};
\downa{3.1};
\downa{3.5};
\downa{5};
}, \\
&&E_{(3)}=
\tikzpic{-0.5}{
\upa{0};
\upa{0.5};
\linkpattern{1.2/2.8,1.6/2.4,5/5.8}{}{}{};
\downa{2.7}
\downa{3.1};
\downa{3.5};
\downa{5};
}, \qquad
F_{(1)}
=\tikzpic{-0.5}{
\upa{0};
\upa{0.5};
\linkpattern{1.2/2.8,1.6/2.4,4.6/6.2,5/5.8}{}{}{};
\upa{2.7};
\downa{3.1};
}, \\
&&F_{(2)}
=\tikzpic{-0.5}{
\upa{0};
\upa{0.5};
\linkpattern{1.2/2.8,1.6/2.4,3.9/4.7,5/5.8}{}{}{};
\upa{2.7};
\downa{5.1};
}, \qquad
F_{(3)}
=\tikzpic{-0.5}{
\upa{0};
\upa{0.5};
\linkpattern{1.2/2.8,1.6/2.4,3,5/5.8}{}{}{};
\upa{2.7};
\upa{3.2};
\downa{3.7};
\downa{5.1};
}.
\end{eqnarray*}
Therefore, the action of $X$ on $D$ is given by 
\begin{eqnarray*}
X(D)=E_{(1)}+qF_{(1)}+[2](E_{(2)}+qF_{(2)})+[3](E_{(3)}+qF_{(3)})+[Q;0]D.
\end{eqnarray*}
\end{example}

\begin{theorem}
The above definition provides the action of $X$ on the Kazhdan--Lusztig 
basis of type A.
\end{theorem}
\begin{proof}
We prove Theorem by induction. 
When $N=1$ or $2$, Theorem holds true by a direct computation.
Suppose that Theorem is true up to some $N\ge2$.
We have two cases for the leftmost arrow: 1) an up arrow and 2) a down arrow.
\paragraph{\bf Case 1}
Let $D$ be a diagram $\uparrow D'$ where $D'$ is a diagram of length $N-1$.
By using the comultiplication, we have 
\begin{eqnarray*}
X(D)&=&(K\otimes X+q^{-1}KE\otimes 1+F\otimes1)(\uparrow D') \\
&=&q\uparrow X(D')+\downarrow D'.
\end{eqnarray*}
From the assumption, we have 
\begin{eqnarray*}
\uparrow X(D')&=&\sum_{1\le i\le n'_{\uparrow}}[i]\uparrow E_{(i)}(D')
+\sum_{1\le i\le n'_{\downarrow}}q^{\mathrm{wt}(D')+1}[i]\uparrow F_{(i)}(D')
+q^{\mathrm{wt}(D')}\frac{Q-Q^{-1}}{q-q^{-1}}\uparrow D' \\
&=&\sum_{2\le i\le n_{\uparrow}}[i-1]E_{(i)}(D)
+\sum_{1\le i\le n_{\downarrow}}q^{\mathrm{wt}(D)}[i]F_{(i)}(D)
+q^{\mathrm{wt}(D)-1}\frac{Q-Q^{-1}}{q-q^{-1}}D
\end{eqnarray*}
where $n'_{\uparrow}$ (resp. $n'_{\downarrow}$) is the number of up
(resp. down) arrows in $D'$ and we have used 
$\mathrm{wt}(D)=\mathrm{wt}(D')+1$, $n_{\uparrow}=n'_{\uparrow}+1$ and 
$n_{\downarrow}=n'_{\downarrow}$.
We also have 
\begin{eqnarray*}
\downarrow D'=\sum_{1\le i\le n_{\uparrow}}q^{-(i-1)}E_{(i)}(D)
\end{eqnarray*}
By $q[i-1]+q^{-(i-1)}=[i]$, the sum $q\uparrow X(D')+\downarrow D'$ 
gives a desired expression.

\paragraph{\bf Case 2}
We have two cases for $D$: a) $D$ has no up arrows 
and b) the leftmost arrow forms an arc.
\paragraph{Case 2-a}
The diagram $D$ is written as $\downarrow D'$.
We want to compute $X(D)=q^{-1}\downarrow X(D')+\uparrow D'$.
We have 
\begin{eqnarray}
\label{typeA-1}
\downarrow X(D')
=\sum_{1\le i\le n'_{\downarrow}}q^{\mathrm{wt}(D')+1}[i]\downarrow F_{(i)}(D')
+q^{\mathrm{wt}(D')}\frac{Q-Q^{-1}}{q-q^{-1}}\downarrow D'.
\end{eqnarray}
Note that $\uparrow D'=F_{(n_{\downarrow})}(D)$, 
$\downarrow F_{(i)}(D')=F_{(i)}(D)$ for $1\le i\le n_{\downarrow}-1$ and  
$\downarrow F_{(n'_{\downarrow})}(D')=F_{(n_{\downarrow}-1)}(D)
+q^{-1}F_{(n_{\downarrow})}(D)$.
We also have $\mathrm{wt}(D')=\mathrm{wt}(D)+1$ and $[i]q^{-i-1}+1=q^{-i}[i+1]$.
Inserting these into Eqn.(\ref{typeA-1}), we obtain Eqn.(\ref{ActionXA}).

\paragraph{Case 2-b}
Let $D'$ be a diagram obtained from $D$ by removing arcs and $D''=X(D')$.
From the definition of the action of $X$, the action of $X$ on $D$ is obtained 
by inserting the removed arcs of $D$ into diagrams $D''$ at the same position as $D$.
Thus, without loss of generality, we assume 
\begin{eqnarray*}
D=
\raisebox{-0.5\totalheight}{
\begin{tikzpicture}
\draw(0,0)..controls(0,-0.5)and(0.4,-0.5)..(0.4,0);
\end{tikzpicture}}
\underbrace{\uparrow\ldots\uparrow}_{x_1}\underbrace{\downarrow\ldots\downarrow}_{x_2}.
\end{eqnarray*}
The action of $X$ on $D$ is given by 
\begin{eqnarray*}
X(D)
&=&X(\downarrow\underbrace{\uparrow\ldots\uparrow}_{x_1+1}
\underbrace{\downarrow\ldots\downarrow}_{x_2})
-q^{-1}X(\uparrow\downarrow\underbrace{\uparrow\ldots\uparrow}_{x_1}
\underbrace{\downarrow\ldots\downarrow}_{x_2}) \\
&=&
q^{-1}\downarrow X(\underbrace{\uparrow\ldots\uparrow}_{x_1+1}
\underbrace{\downarrow\ldots\downarrow}_{x_2})
+\underbrace{\uparrow\ldots\uparrow}_{x_1+2}
\underbrace{\downarrow\ldots\downarrow}_{x_2}
-\uparrow X(\downarrow\underbrace{\uparrow\ldots\uparrow}_{x_1}
\underbrace{\downarrow\ldots\downarrow}_{x_2})
-q^{-1}\downarrow\downarrow\underbrace{\uparrow\ldots\uparrow}_{x_1}
\underbrace{\downarrow\ldots\downarrow}_{x_2} \\
&=&
\raisebox{-0.5\totalheight}{
\begin{tikzpicture}
\draw(0,0)..controls(0,-0.5)and(0.4,-0.5)..(0.4,0);
\end{tikzpicture}}
\ 
X(\underbrace{\uparrow\ldots\uparrow}_{x_1}
\underbrace{\downarrow\ldots\downarrow}_{x_2}).
\end{eqnarray*}
Thus we have a desired expression (\ref{ActionXA}).
\end{proof}

\begin{theorem}
$X$ has the eigenvalue  $[Q;N-2i]$, $0\le i\le N$,
of multiplicity $\displaystyle\genfrac{(}{)}{0pt}{}{N}{i}$.
\end{theorem}
\begin{proof}
We consider the matrix representation of $X=(X_{D,D'})$ on the 
Kazhdan--Lusztig bases.
We will construct eigenvectors of $X$.


Let $\mathcal{D}_n$, $0\le n\le \lfloor N/2\rfloor$, be the set of diagrams 
with $n$ arcs, $\mathcal{D}_n^{\le}:=\bigcup_{0\le i\le n}\mathcal{D}_{i}$ 
and $\mathcal{D}_{n}^{\ge}:=\bigcup_{n\le i\le\lfloor N/2\rfloor}\mathcal{D}_{i}$.
The cardinality of $\mathcal{D}_n$, $|\mathcal{D}_{n}|$, is given by 
\begin{eqnarray*}
|\mathcal{D}_n|=(N-2n+1)\left(\genfrac{(}{)}{0pt}{}{N}{n}
-\genfrac{(}{)}{0pt}{}{N}{n-1}\right).
\end{eqnarray*}
Let $I$ be a set of the positions of arcs from left in a diagram 
$D\in\mathcal{D}_{n}$ and denote by $\mathcal{D}_n^{I}$ the set of 
diagrams with arcs located as $I$.
Then, the set $\mathcal{D}_{n}$ is a direct sum of $\mathcal{D}_{n}^{I}$, 
that is, $\mathcal{D}_{n}=\bigsqcup_{I}\mathcal{D}_{n}^{I}$. 
The cardinality of $\mathcal{D}_n^{I}$ is given by 
$|\mathcal{D}_{n}^{I}|=N-2n+1$.

We define a vector $\psi:=\sum_{D}\psi_{D}D$ with the following property.
We set $\psi_{D}=0$ for all $D\in\mathcal{D}_{n}^{\le}$ except 
some $D\in\mathcal{D}_{n}^{I}$.
Let $A=(X_{D,D'})_{D,D'\in\mathcal{D}_n^I}$ be a submatrix of $X$. 
If there exists an eigenvector $\psi$ of $X$ with the above property,
the eigenvalues of $A$ coincides with the ones of $X$.
This is because an element of $\mathcal{D}_{n-1}^{\le}$ cannot be 
appeared in the expansion of $X(D)$ for $D\in\mathcal{D}_{n}^{\ge}$ 
(see Eqn.(\ref{ActionXA})).
The submatrix $A$ is of size $N-2n+1$ and tridiagonal whose entries 
are
\begin{eqnarray*}
A_{i,i}=q^{N-2n+2-2i}\frac{Q-Q^{-1}}{q-q^{-1}}, \qquad
A_{i,i-1}=[N-2n+1-i],\qquad
A_{i,i+1}=q^{N-2n-2i}[i].
\end{eqnarray*}
From Lemma~\ref{lemma-app-A}, the eigenvalues are
$[Q;N-2n-2\lambda]$, $\lambda=0,1,\ldots,N-2n$
and the multiplicities are one.
For each eigenvalue of $A$, there exists a unique eigenvector and 
we set $\psi_D$, $D\in\mathcal{D}^{I}$, as this eigenvector.
Given an eigenvalue $a$ of $A$ and $\psi_D$, $D\in\mathcal{D}_{n}^{\le}$, 
other components $\psi_D$, $D\in\mathcal{D}_{n+1}^{\ge}$ are determined 
by solving the eigenvalue problem. 
If the multiplicity of $a$ (as the eigenvalue of $X$) is not one, then 
$\psi$ may not be determined uniquely. 
However, we have at least one eigenvector of $X$ and this eigenvector 
is characterized by $n$, $a$ and $I$.
Since the eigenvalues are of the form $[Q;N-2j]$, $0\le j\le N$, 
the multiplicity is given by
\begin{eqnarray*}
\sum_{i=0}^{\min(j,N-j)}|\mathcal{D}_n|/|\mathcal{D}_n^{I}|
&=&\sum_{i=0}^{\min(j,N-j)}\genfrac{(}{)}{0pt}{}{N}{i}-\genfrac{(}{)}{0pt}{}{N}{i-1} \\
&=&\genfrac{(}{)}{0pt}{}{N}{j}.
\end{eqnarray*}
This completes the proof.
\end{proof}

Let $D$ be a diagram of Type A. 
We define $S$ as the set of arcs, $S_{\uparrow}$ as the set of unpaired up arrows
and $S_{\downarrow}$ as the set of unpaired down arrows.
We define 
\begin{eqnarray*}
N_1&=&q^{d(d-1)/2}Q^{d} 
\end{eqnarray*}
where $d=|S_{\uparrow}|+|S|$.

We enumerate up arrows, down arrows and arcs from left. 
If there are arcs inside of an arc, we increase an integer 
one by one from outside to inside.
Let $N_A$ be an integer assigned to 
$A\in S\cup S_{\uparrow}\cup S_{\downarrow}$. 
We define 
\begin{eqnarray*}
N_2:=\prod_{A\in S\cup S_{\downarrow}}[N_A].
\end{eqnarray*}
If the $i$-th down arrow and the $j$-th ($j>i$) up arrow
form an arc, we define the size of arc as $(j-i+1)/2$.
Let $B$ be an arc and $m_B$ be its size. 
We define 
\begin{eqnarray*}
N_3:=\prod_{B\in S}[m_B]^{-1}.
\end{eqnarray*}

Similarly, we enumerate down arrows and arcs from right. 
Let $N_C$ be an integer assigned to $C\in S\cup S_{\downarrow}$.
We define 
\begin{eqnarray*}
N_4:=\prod_{C\in S_{\downarrow}}[N_C]^{-1}.
\end{eqnarray*}

In the above notation, we define the vector $\Psi:=\sum_{D}\Psi_{D}|D\rangle$:
\begin{defn}
\label{PsiA}
$\Psi_{D}=N_1\cdot N_2\cdot N_3\cdot N_4$.
\end{defn}

\begin{example}
Let $D$ be a diagram depicted as 
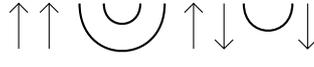
\begin{eqnarray*}
\tikzpic{0}{
\draw(0,0)--(0,-0.6)(-0.12,-0.12)--(0,0)--(0.12,-0.12);
\draw(0.4,0)--(0.4,-0.6)(-0.12+0.4,-0.12)--(0.4,0)--(0.4+0.12,-0.12);
\draw(2.3,0)--(2.3,-0.6)(2.3-0.12,0-0.12)--(2.3,-0)--(2.3+0.12,-0.12);
\draw(2.7,0)--(2.7,-0.6)(2.7-0.12,-0.6+0.12)--(2.7,-0.6)--(2.7+0.12,-0.6+0.12);
\draw(3.8,0)--(3.8,-0.6)(3.8-0.12,-0.6+0.12)--(3.8,-0.6)--(3.8+0.12,-0.6+0.12);
\linkpattern{1.0/2.4,1.4/2,3.7/4.5}{}{}{};
}.
\end{eqnarray*}
We have 
\begin{eqnarray*}
&&N_1=q^{15}Q^{6},\qquad N_2=\frac{[8]!}{[2][5]}, \qquad
N_3=[2]^{-1}, \qquad N_4=[3]^{-1}.
\end{eqnarray*}
\end{example}

\begin{theorem}
\label{thrm-PsiA}
$\Psi$ is the eigenvector of $X$ with the eigenvalue 
$[Q;N]$.
\end{theorem}
\begin{proof}
Let $D$ be a diagram starting with $n_1$ up arrows, followed by an outer 
arc of size $m_1$, followed by $n_2$ up arrows, followed by an outer arc 
of size $m_2$, $\cdots$, followed by $n_{I+1}$ up arrows, followed by 
$n'_{J+1}$ down arrows, followed by an outer arc of size $m'_{J}$,
followed by $n'_{J}$ down arrows, $\cdots$, and ending with $n'_1$ 
down arrows.
As a diagram, $D$ is 
\begin{eqnarray}
\label{Diagram0}
\underbrace{\uparrow\ldots\uparrow}_{n_1}\! \!
\raisebox{-0.8\totalheight}{
\begin{tikzpicture}
\draw(0,0)..controls(0,-0.8)and(1,-0.8)..(1,0);
\draw(0.5,-0.8)node{size $m_1$};
\end{tikzpicture}}\! \! 
\uparrow\ldots\uparrow \! \! \! 
\raisebox{-0.8\totalheight}{
\begin{tikzpicture}
\draw(0,0)..controls(0,-0.8)and(1,-0.8)..(1,0);
\draw(0.5,-0.8)node{size $m_I$};
\end{tikzpicture}}\! \! 
\underbrace{\uparrow\ldots\uparrow}_{n_{I+1}}
\underbrace{\downarrow\ldots\downarrow}_{n'_{J+1}}
\! \!\!
\raisebox{-0.8\totalheight}{
\begin{tikzpicture}
\draw(0,0)..controls(0,-0.8)and(1,-0.8)..(1,0);
\draw(0.5,-0.8)node{size $m'_J$};
\end{tikzpicture}} \! \! \!
\downarrow\ldots\downarrow \!\!\!
\raisebox{-0.8\totalheight}{
\begin{tikzpicture}
\draw(0,0)..controls(0,-0.8)and(1,-0.8)..(1,0);
\draw(0.5,-0.8)node{size $m'_1$};
\end{tikzpicture}} \!\!
\underbrace{\downarrow\ldots\downarrow}_{n'_1}
\end{eqnarray}
where the inside of an outer arc is filled with arcs.

Set $N_{\uparrow}=\sum_{i=1}^{I+1}n_i$, $N_{\downarrow}=\sum_{i=1}^{J+1}n'_i$, 
$M=\sum_{i=1}^{I}m_i$ and $M'=\sum_{i=1}^{J}m'_i$.
The component $\Psi_D$ is explicitly given by
\begin{multline}
\label{PsiA2}
\Psi_{D}
=q^{d(d-1)/2}Q^d
\prod_{i=1}^{I}\frac{[\sum_{j=1}^{i}(n_j+m_j)]!}{[n_i+\sum_{j=1}^{i-1}(n_j+m_j)]!}
\cdot 
\frac{[N_{\uparrow}+M+N_{\downarrow}+M']!}{[N_{\uparrow}+M]!}\cdot
\prod_{A\in S}[m_A]^{-1} \\
\times
\prod_{i=1}^{J+1}\frac{[\sum_{j=1}^{i-1}(n'_j+m'_j)]!}{[n'_i+\sum_{j=1}^{i-1}(n'_j+m'_j)]!}
\end{multline}
where $d=N_{\uparrow}+M+M'$.

Let $X_{D,D'}$ be the matrix representation of the action of $X$ on the Kazhdan--Lusztig 
bases, that is, $X(D)=\sum_{D'}X_{D',D}D'$. 
Note that the explicit formulae for $X_{D',D}$ is given by Eqn.(\ref{ActionXA}).
We want to show that 
\begin{eqnarray}
\label{EPA}
\sum_{D'}X_{D,D'}\Psi_{D'}=[Q;N]\Psi_{D}.
\end{eqnarray}
We have five cases for $X_{D,D'}\neq0$: 1) $D'$ does not have an outer arc of size $m_i$,
2) $D'$ does not have an outer arc of size $m'_i$, 3) $D'$ has $n_{I+1}+1$ up arrows and 
$n'_{J+1}-1$ down arrows instead of $n_{I+1}$ up arrows and $n_{J+1}$ down arrows,  
4) $D'$ has $n_{I+1}-1$ up arrows and $n'_{J+1}+1$ down arrows instead of $n_{I+1}$ up 
arrows and $n_{J+1}$ down arrows, and 5) $D'=D$.

In the first case, we have $X_{D,D'}=[1+\sum_{j=1}^{i}n_{j}]$. The contribution to the left 
hand side of Eqn.(\ref{EPA}) is 
\begin{eqnarray*}
q^{d}Q\Psi_{D}\sum_{i=1}^{I}\left[1+\sum_{j=1}^{i}n_{j}\right][m_i]
\frac{[1+N_{\uparrow}+N_{\downarrow}+M+M']}{[1+N_{\uparrow}+M]}
\frac{\prod_{j\ge i+1}^{I}[1+\sum_{k=1}^{j}(n_k+m_k)]}
{\prod_{j\ge i}^{I}[1+n_j+\sum_{k=1}^{j-1}(n_k+m_k)]}.
\end{eqnarray*}
Inserting Lemma~\ref{lemma-app0} into the above expression, 
we obtain 
\begin{eqnarray}
\label{EPA-c1}
q^{d}Q\Psi_{D}\frac{[1+N_{\uparrow}+N_{\downarrow}+M+M'][M]}{[1+M+N_{\uparrow}]}.
\end{eqnarray}
In the second case, we have 
$X_{D,D'}=q^{N_{\uparrow}-N_{\downarrow}-1}[1+\sum_{j=1}^{i}n'_{j}]$.
The contribution to the left hand side of Eqn.(\ref{EPA}) is 
\begin{eqnarray*}
q^{-d'}Q^{-1}\Psi_{D}\sum_{i=1}^{J}
\left[1+\sum_{j=1}^{i}n'_{j}\right][m'_i][1+N_{\uparrow}+d']
\frac{\prod_{j\ge i+2}[1+\sum_{k=1}^{j-1}(n'_{k-1}+m'_{k-1})]}
{\prod_{j\ge i}[1+n'_j+\sum_{k=1}^{j-1}(n'_{k-1}+m'_{k-1})]}
\end{eqnarray*}
where $d'=M+M'+N_{\downarrow}$. 
Inserting Lemma~\ref{lemma-app1} into the above expression, we obtain
\begin{eqnarray}
\label{EPA-c2}
q^{-d'}Q^{-1}\Psi_{D}
\frac{[1+N_{\uparrow}+M+N_{\downarrow}+M'][M']}{[1+N_{\downarrow}+M']}.
\end{eqnarray}
In the third case, we have $X_{D,D'}=[N_{\uparrow}+1]$. The contribution is
\begin{eqnarray}
\label{EPA-c3}
q^{d}Q\Psi_{D}\frac{[N_{\downarrow}+M'][N_{\uparrow}+1]}{[N_{\uparrow}+M+1]}.
\end{eqnarray}
In the fourth case, we have $X_{D,D'}=[N_{\downarrow}+1]$. 
The contribution is 
\begin{eqnarray}
\label{EPA-c4}
q^{-d'}Q^{-1}\Psi_{D}\frac{[N_{\uparrow}+M][N_{\downarrow}+1]}{[N_{\downarrow}+M'+1]}.
\end{eqnarray}
In the fifth case, we have $X_{D,D}=q^{N_{\uparrow}-N_{\downarrow}}[Q;0]$. 
The contribution is 
\begin{eqnarray}
\label{EPA-c5}
q^{N_{\uparrow}-N_{\downarrow}}[Q;0]\Psi_{D}.
\end{eqnarray}
Note that $N=N_{\uparrow}+N_{\downarrow}+2M+2M'$. 
We obtain the right hand side of Eqn.(\ref{EPA}) as the sum of 
Eqns.(\ref{EPA-c1}) to (\ref{EPA-c5}).
This completes the proof.
\end{proof}

\subsection{Type BI}
We consider the action of $X$ on the Kazhdan--Lusztig basis of type BI.

Let $D$ be a diagram of type BI and $N_{\uparrow}$ be the number of 
unpaired up arrows.
Recall that $D$ consists of up arrows, arcs, at most one unpaired 
down arrow, dashed arcs and down arrows with an integer 
$p$, $1\le p\le M$. 
We enumerate the (unpaired) up arrows from left to right by 
$1,2\ldots,N_{\uparrow}$.
For each $i$, $1\le i< N_{\uparrow}$, we denote by $X_{(i)}(D)$ 
a diagram obtained from $D$ by connecting the $i$-th up arrow 
and $(i+1)$-th up arrow via an arc. 

Suppose $D$ has an unpaired down arrow. 
We denote by $X_{(N_{\uparrow})}(D)$ a diagram obtained from $D$ 
by connecting $N_{\uparrow}$-th up arrow and the unpaired down arrow
via a dashed arc.
We denote by $X_{(N_{\uparrow}+1)}(D)$ a diagram obtained from $D$ 
by changing the unpaired down arrow to an up arrow.
We define the action of $X$ on $D$ by 
\begin{eqnarray}
\label{ActionXBI-1}
X(D):=\sum_{i=1}^{N_{\uparrow}+1}[i]X_{(i)}(D).
\end{eqnarray}

Suppose $D$ does not have an unpaired down arrow.
We regard the down arrow with a star as the down arrow with the integer one.
Let $r$ be the smallest integer attached to down arrows with an integer 
$1\le p\le M$.
If there is no down arrow with an integer, we define $r=M+1$.
We denote by $X_{(N_\uparrow)}(D)$ a diagram obtained from $D$ by changing 
the $N_{\uparrow}$-th up arrow to a down arrow.
The action of $X$ on $D$ is defined by 
\begin{eqnarray}
\label{ActionXBI-2}
X(D):=\sum_{i=1}^{N_{\uparrow}}[i]X_{(i)}(D)
+(1-\delta_{1,r})[N_{\uparrow}+r-1]D.
\end{eqnarray}
where $\delta_{i,j}$ is the Delta function satisfying $\delta_{i,i}=1$ and 
$\delta_{i,j}=0$ for $i\neq j$.

\begin{example}
Let $D$ be a diagram depicted as 
\begin{eqnarray*}
D=
\tikzpic{-0.6}{
\upa{0};
\downa{0.5};
\linkpattern{1.4/2,3.2/4.0,4.8/5.6}{1/2.4}{2.8/$\bigstar$,4.4/2}{};
},
\end{eqnarray*}
where $M=2$. 
We have 
\begin{eqnarray*}
X_{(1)}=
\tikzpic{-0.6}{
\linkpattern{1.4/2,3.2/4.0,4.8/5.6}{-0.2/0.6,1/2.4}{2.8/$\bigstar$,4.4/2}{};
}, \qquad 
X_{(2)}
=
\tikzpic{-0.6}{
\upa{0};
\upa{0.5};
\linkpattern{1.4/2,3.2/4.0,4.8/5.6}{1/2.4}{2.8/$\bigstar$,4.4/2}{};
}.
\end{eqnarray*}
The action of $X$ on $D$ is 
\begin{eqnarray*}
X(D)=X_{(1)}(D)+[2]X_{(2)}(D).
\end{eqnarray*}
\end{example}

\begin{example}
Let $D$ be a diagram depicted as 
\begin{eqnarray*}
D=\tikzpic{-0.6}{
\upa{0};
\linkpattern{0.5/1.3,2.6/4.2,3/3.8}{}{2.2/2,4.6/3}{};
\upa{1.4}
}
\end{eqnarray*}
where $M=3$ and $r=2$. 
We have 
\begin{eqnarray*}
X_{(1)}=
\tikzpic{-0.6}{
\linkpattern{0.2/1.6,0.5/1.3,2.6/4.2,3/3.8}{}{2.2/2,4.6/3}{};
}, \qquad 
X_{(2)}=
\tikzpic{-0.6}{
\upa{-0.2};
\linkpattern{0.2/1.6,0.5/1.3,2.8/4.4,3.2/4}{}{2/$\bigstar$,2.4/2,4.8/3}{};
}.
\end{eqnarray*}
The action of $X$ on $D$ is 
\begin{eqnarray*}
X(D)=X_{(1)}(D)+[2]X_{(2)}(D)+[3]D.
\end{eqnarray*}
\end{example}

\begin{theorem}
\label{thm-ActionXBI}
The above definitions (\ref{ActionXBI-1}) and (\ref{ActionXBI-2})  
provides the action of $X$ on the Kazhdan--Lusztig basis of type BI.
\end{theorem}
\begin{proof}
We prove Theorem by induction.
When $N=1$, Theorem is true by a straightforward calculation.
We assume that Theorem holds true up to $N-1\ge1$.
Let $D$ be a diagram of length $N$.
We have two cases for the leftmost arrow of $D$: 
1) an up arrow and 2) a down arrow.

\paragraph{\bf Case 1}
In this case, a diagram $D$ is written as $D=\uparrow D'$. 
By using the comultiplication, we have 
\begin{eqnarray}
\label{eqnBID-1}
X(\uparrow D')=q\uparrow X(D')+\downarrow D'.
\end{eqnarray}
We have two cases for $D'$: a) $D'$ has an unpaired down arrow and 
b) $D'$ does not have an unpaired down arrow.
\paragraph{Case 1-a}

Inserting Eqn.(\ref{ActionXBI-1}) and 
$\downarrow D'=\sum_{i=1}^{N_{\uparrow}+1}q^{-(i-1)}X_{(i)}(D)$
into Eqn.(\ref{eqnBID-1}), we obtain 
\begin{eqnarray*}
X(\uparrow D')
&=&q\sum_{i=2}^{N_{\uparrow}+1}[i-1]X_{(i)}(D)
+\sum_{i=1}^{N_{\uparrow}+1}q^{-(i-1)}X_{(i)}(D) \\
&=&\sum_{i=1}^{N_{\uparrow}+1}[i]X_{(i)}(D),	
\end{eqnarray*}
where we have used $q[i-1]+q^{-(i-1)}=[i]$.

\paragraph{Case 1-b}
We have 
\begin{eqnarray}
\label{eqnBID-2}
\downarrow D'
=\sum_{i=1}^{N_{\uparrow}}q^{-(i-1)}X_{(i)}(D)
+(1-\delta_{1,r})q^{N_{\uparrow}+r-2}D.
\end{eqnarray}
Inserting Eqns.(\ref{ActionXBI-2}) and (\ref{eqnBID-2})
into Eqn.(\ref{eqnBID-1}), we obtain 
\begin{eqnarray*}
X(\uparrow D')&=&q\sum_{i=2}^{N_{\uparrow}}[i-1]X_{(i)}(D)
+q(1-\delta_{1,r})[N_{\uparrow}+r-2]D 
+\sum_{i=1}^{N_{\uparrow}}q^{-(i-1)}X_{(i)}(D) \\
&&\qquad+(1-\delta_{1,r})q^{N_{\uparrow}+r-2}D \\
&=&\sum_{i=1}^{N_{\uparrow}}[i]X_{(i)}(D)
+(1-\delta_{1,r})[N_{\uparrow}+r-1]D.
\end{eqnarray*}

\paragraph{\bf Case 2}
We have four cases for the leftmost down arrow $a$ of $D$: 
a) the arrow $a$ is an unpaired arrow, 
b) the arrow $a$ forms a dashed arc,
c) the arrow $a$ is a down arrow with the integer $r$, $1\le r\le M$,
and d) the arrow $a$ forms an arc.
\paragraph{Case 2-a}
The diagram $D$ is written as $D=\downarrow D'$.
We have $X(\downarrow D')=q^{-1}\downarrow X(D')+\uparrow D'$. 
Since the diagram $D'$ has no unpaired up arrows and $r=1$, 
we obtain $X(D')=0$ by using Eqn.(\ref{ActionXBI-2}).
Therefore, we have $X(\downarrow D')=\uparrow D'$.

\paragraph{Case 2-b}
Let $E$ be a diagram obtained from $D$ by removing arcs and 
$E'=X(E)$. 
From the definition of the action of $X$, the action of $X$
on $D$ is obtained by inserting the removed arcs of $D$ into 
diagrams of $E'$ at the same position as $D$.
Thus, without loss of generality, the diagram $D$ is written 
as 
$D=
\raisebox{-0.5\totalheight}{
\begin{tikzpicture}
\draw[dashed,very thick](0,0)..controls(0,-0.6)and(0.6,-0.6)..(0.6,0);
\end{tikzpicture}
}D'$
where $D'$ is a diagram of length $N-2$. 
We have 
\begin{eqnarray*}
X(D)&=&X(\downarrow\downarrow D')-q^{-1}X(\uparrow\uparrow D') \\
&=&q^{-1}\downarrow X(\downarrow D')+\uparrow\downarrow D'
-\uparrow X(\uparrow D')-q^{-1}\downarrow\uparrow D' \\
&=&0,
\end{eqnarray*}
where we have used $X(\uparrow D')=\downarrow D'$ and 
$X(\downarrow D')=\uparrow D'$.

\paragraph{Case 2-c}
The diagram $D$ is graphically written as 
$D= \!\!\!
\raisebox{-0.75\totalheight}{
\begin{tikzpicture}
\draw(0,0)--(0,-0.6)(0,-0.8)node{$r$};
\end{tikzpicture}
}\!\!\!
D'$ 
where $D'$ is a diagram of 
length $N-1$.
We have 
\begin{eqnarray*}
X(D)&=&X(\downarrow D')-q^{-r}X(\uparrow D') \\
&=&q^{-1}\downarrow X(D')+\uparrow D'
-q^{-r+1}\uparrow X(D')-q^{-r}\downarrow D' \\
&=&[r-1]\downarrow D'-q^{-r}[r-1]\uparrow D' \\
&=&[r-1]D,
\end{eqnarray*}
where we have used $X(D')=[r]D'$.

\paragraph{Case 2-d}
By a similar argument to Case 2-b, without loss of generality,
we assume that the diagram $D$ is written as 
$D=
\raisebox{-0.5\totalheight}{
\begin{tikzpicture}
\draw(0,0)..controls(0,-0.6)and(0.6,-0.6)..(0.6,0);
\end{tikzpicture}
}D'$ 
where $D'$ is a diagram of length $N-2$. 
We have 
\begin{eqnarray*}
X(D)&=&X(\downarrow\uparrow D')-q^{-1}X(\uparrow\downarrow D') \\
&=&q^{-1}\downarrow X(\uparrow D')+\uparrow\uparrow D'
-\uparrow X(\downarrow D')-q^{-1}\downarrow\downarrow D' \\
&=&\downarrow\uparrow X(D')+q^{-1}\downarrow\downarrow D'
+\uparrow\uparrow D'-q^{-1}\uparrow\downarrow X(D')
-\uparrow\uparrow D'-q^{-1}\downarrow\downarrow D'\\
&=&\raisebox{-0.5\totalheight}{
\begin{tikzpicture}
\draw(0,0)..controls(0,-0.6)and(0.6,-0.6)..(0.6,0);
\end{tikzpicture}
}X(D').
\end{eqnarray*}

In both Case 1 and 2, $X(D)$ coincides with the definitions 
(\ref{ActionXBI-1}) and (\ref{ActionXBI-2}).
This completes the proof.
\end{proof}

Let $N_{\uparrow}$ be the number of up arrows in $D$. 
Let $r$ be the smallest integer attached to down arrows with an integer 
$1\le p\le M$.
If there is no down arrow with an integer, we define $r=M+1$.
We define an integer $E_{D}$ as follows:
\begin{enumerate}
\item If $D$ has an unpaired down arrow, $E_{D}=-(N_{\uparrow}+1)$.
\item If $D$ does not have an unpaired down arrow, $E_{D}=N_{\uparrow}+r-1$.
\end{enumerate}
Note that $|E_{D}|$ is the maximum integer which appears in the expansion of 
$X(D)$.
We denote by $\mathcal{D}$ the set of diagrams of length $N$.
For an integer $i\in\mathbb{Z}$, we define 
\begin{eqnarray*}
Z_i:=\#\{E_{D}|E_{D}=i \text{ and } D\in\mathcal{D}\}.
\end{eqnarray*}

\begin{theorem}
\label{BIev}
$X$ has an eigenvalue $[N+M-2i]$, $0\le i\le N$, of multiplicity 
$\displaystyle\genfrac{(}{)}{0pt}{}{N}{i}$.
\end{theorem}
We omit the proof since one can apply the same method 
as \cite[Theorem 6.11]{Shi14-2}. 
As a corollary, we have $Z_{N+M-2i}=\genfrac{(}{)}{0pt}{}{N}{i}$.
For each eigenvalue $[N+M-2i]$, an eigenvector is characterized by 
a diagram $D$ with $E_D=N+M-2i$.
See \cite{Shi14-2} for $M=1$ case.

Let $D$ be a diagram of Type BI, $N_{\uparrow}$ be the number of
up arrows (excluding up arrows forming arcs), $N_1$ be the number 
of the unpaired down arrow ($N_1$ is either $0$ or $1$).
Let $S$ be the set of all arcs of $D$.
If $D$ has the down arrow with the integer $M$, $S_R$ is 
defined as the set of arcs right to the down arrow with 
the integer $M$. Otherwise, $S_R$ is the empty set.
If $D$ has the down arrow with a star, $S_{W'}$ is defined 
as the set of arcs which are left to the down arrow with 
the integer $M$ and right to the down arrow with a star.
Otherwise, $S_{W'}$ is the empty set.
If $D$ has the down arrow with a star, $S_W$ is defined
as the set of arcs which are left to the down arrow with 
a star and either right to the unpaired down arrow for $N_1=1$
or right to the leftmost down arrow forming a dashed arc
for $N_1=0$. Otherwise, $S_W$ is the empty set.
If $D$ has the down arrow with a star and up arrows,
$S_L$ is defined as the set of arcs which are left to 
the leftmost down arrow (which is an unpaired down arrow for $N_1=1$,  
the leftmost down arrow forming a dashed arc for $N_1=0$, or 
the down arrow with a star for $N_1=0$ and $D$ without dashed arcs)  
and right to the rightmost up arrow.
If $D$ has the down arrow with a star but no up arrows, 
$S_L$ is defined as the set of arcs which are left to the 
leftmost down arrow. Otherwise, $S_L$ is the empty set.
An arc $A$ is called an {\it outer arc} if there are no 
arcs and no dashed arc outside of $A$. 
We denote the set of outer arc by $S^{+}$.
We define $S_R^+:=S^{+}\cap S_R$, $S_W^{+}:=S^{+}\cap S_W$ 
and $S_L^{+}:=S^{+}\cap S_L$.

Let $T$ be the set of dashed arcs, $U$ be the set of 
down arrows with integers $p, 2\le p\le M$. 
We define $T'$ as the set of dashed arcs except the leftmost 
one and $U'$ 
as the set of down arrows with integers $p, 2\le p\le M-1$.
Then, $V$ (resp. $V'$) is given by the union of $U$ (resp. $U'$) 
and the down arrow with a star if it exists.

We define the following values:
\begin{eqnarray*}
N_{2}&:=&N_{\uparrow}+N_1+|S|+|T|, \\
N_{3}&:=&|S_W|+|T|, \\
N_{4}&:=&|S_{W'}|+|S_W|+|T|+M   \\
N_{5}&:=&N-|S|+|S_L|+|S_W|+|S_{W'}|+|S_R|+M   \\
N_{6}&:=&
\begin{cases}
[N_5]/[N_4], & \text{for $|V|=M$ and $N_1=0$}, \\
1, & \text{otherwise}.
\end{cases}
\end{eqnarray*}

We enumerate all arrows from left to right. 
Let $s_1$ be the integer assigned to the down arrow with a star 
and $s_p$ be the integer assigned to the down arrow with the 
integer $p, 2\le p\le M$.
If $i$-th down arrow and $j$-th ($i<j$) up arrow forms an arc $A$,	
then the size of $A$ is $(j-i+1)/2$ and denoted by $m_A$.	
Similarly, if $k$-th down arrow and $l$-th ($k<l$) down arrow forms
a dashed arc $B$, then the size of $B$ is $(l-k+1)/2$.
Let $C$ be the down arrow with the integer $p$ or the down arrow with
a star. Let $E$ be a dashed arc.
We define 
\begin{eqnarray*}
d_{1,A,C}&:=&
(i-s_p+M-p+1)/2, \qquad \text{for $2\le p\le M$},  \\
d_{2,A}&:=&(i-s_1+M)/2 \\
d_{3,A}&:=&N-j, \\
N_{7}&:=&\prod_{C\in U}\prod_{A\in S_R^{+}}
\frac{[d_{1,A,C}]}{[d_{1,A,C}+m_A]}, \\
N_{8}&:=&
\begin{cases}
\displaystyle
\prod_{A\in S_{R}^{+}}\prod_{i=0}^{N_{3}}
\frac{[d_{2,A}+i]}{[d_{2,A}+m_A+i]}
\prod_{A'\in S_W}\frac{[d_{2,A}+m_A+h_{A'}]}{[d_{2,A}+h_{A'}]}, 
& \text{for $|V|=M$}, \\
1 & \text{otherwise},
\end{cases} \\
N_{9}&:=&
\begin{cases}
\displaystyle
\prod_{A\in S_{W}^{+}}\frac{[d_{3,A}+m_A+M]}{[d_{3,A}+2m_A+M]} 
& \text{for $N_1=1$},\\
\displaystyle
\prod_{A\in S_{W}^{+}\cup S_{L}^{+}}
\frac{[d_{3,A}+m_A+M]}{[d_{3,A}+2m_A+M]} 
& \text{for $N_1=0$},
\end{cases}
\end{eqnarray*}
where $h_{A'}, A'\in S_{W}$ is the sum of the number of arcs in 
$S_{W}$ right to $A'$ or outside of $A'$ (including $A'$), and 
the number of dashed arcs right to $A'$. 
We also define 
\begin{eqnarray*}
d_{4,C}&:=&(s_M-s_p+M-p)/2+1, \\
d_{5,E}&:=&(s_M-k+M+1)/2, \\
N_{10}&:=&
\begin{cases}
\displaystyle
\prod_{C\in U'}[d_{4,C}]^{-1}, 
& \text{for $N_1=0$ and $T=\emptyset$} \\
\displaystyle
\prod_{C\in V'}[d_{4,C}]^{-1}\prod_{E\in T'}[d_{5,E}]^{-1}, 
& \text{for $N_1=0$ and $T\neq\emptyset$}, \\
\displaystyle
\prod_{C\in V'}[d_{4,C}]^{-1}\prod_{E\in T}[d_{5,E}]^{-1}
& \text{for $N_1=1$}, \\
1 & \text{otherwise},
\end{cases}
\end{eqnarray*}

We enumerate up arrows, an unpaired down arrow (if it exists),
arcs, dashed arcs and down arrows with the integer $p$, $2\le p\le M$ 
from left to right. 
If there exist arcs inside of an arc or a dashed arc, we increase the 
integer one by one from outside to inside.
Let $N_A$ be the integer assigned to an arc, a dashed arc or a down arrow
with the integer $p$, and $N_\downarrow$ be the integer assigned to the 
unpaired down arrow.
We define
\begin{eqnarray*}
N_{11}:=
\begin{cases}
\displaystyle
\prod_{A\in S\cup T\cup U}[N_A], & N_1=0, \\
\displaystyle
[N_{\downarrow}]\cdot\prod_{A\in S\cup T\cup U}[N_A], & N_1=1. 
\end{cases}
\end{eqnarray*}

We enumerate arcs, dashed arcs, down arrows with the integer 
$p, 2\le p\le M$, and the down with a star from right to left.
If there exists arcs inside of an arc or a dashed arc, 
we increase the integer one by one from inside to outside. 
Let $N_B$ be the integer assigned to a dashed arc or a down 
arrow with a star.
Then, we define
\begin{eqnarray*}
N_{12}:=\frac{\displaystyle\prod_{i=1}^{N_{2}}(q^{i+M-1}+q^{-(i+M-1)})}
{\displaystyle\prod_{B\in T\cup (V\setminus U)}(q^{N_B}+q^{-N_B})}.
\end{eqnarray*}

\begin{defn}
\label{PsiBI}
\begin{eqnarray*}
\Psi_D:=\prod_{A\in S}[m_A]^{-1}\cdot N_6\cdot N_7\cdot
N_8\cdot N_9\cdot N_{10}\cdot N_{11}\cdot N_{12}. 
\end{eqnarray*}
\end{defn}

\begin{example}
Let $D$ be a diagram depicted as 
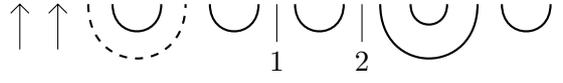
\begin{eqnarray*}
\tikzpic{0}{
\draw(0,0)--(0,-0.6)(-0.12,-0.12)--(0,0)--(0.12,-0.12);
\draw(-0.5,0)--(-0.5,-0.6)(-0.12-0.5,-0.12)--(-0.5,0)--(0.12-0.5,-0.12);
\linkpattern{0.9/1.7,2.5/3.3,3.9/4.7,5.3/6.9,5.8/6.4,7.3/8.1}
{0.5/2.1}{3.6/1,5/2}{};
}.
\end{eqnarray*}
We have 
\begin{eqnarray*}
&&N_1=0, \qquad N_{6}=[20]/[6],\qquad N_7=1/[4], \\
&&N_{8}=[3]/[9],\qquad N_9=[13]/[14], \qquad N_{10}=1/[3], \\
&&N_{11}=[10]!/[2],\qquad N_{12}=\prod_{i\in I}(q^{i+1}+q^{-i-1}),\\
&&\prod_{A\in S}[m_A]^{-1}=[2]^{-1},
\end{eqnarray*}
where $I=$\{1,2,3,4,6,7,9\}.
\end{example}

\begin{theorem}
$\Psi$ is the eigenvector of $X$ with the eigenvalue $[N+M]$.
The multiplicity is one.
\end{theorem}
\begin{proof}
From Theorem~\ref{BIev}, the multiplicity of the eigenvalue $[N+M]$ 
is obviously one.
Let $(X_{D,D'})$ be the matrix representation of $X$. 
We will show that 
\begin{eqnarray}
\label{EPBI-1}
\sum_{D'}X_{D,D'}\Psi_{D'}=[N+M]\Psi_{D}.
\end{eqnarray}
We have four cases for $D$: 
1) $D$ does not have down arrows,
2) $2\le r\le M$,
3) $N_1=1$ and 
4) $r=1$ and $N_1=0$.

\paragraph{\bf Case 1}
Let $D$ be a diagram starting with $n_1$ up arrows, followed by 
an outer arc of size $m_1$, followed by $n_2$ up arrows, followed 
by an outer arc of size $m_2$, $\cdots$ and ending with $n_{I+1}$ 
up arrows. Inside of an outer arc of size $m_i$, $1\le i\le I$, is 
filled with arcs.
Set $M'=\sum_{i=1}^{I}m_i$ and $N_{\uparrow}=\sum_{i=1}^{I+1}n_i$.
Let $D'$ be a diagram obtained from $D$ by changing the outer arc 
of size $m_i$ to two up arrows. 
We have $X_{D,D'}=[1+\sum_{j=1}^{i}n_j]$.
The contribution of these $D'$'s to the left hand side 
of Eqn.(\ref{EPBI-1}) is 
\begin{eqnarray*}
\sum_{i=1}^{I}\left[1+\sum_{j=1}^{i}n_j\right][m_i]
\frac{\prod_{j\ge i+1}[1+\sum_{k=1}^{j}(n_k+m_k)]}
{\prod_{j\ge i}[1+n_j+\sum_{k=1}^{j-1}(n_k+m_k)]}
(q^{N_{\uparrow}+M'+M}+q^{-(N_{\uparrow}+M'+M)})\Psi_D.
\end{eqnarray*}
Inserting Lemma~\ref{lemma-app0}, 
we obtain $[M'](q^{N_{\uparrow}+M'+M}+q^{-(N_{\uparrow}+M'+M)})\Psi_D$.

From Eqn.(\ref{ActionXBI-2}), the contribution of the diagonal term 
is $[N_{\uparrow}+M]\Psi_D$.
Note $N=N_{\uparrow}+2M'$. 
The sum of these two contributions is the right hand side of 
Eqn.(\ref{EPBI-1}).

\paragraph{\bf Case 2}
In this case, we have $S_R=S_W=S_L=\emptyset$.
Let $D$ be a diagram starting with $n_1$ up arrows, followed by 
an outer arc of size $m_1$, followed by $n_2$ up arrows, followed 
by an outer arc of size $m_2$, $\cdots$, followed by $n_{I+1}$ up 
arrows, followed by a down arrow with the integer $r$ and ending 
with $G$ arcs, $|T|$ dashed arcs and down arrows with the 
integer $p$, $r+1\le p\le M$.
Set $M'=\sum_{i=1}^{I}m_i$ and $N_{\uparrow}=\sum_{i=1}^{I+1}n_i$.
Let $D'$ be a diagram obtained from $D$ by changing the outer arc 
of size $m_i$ to two up arrows.
By a similar argument to Case 1, the contribution of these $D'$'s to 
the left hand side of Eqn.(\ref{EPBI-1}) is 
\begin{eqnarray*}
\frac{[M'][1+N_{\uparrow}+M'+G+M-r+1]}{[1+M'+N_{\uparrow}]}
(q^{N_{\uparrow}+M'+G+M}+q^{-(N_{\uparrow}+M'+G+M)})\Psi_{D}.
\end{eqnarray*}
Let $D'$ be a diagram obtained form $D$ by changing the down arrow 
with the integer $r$ to an up arrow.
We have $X_{D,D'}=[N_{\uparrow}+1]$. 
The contribution of this diagram is 
\begin{eqnarray*}
[N_{\uparrow}+1]\frac{[G+M-r+1]}{[1+N_{\uparrow}+M']}
(q^{N_{\uparrow}+M'+G+M}+q^{-(N_{\uparrow}+M'+G+M)})\Psi_{D}.
\end{eqnarray*}
The contribution of the diagonal term is $[N_{\uparrow}+r-1]\Psi_{D}$.
Note that $N=N_{\uparrow}+2M'+2G+M-r+1$. 
The sum of these three contributions is the right hand side of 
Eqn.(\ref{EPBI-1}).

\paragraph{\bf Case 3}
Let $D$ be diagram starting with $n_1$ up arrows, followed by 
an outer arc of size $m_1$, followed by $n_2$ up arrows, followed 
by an outer arc of size $m_2$, $\cdots$, followed by $n_{I+1}$ up 
arrows, followed by an unpaired down arrow and ending with 
$|S_W|+|S_{W'}|+|S_R|$ arcs, $|T|$ dashed arcs, a down 
arrow with a star and down arrows with an integer $p$, $2\le p\le M$.
Set $M'=\sum_{i=1}^{I}m_i$ and $N_{\uparrow}=\sum_{i=1}^{I+1}n_i$.
Let $D'$ be a diagram obtained from $D$ by changing the outer arc 
of size $m_i$ to two up arrows.
By a similar argument to Case 1, the contribution of these $D'$'s 
to the left hand side of (\ref{EPBI-1}) is 
\begin{eqnarray*}
[M']\frac{[1+d]}{[1+M'+N_{\uparrow}]}(q^{d+1}+q^{-(d+1)}),
\end{eqnarray*}
where $d=N_{\uparrow}+|S|+|T|+M$.
Let $D'$ be a diagram obtained from $D$ by changing the unpaired 
down arrow to an up arrow.
We have $X_{D,D'}=[N_{\uparrow}+1]$.
The contribution is 
\begin{eqnarray*}
[N_{\uparrow}+1]\frac{[1+N_{\uparrow}+2|S|-M'+2|T|+2M]}{[1+N_{\uparrow}+M']}.
\end{eqnarray*}
Note that $N=N_{\uparrow}+2|S|+2|T|+M+1$. 
The sum of two contributions gives the right hand side of Eqn.(\ref{EPBI-1}).

\paragraph{\bf Case 4}
Let $D$ be diagram starting with $n_1$ up arrows, followed by 
an outer arc of size $m_1$, followed by $n_2$ up arrows, followed 
by an outer arc of size $m_2$, $\cdots$, followed by $n_{I+1}$ up 
arrows and ending with 
$|S_W|+|S_{W'}|+|S_R|$ arcs, $|T|$ dashed arcs, a down 
arrow with a star and down arrows with an integer $p$, $2\le p\le M$.
We have two cases for $D$: a) $n_{I+1}\neq0$, that is, $|S_{L}|=\emptyset$ 
and b) for a given $H$ satisfying $1\le H\le I$, $n_{H+1}\neq0$ and 
$n_{i}=0$ for $H+2\le i\le I+1$.  
Set $N_{\uparrow}=\sum_{i=1}^{I+1}n_i$, $M'=\sum_{i=1}^{I}m_i$ 
and $d=N_{\uparrow}+|S|+|T|+M$. 

\paragraph{Case 4-a}
Set $M'=\sum_{i=1}^{I}m_i$ and $L_1:=N_{\uparrow}+2|S|-M'+2|T|+2M$. 
Let $D'$ be a diagram obtained form $D$ by changing the outer arc of 
size $m_i$, $1\le i\le I$ to two up arrows. 
We have $X_{D,D'}=[1+\sum_{j=1}^{i}n_j]$.
By a similar argument to Case 1, the contribution to the left hand side of 
(\ref{EPBI-1}) is 
\begin{eqnarray}
\label{BI-cont4a1}
[M']\frac{[d][L_1+1]}{[N_{\uparrow}+M'+1][L_1]}(q^d+q^{-d})\Psi_{D}.
\end{eqnarray}
Let $D'$ be a diagram obtained from $D$ by changing the leftmost dashed 
arc to an up arrow and an unpaired down arrow.
We have $X_{D,D'}=[N_{\uparrow}+1]$.
The contribution is 
\begin{eqnarray}
\label{BI-cont4a2}
[N_{\uparrow}+1]\frac{[|S_{W'}|+|S_{W}|+|T|+|S_R|+M][d]}{[L_1][1+N_{\uparrow}+M']}
(q^d+q^{-d})(q^{d'}+q^{-d'})\Psi_{D},
\end{eqnarray}
where $d'=d-N_{\uparrow}-M'$. 
Let $D'$ be a diagram obtained from $D$ by changing the rightmost up arrow to
an unpaired down arrow. We have $X_{D,D'}=[N_{\uparrow}]$.
The contribution is 
\begin{eqnarray}
\label{BI-cont4a3}
[N_{\uparrow}]\frac{[N_{\uparrow}+M']}{[L_1]}\Psi_D.
\end{eqnarray}
Note $N=N_{\uparrow}+2|S|+2|T|+M$. 
The sum of three contributions (\ref{BI-cont4a1}), (\ref{BI-cont4a2}) and (\ref{BI-cont4a3})
gives the right hand side of Eqn.(\ref{EPBI-1}).

\paragraph{Case 4-b}
We have $S_{L}\neq\emptyset$ and $|S_L|=\sum_{i=H+1}^{I}m_i$.
Set $M'=\sum_{i=1}^{H}m_i$ and $L_1:=N_{\uparrow}+2|S|-M'+2|T|+2M$. 
Let $D'$ be a diagram obtained from $D$ by changing the outer arc of size 
$m_i$, $1\le i\le H$, to two up arrows.
By a similar argument to Case 1, the contribution of these $D'$'s to the 
left hand side of (\ref{EPBI-1}) is 
\begin{eqnarray}
\label{BI-cont41}
[M']\frac{[L_1+1]}{[L_1]}\frac{[d]}{[1+N_{\uparrow}+M']}
(q^{d}+q^{-d})\Psi_{D}.
\end{eqnarray}
Let $D'$ be a diagram obtained from $D$ by changing the outer arc of 
size $m_i$, $H+1\le i\le I$, to two up arrows.
We have $X_{D,D'}=[N_{\uparrow}+1]$.
Thus, the contribution of these $D'$'s is 
\begin{eqnarray}
\label{BI-EP-41}
\Psi_{D}\frac{[N_{\uparrow}+1][d]}{[1+N_{\uparrow}+M'+|S_L|][L_1]}(q^d+q^{-d})
\sum_{i=H+1}^{I}[m_i] [l_i]
\frac{\prod_{j\ge i+1}[w_{j}]}{\prod_{j\ge i}[w_{j-1}]}
\prod_{j=H+1}^{i}\frac{[g_j]}{[g_j-m_j]} 
\end{eqnarray}
where 
\begin{eqnarray*}
l_i&:=&1+N_{\uparrow}+2|S|-M'+2|T|+2M-\sum_{j=H+1}^{i}m_j,  \\
w_j&:=&1+N_{\uparrow}+M'+\sum_{k=H+1}^{j}m_k, \\
g_i&:=&2|S_R|+2|S_{W'}|+2|S_W|+2M+2|T|+2\sum_{j=i}^{I}m_j.
\end{eqnarray*}
Let $D'$ be a diagram obtained from $D$ by changing the leftmost 
dashed arc to an up arrow and a down arrow. 
The contribution of $D'$ is 
\begin{eqnarray}
\label{BI-EP-42}
\Psi_{D}[N_{\uparrow}+1]\frac{[2L_2][N_{\uparrow}+|S|+|T|+M]}
{[L_1][1+N_{\uparrow}+M'+|S_L|]}(q^d+q^{-d})
\prod_{i=H+1}^{I}\frac{[g_i]}{[g_i-m_i]},
\end{eqnarray}
where $L_2:=|S_{W'}|+|S_{W}|+|T|+|S_R|+M$.
By Lemma~\ref{lemma-app-13} with $x=N_{\uparrow}+M'$ and 
$z=|S_{W'}|+|S_{W}|+|S_{R}|+|T|+M-1$,  
the sum of contributions (\ref{BI-EP-41}) and 
(\ref{BI-EP-42}) is given by
\begin{eqnarray}
\label{BI-cont42}
\frac{[N_{\uparrow}+1][N_{\uparrow}+|S|+|T|+M][2L_2+2|S_L|]}
{[L_1][1+N_{\uparrow}+M']}(q^d+q^{-d})\Psi_D.
\end{eqnarray}
Finally, let $D'$ be a diagram obtained from $D$ by changing 
the rightmost up arrow to an unpaired down arrow.
The contribution of this $D'$ is 
\begin{eqnarray}
\label{BI-cont43}
\frac{[N_{\uparrow}][N_{\uparrow}+M']}{[L_1]}\Psi_D.
\end{eqnarray}
The sum of the contributions (\ref{BI-cont41}), (\ref{BI-cont42}) and 
(\ref{BI-cont43}) gives the right hand side of Eqn.(\ref{EPBI-1}).
This completes the proof.
\end{proof}

\subsection{Type BII}
We consider the action of $X$ on the Kazhdan--Lusztig basis of type BII.

Let $D$ be a diagram of type BII and $N_{\uparrow}$ be the number of 
(unpaired) up arrows.
we enumerate the up arrows from left to right by $1,2,\ldots,N_{\uparrow}$.
For each $i$, $1\le i<N_{\uparrow}$, we denote by $X_{(i)}(D)$ a diagram 
obtained from $D$ by connecting the $i$-th up arrow and $(i+1)$-th up arrow 
via an arc.
We denote by $X_{(N_{\uparrow})}(D)$ a diagram obtained from $D$ by changing 
the $N_\uparrow$-th up arrow to an e-unpaired or an o-unpaired down arrow.
Suppose that $D$ does not have down arrows or the leftmost down arrow 
of $D$ is an e-unpaired down arrow. 
We define the action of $X$ by 
\begin{eqnarray}
\label{ActionBII-1}
X(D):=\sum_{i=1}^{N_{\uparrow}}[i]X_{(i)}(D)+[Q;N_{\uparrow}]D.
\end{eqnarray}
Suppose that the leftmost down arrow of $D$ is an o-unpaired down arrow.
Then, we define the action of $X$ by
\begin{eqnarray}
\label{ActionBII-2}
X(D):=\sum_{i=1}^{N_{\uparrow}}[i]X_{(i)}(D)+[Q;-N_{\uparrow}-1]D.
\end{eqnarray}

\begin{example}
Let $D$ be a diagram depicted as
\begin{eqnarray*}
D=\tikzpic{-0.6}{
\upa{0};
\upa{0.5};
\linkpattern{1.8/3.4,2.2/3,4.2/5}{}{1.2/e,3.8/o}{};
}.
\end{eqnarray*} 
We have 
\begin{eqnarray*}
X_{(1)}
=
\tikzpic{-0.6}{
\linkpattern{0/0.8,1.5/3.1,1.9/2.7,3.7/4.5}{}{1.2/e,3.4/o}{};
}, \qquad
X_{(2)}
=
\tikzpic{-0.6}{
\upa{0.14};
\linkpattern{1.5/3.1,1.9/2.7,3.7/4.5}{}{0.7/o,1.2/e,3.4/o}{};
}.
\end{eqnarray*}
The action of $X$ on $D$ is 
\begin{eqnarray*}
X(D)=X_{(1)}(D)+[2]X_{(2)}(D)+[Q;2]D.
\end{eqnarray*}

\end{example}

\begin{theorem}
The action of $X$ defined in Eqns.(\ref{ActionBII-1}) and (\ref{ActionBII-2})
provides the action on the Kazhdan--Lusztig bases of type BII.
\end{theorem}
\begin{proof}
We prove Theorem by induction. 
When $N=1$, Theorem is true by a straightforward calculation.
We assume that Theorem holds true for diagrams 
of length up to $N-1$.
Let $D$ be a diagram of length $N$.
We have two cases for the leftmost arrow $a$ of $D$: 
1) the arrow $a$ is an up arrow and 
2) the arrow $a$ is a down arrow.

\paragraph{\bf Case 1}
A diagram $D$ is written as $\uparrow D'$ where $D'$ is a diagram of length $N-1$.
The action of $X$ on $D$ is 
\begin{eqnarray}
X(\uparrow D')=q\uparrow X(D')+\downarrow D'.
\end{eqnarray}
We have two cases for the leftmost down arrow $a'$ of $D'$: 
a) the arrow $a'$ is an e-unpaired down arrow and 
b) the arrow $a'$ is an o-unpaired down arrow.

\paragraph{Case 1-a}
From Eqn.(\ref{ActionBII-1}) for $D'$, we have 
\begin{eqnarray}
\label{BII-1-1}
q\uparrow X(D')&=&
\sum_{i=2}^{N_{\uparrow}}q[i-1]X_{(i)}(D)+q[Q;N_{\uparrow}-1]D, \\
\label{BII-1-2}
\downarrow D'&=&
\sum_{i=1}^{N_{\uparrow}}q^{-(i-1)}X_{(i)}(D)
+q^{-(N_{\uparrow}-1)}Q^{-1}D.
\end{eqnarray}
From $q[i-1]+q^{-(i-1)}=[i]$ and 
$q[Q;N_{\uparrow}-1]+q^{-(N_{\uparrow}-1)}Q^{-1}=[Q;N_{\uparrow}]$,
the sum of Eqns.(\ref{BII-1-1}) and (\ref{BII-1-2}) is equal to the 
right hand side of Eqn.(\ref{ActionBII-1}).

\paragraph{Case 1-b}
Similarly, we have 
\begin{eqnarray*}
q\uparrow X(D')
&=&\sum_{i=2}^{N_{\uparrow}}q[i-1]X_{(i)}(D)+q[Q;-N_{\uparrow}]D, \\
\downarrow D'&=&
\sum_{i=1}^{N_{\uparrow}}q^{-(i-1)}X_{(i)}(D)-q^{-N_{\uparrow}}QD
\end{eqnarray*}
From $q[Q;-N_{\uparrow}]-q^{-N_{\uparrow}}Q=[Q;-N_{\uparrow}-1]$, 
the sum of $q\uparrow X(D')$ and $\downarrow D'$ is equal to the 
right hand side of Eqn.(\ref{ActionBII-2}).

\paragraph{\bf Case 2}
We have three cases for the leftmost down arrow $a$ of $D$:
a) the arrow $a$ is an e-unpaired down arrow,
b) the arrow $a$ is an o-unpaired down arrow and 
c) the arrow $a$ forms an arc.

\paragraph{Case 2-a}
The diagram $D$ is written as 
$\raisebox{-0.75\totalheight}{
\begin{tikzpicture}
\draw(0,0)--(0,-0.5)(0,-0.7)node{e};
\end{tikzpicture}
}\!\!\! D'$
where $D'$ is a diagram of length $N-1$.
We have 
\begin{eqnarray*}
X(D)&=&X(\downarrow D')+q^{-1}QX(\uparrow D') \\
&=&q^{-1}\downarrow X(D')+\uparrow D'+Q\uparrow X(D')+q^{-1}Q\downarrow D'\\
&=&[Q;0]D
\end{eqnarray*}
where we have used $X(D')=[Q;-1]D'$.

\paragraph{Case 2-b}
The diagram $D$ is written as 
$\raisebox{-0.75\totalheight}{
\begin{tikzpicture}
\draw(0,0)--(0,-0.5)(0,-0.7)node{o};
\end{tikzpicture}
}\!\!\! D'$
where $D'$ is a diagram of length $N-1$.
We have 
\begin{eqnarray*}
X(D)&=&X(\downarrow D')-Q^{-1}X(\uparrow D') \\
&=&q^{-1}\downarrow X(D')+\uparrow D'-qQ^{-1}\uparrow X(D')-Q^{-1}\downarrow D'\\
&=&[Q;-1]D
\end{eqnarray*}
where we have used $X(D')=[Q;0]D'$.

\paragraph{Case 2-c}
By a similar argument to Case 2-b in the proof of Theorem~\ref{thm-ActionXBI},
we can assume $D$ is written as 
$D=
\raisebox{-0.5\totalheight}{
\begin{tikzpicture}
\draw(0,0)..controls(0,-0.6)and(0.6,-0.6)..(0.6,0);
\end{tikzpicture}
}D'$ where $D'$ is a diagram of length $N-2$.
We have 
\begin{eqnarray*}
X(D)&=&X(\downarrow\uparrow D')-q^{-1}X(\uparrow\downarrow D') \\
&=&q^{-1}\downarrow X(\uparrow D')+\uparrow\uparrow D'
-\uparrow X(\downarrow D')-q^{-1}\downarrow\downarrow D' \\
&=&\downarrow\uparrow X(D')+q^{-1}\downarrow\downarrow+\uparrow\uparrow D'
-q^{-1}\uparrow\downarrow X(D')-\uparrow\uparrow D'-q^{-1}\downarrow\downarrow D'\\
&=&\raisebox{-0.5\totalheight}{
\begin{tikzpicture}
\draw(0,0)..controls(0,-0.6)and(0.6,-0.6)..(0.6,0);
\end{tikzpicture}
}X(D').
\end{eqnarray*}

In both Case 1 and 2, $X(D)$ coincides with the definitions (\ref{ActionBII-1})
and (\ref{ActionBII-2}). This completes the proof.
\end{proof}

\begin{theorem}
\label{theorem-XBII}
$X$ has the eigenvalue $[Q;N-2i]$, $0\le i\le N$,  
of multiplicity $\displaystyle\genfrac{(}{)}{0pt}{}{N}{i}$.
\end{theorem}
\begin{proof}
Recall that a diagram $D$ is constructed from a binary string
$b$.
The lexicographic order of binary strings induces a natural 
lexicographic order of diagrams.
We consider the matrix representation of $X$ on 
the Kazhdan--Lusztig bases. 
In the lexicographic order of bases, $X$ is a lower triangular
matrix.
From Eqns.(\ref{ActionBII-1}) and (\ref{ActionBII-2}), 
the diagonal entries of $X$ are 
$[Q;N_{\uparrow}]$ or $[Q;-N_{\uparrow}-1]$.
Thus an eigenvalue of $X$ is of the form $[Q;n]$ with some $n$.
The multiplicity of an eigenvalue $[Q;n]$, $n\in\mathbb{N}$, 
is the number of diagrams which has $n$ up arrows and the leftmost 
down arrow is e-unpaired down arrow.
Similarly, the multiplicity of an eigenvalue $[Q;-n-1]$, $n\in\mathbb{N}$, 
is the number of diagrams which has $n$ up arrows and the leftmost 
down arrow is o-unpaired down arrow.
Set $\lambda_D=n$ (resp. $\lambda_D=-n-1$) if the leftmost down arrow 
is e-unpaired (resp. o-unpaired) down arrow.
If $D$ does not have down arrows, we set $\lambda_D=n$.

Suppose that a diagram $D$ has $m$ arcs. 
We denote by $D'$ a diagram obtained from $D$ by removing all the arcs.
The diagram $D'$ is of length $N-2m$ and without arcs.
We denote by $\mathcal{D'}$ the set of such $D'$'s. The cardinality of 
$\mathcal{D'}$ is $N-2m+1$.
Reversely, if we have a diagram $D'\in\mathcal{D'}$, one can construct 
a diagram $D$ by inserting arcs into $D'$. 
Given $D'$, the number of possible $D$'s is given
by $\displaystyle\genfrac{(}{)}{0pt}{}{N}{m}-\genfrac{(}{)}{0pt}{}{N}{m-1}$.
Obviously, we have $\lambda_{D'}=\lambda_{D}$.
Given a diagram $D'\in\mathcal{D'}$, we have $\lambda_{D'}=N-2m-2j$ with some 
$j$, $0\le j\le N-2m$. Note that there is a one-to-one correspondence between 
$j$ and $D'$. 
Thus the number of diagrams for $\lambda_{D}=N-2i$, $0\le i\le N$,   
is 
\begin{eqnarray*}
\sum_{k=0}^{\min(i,N-i)}\genfrac{(}{)}{0pt}{}{N}{k}-\genfrac{(}{)}{0pt}{}{N}{k-1}
&=&\genfrac{(}{)}{0pt}{}{N}{\min(i,N-i)} \\
&=&\genfrac{(}{)}{0pt}{}{N}{i}
\end{eqnarray*}
This completes the proof.
\end{proof}

Let $D$ be a diagram of Type BII.
We denote by $S$ the set of arcs, by $S_{\uparrow}$ 
the set of up arrows, by $S_{\mathrm{e}}$ the 
set of e-unpaired down arrows and by $S_{\mathrm{o}}$ 
the set of o-unpaired down arrows.
We enumerate arcs, up arrows and e-unpaired down arrows 
from left.
Let $N_A$ be the integer assigned to 
$A\in S\cup S_{\uparrow}\cup S_{\mathrm{e}}$. 
We define 
\begin{eqnarray*}
N_1:&=&\prod_{A\in S\cup S_{\mathrm{e}}}[N_A],\\
N_2:&=&\prod_{B\in S}[m_B]^{-1}.
\end{eqnarray*}
where $m_B$ is the size of $B\in S$.

We enumerate arcs and e-unpaired from right to left by $1, 2, \ldots$.
If there exist arcs inside of an arc or a dashed arc, we increase the 
integer one by one from inside to outside.
Let $N_C$ be the integer assigned to 
$C\in S\cup S_{\mathrm{e}}$.
We define 
\begin{eqnarray*}
N_3:=\prod_{C\in S_{\mathrm{e}}}[N_C]^{-1}.
\end{eqnarray*}
We enumerate arcs, up arrows and o-unpaired down arrows
from right.
Let $N_E$ be the integer assigned to 
$E\in S\cup S_{\uparrow}\cup S_{\mathrm{o}}$.
We define 
\begin{eqnarray*}
N_4:=\prod_{E\in S\cup S_{\uparrow}}
(q^{N_E-1}Q+q^{-(N_E-1)}Q^{-1}).
\end{eqnarray*}

In the above notation, we define
\begin{defn}
\label{PsiBII}
$\Psi_D=N_1\cdot N_2\cdot N_3\cdot N_4$.
\end{defn}

\begin{example}
Let $D$ be a diagram depicted as 
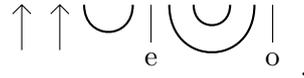
\begin{eqnarray*}
\tikzpic{0}{
\draw(0,0)--(0,-0.6)(-0.12,-0.12)--(0,0)--(0.12,-0.12);
\draw(-0.5,0)--(-0.5,-0.6)(-0.12-0.5,-0.12)--(-0.5,0)--(0.12-0.5,-0.12);
\linkpattern{0.4/1.2,1.8/3.2,2.2/2.8}{}{1.5/e,3.5/o}{};
}.
\end{eqnarray*}
We have 
\begin{eqnarray*}
N_1=\frac{[6]!}{[2]}, \qquad N_2=[2]^{-1}, \qquad
N_3=[3]^{-1},\qquad N_4=\prod_{i=2}^{6}(q^{i-1}Q+q^{-(i-1)}Q^{-1}).
\end{eqnarray*}
\end{example}

\begin{theorem}
$\Psi$ is the eigenvector of $X$ with the eigenvalue $[Q;N]$. 
The multiplicity is one.
\end{theorem}
\begin{proof}
Let $X=(X_{D,D'})$ be the matrix representation of $X$.
We show that 
\begin{eqnarray}
\label{EPBII}
\sum_{D'}X_{D,D'}\Psi_{D'}=[Q;N]\Psi_{D}.
\end{eqnarray}
We have three cases for $D$: 1) $D$ does not have down arrows,
2) the leftmost down arrow of $D$ is an e-unpaired down arrow 
and 3) the leftmost down arrow of $D$ is an o-unpaired down arrow.

\paragraph{\bf Case 1}
Let $D$ be a diagram starting with $n_1$ up arrows, followed by 
an outer arc of size $m_1$, followed by $n_2$ up arrows, followed 
by an outer arc of size $m_2$, $\cdots$ and ending with $n_{I+1}$ 
up arrows. Inside of an outer arc of size $m_i$, $1\le i\le I$, is 
filled with arcs.
As a diagram, $D$ is
\begin{eqnarray}
\label{diagram1up}
\underbrace{\uparrow\ldots\uparrow}_{n_1}\! \!
\raisebox{-0.8\totalheight}{
\begin{tikzpicture}
\draw(0,0)..controls(0,-0.8)and(1,-0.8)..(1,0);
\draw(0.5,-0.8)node{size $m_1$};
\end{tikzpicture}}\! \! 
\uparrow\ldots\uparrow \! \! \! 
\raisebox{-0.8\totalheight}{
\begin{tikzpicture}
\draw(0,0)..controls(0,-0.8)and(1,-0.8)..(1,0);
\draw(0.5,-0.8)node{size $m_I$};
\end{tikzpicture}}\! \! 
\underbrace{\uparrow\ldots\uparrow}_{n_{I+1}}.
\end{eqnarray}
Set $N_{\uparrow}=\sum_{i=1}^{I+1}n_i$ and $M=\sum_{i=1}^{I}m_i$.
The component $\Psi_D$ is explicitly written as 
\begin{eqnarray*}
\Psi_{D}=
\prod_{i=1}^{I}\frac{[\sum_{j=1}^{i}(n_j+m_j)]}{[n_i+\sum_{j=1}^{i-1}(n_j+m_j)]}
\cdot \prod_{B\in S}[m_B]^{-1} \cdot
\prod_{i=1}^{N_{\uparrow}+M}(Qq^{i-1}+Q^{-1}q^{-(i-1)}).
\end{eqnarray*}
Let $D'$ be a diagram obtained from $D$ by changing the arc of 
size $m_i$ to two up arrows. Then, we have $X_{D,D'}=[1+\sum_{j=1}^{i}n_j]$.
The contribution to the left hand side of Eqn.(\ref{EPBII}) is 
\begin{eqnarray*}
\Psi_{D}\sum_{i=1}^{I}\left[1+\sum_{j=1}^{i}n_j\right][m_i]
\frac{\prod_{j\ge i+1}^{I}[1+\sum_{k=1}^{j}(n_k+m_k)]}
{\prod_{j\ge i}^{I}[1+n_j+\sum_{k=1}^{j-1}(n_k+m_k)]}
(q^{N_{\uparrow}+M}Q+q^{-(N_{\uparrow}+M)}Q^{-1}).
\end{eqnarray*}
Inserting Lemma~\ref{lemma-app0}, the above expression is reduced to 
\begin{eqnarray*}
[M](q^{N_{\uparrow}+M}Q+q^{-(N_{\uparrow}+M)}Q^{-1})\Psi_{D}.
\end{eqnarray*}
The contribution of the diagonal term is $[Q;N_{\uparrow}]\Psi_{D}$. 
Therefore, the left hand side of Eqn.(\ref{EPBII}) is 
\begin{eqnarray*}
[M](q^{N_{\uparrow}+M}Q+q^{-(N_{\uparrow}+M)}Q^{-1})\Psi_{D}
+[Q;N_{\uparrow}]\Psi_{D}
&=&[Q;N_{\uparrow}+2M]\Psi_{D} \\
&=&[Q;N]\Psi_{D}.
\end{eqnarray*}

\paragraph{\bf Case 2}
Let $D$ be a diagram starting with $n_1$ up arrows, followed by 
an outer arc of size $m_1$, followed by $n_2$ up arrows, followed 
by an outer arc of size $m_2$, $\cdots$, followed by $n_{I+1}$ 
up arrows, followed by e-unpaired down arrow and ending with 
$G-1$ e-unpaired down arrows, $M'$ arcs and $G$ o-unpaired down 
arrows.
Set $N_{\uparrow}:=\sum_{i=1}^{I+1}n_i$ and $M=\sum_{i=1}^{I}m_{i}$. 
Let $D'$ be a diagram obtained from $D$ by changing the outer arc 
of size $m_i$ to two up arrows. 
We have $X_{D,D'}=[1+\sum_{j=1}^{i}n_j]$.
By a similar calculation to Case 1, the contribution to the left 
hand side of Eqn.(\ref{EPBII}) is 
\begin{eqnarray*}
\frac{[M][d+1]}{[1+M+N_{\uparrow}]}(Qq^d+Q^{-1}q^{-d})\Psi_{D} 
\end{eqnarray*}
where $d=G+M+M'+N$.

Let $D'$ be a diagram obtained from $D$ by changing the leftmost 
down arrow to an up arrow.
We have $X_{D,D'}=[N_{\uparrow}+1]$.
The contribution to the left hand side of Eqn.(\ref{EPBII}) is 
\begin{eqnarray}
\frac{[N_{\uparrow}+1][G+M']}{[1+M+N_{\uparrow}]}(Qq^{d}+Q^{-1}q^{-d})\Psi_{D}.
\end{eqnarray} 
The contribution of the diagonal term is $[Q;N]\Psi_{D}$.
Note that $N=N_{\uparrow}+2M+2M'+2G$. 
The sum of three contributions gives the right hand side of (\ref{EPBII}).

\paragraph{\bf Case 3}
Let $D$ be a diagram starting with $n_1$ up arrows, followed by 
an outer arc of size $m_1$, followed by $n_2$ up arrows, followed 
by an outer arc of size $m_2$, $\cdots$, followed by $n_{I+1}$ 
up arrows, followed by o-unpaired down arrow and ending with 
$G-1$ e-unpaired down arrows, $M'$ arcs and $G-1$ o-unpaired 
down arrows.
We set $N_{\uparrow}:=\sum_{i=1}^{I+1}n_{i}$ and 
$M:=\sum_{i=1}^{I}m_{i}$.
Let $D'$ be a diagram obtained from $D$ by changing the arc of 
size $m_i$ to two up arrows.
By a similar calculation to Case 2, the contribution to the left 
hand side of (\ref{EPBII}) is 
\begin{eqnarray*}
\frac{[M][d]}{[1+M+N_{\uparrow}]}(Qq^{d}+Q^{-1}q^{-d})\Psi_D.
\end{eqnarray*}
where $d=M+M'+N_{\uparrow}+G$.
Let $D'$ be a diagram obtained from $D$ by changing the leftmost 
down arrow to an up arrow. The contribution of this diagram is
\begin{eqnarray*}
\frac{[N_{\uparrow}+1][d]}{[1+M+N_{\uparrow}]}(Qq^{M'+G-1}+Q^{-1}q^{-(M'+G-1)})\Psi_D.
\end{eqnarray*}
The contribution of the diagonal term is $[Q;-N_{\uparrow}-1]\Psi_{D}$.
Note that $N=N_{\uparrow}+2M+2M'+2G-1$.
The sum of three contributions gives the right hand side of Eqn.(\ref{EPBII}).
This completes the proof.
\end{proof}

\subsection{Type BIII}
We consider the action of $X$ on Kazhdan--Lusztig bases of type BIII.

Let $D$ be a diagram of type BIII, $N_{\uparrow}$ be the number 
of up arrows and $N_{\downarrow}$ be the number of down arrows 
with a circled integer.
We define the weight of $D$ by $\mathrm{wt}(D):=N_{\uparrow}-N_{\downarrow}$. 
We enumerate the up arrows from left to right by 
$1,2,\ldots,N_{\uparrow}$.
For each $1\le i<N_{\uparrow}$, we denote by $X_{(i)}(D)$ a diagram
obtained from $D$ by connecting the $i$-th up arrow and the $(i+1)$-th 
up arrow via an arc.
We denote by $X_{(N_{\uparrow})}(D)$ a diagram obtained from $D$ by 
changing the $N_{\uparrow}$-th up arrow to the down arrow with a 
circled integer $N_{\downarrow}+1$.
We define the action of $X$ by 
\begin{eqnarray}
\label{ActionXBIII}
X(D):=\sum_{i=1}^{N_{\uparrow}}[i]X_{(i)}(D)
+[Q;\mathrm{wt}(D)]D.
\end{eqnarray} 

\begin{example}
Let $D$ be a diagram depicted as 
\begin{eqnarray*}
D=\tikzpic{-0.6}{
\upa{0};
\linkpattern{0.5/1.3,2.6/4.2,3/3.8,5.1/5.9}{}{}{2.2/2,4.7/1};
\upa{1.4};
}.
\end{eqnarray*}
We have 
\begin{eqnarray*}
X_{(1)}(D)=\tikzpic{-0.6}{
\linkpattern{0.1/1.7,0.5/1.3,2.6/4.2,3/3.8,5.1/5.9}{}{}{2.2/2,4.7/1};
},\qquad
X_{(2)}(D)=\tikzpic{-0.6}{
\upa{0};
\linkpattern{0.5/1.3,2.65/4.25,3.05/3.85,5.15/5.95}{}{}{1.7/3,2.25/2,4.75/1};
}.
\end{eqnarray*}
Then, the action of $X$ on $D$ is 
\begin{eqnarray*}
X(D)=X_{(1)}(D)+[2]X_{(2)}(D)+[Q;0]D.
\end{eqnarray*}

\end{example}

\begin{theorem}
The action of $X$ defined in Eqn.(\ref{ActionXBIII}) provides
the action of $X$ on the Kazhdan--Lusztig basis of Type BIII.
\end{theorem}
\begin{proof}
We prove Theorem by induction. 
When $N=1$, Theorem is true by a straightforward calculation.
We assume that Theorem holds true for diagrams of length up to $N-1$.
Let $D$ be a diagram of length $N$.
We have two cases for the leftmost arrow $a$ of $D$:
1) the arrow $a$ is an up arrow and 
2) the arrow $a$ is a down arrow.

\paragraph{\bf Case 1}
A diagram $D$ is written as $\uparrow D'$ where $D'$ is a diagram of 
length $N-1$.
The action of $X$ on $D$ is 
\begin{eqnarray}
\label{actionX-BIII-1}
X(\uparrow D')=q\uparrow X(D')+\downarrow D.
\end{eqnarray}
From the assumption, we have 
\begin{eqnarray*}
\uparrow X(D')&=&\sum_{i=2}^{N_{\uparrow}}[i-1]t_{(i)}(D)
+[Q;\mathrm{wt}(D)-1], \\
\downarrow D'&=&
\sum_{i=1}^{N_{\uparrow}}q^{-(i-1)}Q^{-1}X_{(i)}(D)
+q^{-N_{\uparrow}+N_{\downarrow}+1}.
\end{eqnarray*}
Note that $q[Q;d]+q^{-d}Q^{-1}=[Q;d+1]$. 
Inserting these two expressions into Eqn.(\ref{actionX-BIII-1}),
we obtain Eqn.(\ref{ActionXBIII}).

\paragraph{Case 2}
We have two cases for the leftmost arrow $a$:
i) the arrow $a$ is a down arrow with a circled integer $r$ and 
ii) the arrow $a$ forms an arc.
\paragraph{Case 2-i}
A diagram $D$ is written as 
$
\raisebox{-0.6\totalheight}{
\begin{tikzpicture}
\draw(0,0)--(0,-0.3)node[circle,inner sep=1pt,draw,anchor=north]{$r$};
\end{tikzpicture}}D'
$
where $D'$ is a diagram of length $N-1$. The weight of $D$ is $-r$.
We have 
\begin{eqnarray*}
X(
\raisebox{-0.5\totalheight}{
\begin{tikzpicture}
\draw(0,0)--(0,-0.3)node[circle,inner sep=1pt,draw,anchor=north]{$r$};
\end{tikzpicture}}D'
)
&=&q^{-1}\downarrow X(D')+\uparrow D'-q^{r}Q^{-1}\uparrow X(D')
-q^{r-1}Q^{-1}\downarrow D', \\
&=&[Q;-r]D
\end{eqnarray*}
where we have used $X(D')=[Q;-(r-1)]D'$, $q^{-1}[Q;-(r-1)]-q^{r-1}Q^{-1}=[Q;-r]$
and $1-q^rQ^{-1}[Q;-(r-1)]=-q^{r-1}Q^{-1}[Q;-r]$.

\paragraph{Case 2-ii}
By a similar argument to Case 2-b in the proof of Theorem~\ref{thm-ActionXBI},
we can assume that a diagram $D$ is written as 
$
\raisebox{-0.5\totalheight}{
\begin{tikzpicture}
\draw(0,0)..controls(0,-0.6)and(0.6,-0.6)..(0.6,0);
\end{tikzpicture}
}D'
$ where $D'$ is a diagram of length $N-2$.
We have 
\begin{eqnarray*}
X(
\raisebox{-0.5\totalheight}{
\begin{tikzpicture}
\draw(0,0)..controls(0,-0.6)and(0.6,-0.6)..(0.6,0);
\end{tikzpicture}
}D')
&=&X(\downarrow\uparrow D'-q^{-1}\uparrow\downarrow D') \\
&=&
\raisebox{-0.5\totalheight}{
\begin{tikzpicture}
\draw(0,0)..controls(0,-0.6)and(0.6,-0.6)..(0.6,0);
\end{tikzpicture}}
X(D').
\end{eqnarray*}
In both Case 1 and 2, we have Eqn.(\ref{ActionXBIII}) for 
a diagram $D$.
This completes the proof.
\end{proof}

\begin{theorem}
$X$ has the eigenvalue $[Q;N-2i]$,$1\le i\le N$, of multiplicity
$\displaystyle\genfrac{(}{)}{0pt}{}{N}{i}$. 
\end{theorem}
\begin{proof}
Let $X_{D,D'}$ be a matrix representation of $X$ with respect to 
diagrams. 
Recall that there is a one-to-one correspondence between 
a binary string and a diagram.
Thus we have a natural order for diagrams induced from 
the lexicographic order for binary strings.
From Eqn.(\ref{ActionXBIII}), the matrix $X_{D,D'}$ is 
a lower triangular matrix.
Since the diagonal terms are $[Q;\mathrm{wt}(D)]$, the 
eigenvalues are of the form $[Q;N-2i]$.
The multiplicities are equal to the number of $D$ such 
that $\mathrm{wt}(D)=N-2j$.
By a similar argument to the second paragraph in the 
proof of Theorem~\ref{theorem-XBII}, the cardinality 
of $D$ satisfying $\mathrm{wt}(D)=N-2j$ is 
$\genfrac{(}{)}{0pt}{}{N}{j}$.

\end{proof}

Let $D$ be a diagram of type BIII, $S$ be the set of arcs, 
$S_{\uparrow}$ be the set of (unpaired) up arrows and 
$S_{\downarrow}$ be the set of down arrows with a circled integer.
We enumerate up arrows, arcs and down arrows with a circled integer
from left to right by $1,2,\ldots$. 
Let $N_A$ be the integer assigned to 
$A\in S\cup S_{\uparrow}\cup S_{\downarrow}$.
We define 
\begin{eqnarray*}
N_1:=\prod_{A\in S\cup S_{\downarrow}}[N_A].
\end{eqnarray*}
Let $B$ be an arc of size $m_B$. 
We define 
\begin{eqnarray*}
N_2:=\prod_{B\in S}[m_B]^{-1}.
\end{eqnarray*}
Similarly, we enumerate arcs and down arrows with a circled integer from right 
to left by $1,2,\ldots$.
Let $N_C$ be the integer assigned to $C\in S\cup S_{\downarrow}$.
We define 
\begin{eqnarray*}
N_3&:=&\prod_{S\in S_{\downarrow}}[N_C]^{-1}, \\
N_4&:=&\prod_{i=1}^{d}(Qq^{i-1}+Q^{-1}q^{-(i-1)}),
\end{eqnarray*}
where $d=|S_{\uparrow}|+|S|$.
\begin{defn}
\label{PsiBIII}
$\Psi_D:=N_1\cdot N_2\cdot N_3\cdot N_4$.
\end{defn}

\begin{example}
\label{example-BIII}
Let $D$ be a diagram depicted as 
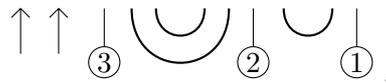
\begin{eqnarray*}
\tikzpic{0}{
\draw(0,0)--(0,-0.6)(-0.12,-0.12)--(0,0)--(0.12,-0.12);
\draw(-0.5,0)--(-0.5,-0.6)(-0.12-0.5,-0.12)--(-0.5,0)--(0.12-0.5,-0.12);
\linkpattern{1.2/2.8,1.6/2.4,3.7/4.5}{}{}{0.75/3,3.2/2,4.9/1};
}.
\end{eqnarray*}
We have 
\begin{eqnarray*}
N_1=\frac{[8]!}{[2]}, \qquad N_2=[2]^{-1}, \qquad 
N_3=[3]^{-1}[6]^{-1},\qquad N_4=\prod_{i=1}^{5}(Qq^{i-1}+Q^{-1}q^{-(i-1)}).
\end{eqnarray*}

\end{example}

\begin{theorem}
The vector $\Psi$ is the eigenvector of $X$ with 
the eigenvalue $[Q;N]$.
\end{theorem}
\begin{proof}
Let $X=(X_{D,D'})$ be a matrix representation of $X$.
We will show $\sum_{D'}X_{D,D'}\Psi_{D'}=[Q;N]\Psi_{D}$ 
by computing the left hand side. 
We have two cases for a diagram $D$:
1) $D$ has no down arrows and 
2) $D$ has down arrows with a circled integer.

\paragraph{\bf Case 1}
A diagram $D$ is depicted as in Eqn.(\ref{diagram1up}) 
where $m_i$, $1\le i\le I$, is the size of an outer arc.
Set $M'=\sum_{i=1}^{I}m_i$ and $N'=\sum_{i=1}^{I}n_i$.
Since $X_{D,D}=[Q;N']$, we have a contribution from $D$ itself.
That is $[Q;N']\Psi_{D}$. 
Let $D'$ be a diagram obtained from $D$ by changing an outer 
arc of size $m_i$ to two up arrows.
Since $X_{D,D'}=[1+\sum_{j=1}^{i}n_j]$, the contribution of such 
$D'$'s is 
\begin{eqnarray*}
\Psi_{D}\sum_{i=1}^{I}(Qq^{d}+Q^{-1}q^{-d})[m_i]
\left[1+\sum_{j=1}^{i}n_j\right]
\frac{\prod_{j=i+1}^{I}[1+\sum_{k=1}^{j}(n_k+m_k)]}
{\prod_{j=i}^{I}[1+n_j+\sum_{k=1}^{j-1}(n_k+m_k)]}.
\end{eqnarray*} 
where $d=N'+M'$. 
From Lemma~\ref{lemma-app0}, the above expression is reduced to 
$[M'](Qq^d+Q^{-1}q^{-d})\Psi_{D}$.
Therefore, the sum of contributions is $[Q;N'+2M']$.

\paragraph{\bf Case 2}
A diagram $D$ is locally depicted as
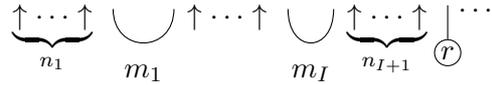
\begin{eqnarray*}
\tikzpic{-0.5}{
\draw(0,0)node[anchor=north]
{$\underbrace{\uparrow\cdots\uparrow}_{n_1}$};
\draw(0.8,-0.2)..controls(0.8,-0.8)and(1.6,-0.8)..(1.6,-0.2);
\draw(1.2,-0.8)node[anchor=north]{$m_1$};
\draw(2.3,0)node[anchor=north]{$\uparrow\cdots\uparrow$};
\draw(3.1,-0.2)..controls(3.1,-0.8)and(3.7,-0.8)..(3.7,-0.2);
\draw(3.4,-0.8)node[anchor=north]{$m_I$};
\draw(4.4,0)node[anchor=north]
{$\underbrace{\uparrow\cdots\uparrow}_{n_{I+1}}$};
\draw(5.2,-0.15)--(5.2,-0.6)node[circle,inner sep=1pt,draw,anchor=north]{$r$};
\draw(5.6,0)node[anchor=north]{$\cdots$};
}
\end{eqnarray*}
where the region inside of an outer arc of size $m_i$ is filled with smaller arcs.
Set $M'=\sum_{i=1}^{I}m_i$ and $N'=\sum_{i=1}^{I+1}n_i$.
Let $S_R$ be the set of arcs which are right to the down arrow 
with a circled integer $r$.
We have three types of contributions:
a) the diagonal term, {\it i.e.,} $D'=D$,
b) $D'$ is obtained from $D$ by changing an outer arc of size 
$m_i$ to two up arrows and 
c) $D'$ is obtained from $D$ by changing the down arrow with 
a circled integer $r$ to an up arrow.

The contribution of case a is $[Q;N'-r]\Psi_{D}$.
By a similar argument to Case 1, the contribution of 
case b is 
\begin{eqnarray*}
[M'](Qq^{d}+Q^{-1}q^{-d})
\frac{[M'+N'+|S_R|+r+1]}{[M'+N'+1]}\Psi_{D}.
\end{eqnarray*} 
where $d=N'+M'+|S_R|$.
For case c, the contribution is 
\begin{eqnarray}
\frac{[N'+1][|S_R|+r]}{[N'+M'+1]}(Qq^{d}+Q^{-1}q^{-d})\Psi_D.
\end{eqnarray}
The sum of three contributions becomes 
$[Q;N'+2M'+2|S_R|+r]\Psi_{D}=[Q;N]\Psi_{D}$.
This completes the proof.
\end{proof}

\subsection{Standard bases}
We consider the action of $X$ on the standard basis 
$v:=v_{\epsilon_1}\otimes\ldots\otimes v_{\epsilon_{N}}$
where $\epsilon_{i}=\pm1$. 
Let $d_{i}:=\sum_{j=1}^{i}\epsilon_{i}$. 
For each $i$, $1\le i\le N$, we define 
\begin{eqnarray*}
X_{(i)}(v):=v_{\epsilon_1}\otimes\ldots\otimes v_{\epsilon_{i-1}}
\otimes v_{-\epsilon_{i}}\otimes
v_{\epsilon_{i+1}}\otimes\ldots\otimes v_{\epsilon_{N}}.
\end{eqnarray*}
The action of $X$ is defined by 
\begin{eqnarray}
\label{XonSB}
X(v):=\sum_{i=1}^{N}q^{d_{i-1}}X_{(i)}(v)
+q^{d_{N}}[Q;0]v.
\end{eqnarray}
\begin{prop}
The definition (\ref{XonSB}) provides the action of $X$ on 
the standard basis.
\end{prop}
\begin{proof}
We prove Proposition by induction on $N$. 
When $N=1$, Proposition holds true by a straightforward 
calculation. 
We assume that Proposition is true up to some $N\ge2$.
A standard basis $v:=v_{\epsilon_1}\otimes\ldots\otimes v_{\epsilon_N}$
is written as $v=v_{\epsilon_1}\otimes v'$ where $v'$ is a standard 
basis of length $N-1$.
Let $d'_{i}:=\sum_{j=2}^{i}\epsilon_{j}$. 
From the induction assumption, we have 
\begin{eqnarray}
\label{XonSBa}
X(v')=\sum_{i=2}^{N}q^{d'_{i-1}}X_{(i-1)}(v')+q^{d'_{N}}[Q;0]v'
\end{eqnarray}
From Eqns.(\ref{coproduct}) and (\ref{XonSBa}), we have 
\begin{eqnarray*}
X(v)&=&q^{\epsilon_1}v_{\epsilon_1}\otimes X(v')
+v_{-\epsilon_1}\otimes v' \\ 
&=&\sum_{i=2}^{N}q^{\epsilon_1+d'_{i-1}}v_{\epsilon_1}\otimes X_{(i-1)}(v')
+v_{-\epsilon_1}\otimes v'+q^{\epsilon_1+d'_{N}}[Q;0]v \\
&=&\sum_{i=1}^{N}q^{d_{i-1}}X_{(i)}(v)
+q^{d_{N}}[Q;0]v,
\end{eqnarray*}
where we have used $v_{\epsilon_1}\otimes X_{(i-1)}(v')=X_{(i)}(v)$.
\end{proof}

For a binary string $\epsilon\in\{\pm\}^{N}$, let $I_{\epsilon}$ 
be the set of positions of pluses from right. 
We define 
\begin{eqnarray*}
N_{1}&:=&q^{d_{\epsilon}}, \\
N_{2}&:=&Q^{d'}
\end{eqnarray*}
where $d_{\epsilon}:=\sum_{i\in I_{\epsilon}}(i-1)$ and $d'$ is the number 
of pluses in the binary string $\epsilon$.
We define a vector $\Psi^{0}:=\sum_{\epsilon}\Psi^{0}_{\epsilon}|\epsilon\rangle$ 
as follows.
\begin{defn}
\label{defPsi-SB}
$\Psi^{0}_{\epsilon}:=N_{1}N_{2}$.
\end{defn}

\begin{example}
Let $\epsilon_{1}:=++-+$, $\epsilon_2:=--++$ and $\epsilon_3:=--+-$.
We have 
\begin{eqnarray*}
\Psi_{\epsilon_1}=q^{5}Q^{3}, \quad \Psi_{\epsilon_2}=qQ^{2}, \quad
\Psi_{\epsilon_3}=qQ.
\end{eqnarray*}
\end{example}
\begin{prop}
The vector $\Psi$ is the eigenvector of $X$ with the eigenvalue $[Q;N]$. 
\end{prop}
\begin{proof}
Let $X=(X_{\epsilon,\epsilon'})$ be a matrix representation of $X$. 
We will show 
$\sum_{\epsilon'}X_{\epsilon,\epsilon'}\Psi_{\epsilon'}=[Q;N]\Psi_{\epsilon}$ 
by computing the left hand side.
We make use of a diagram $D$ of type A associated with $\epsilon$.
The diagram $D$ is depicted as in Eqn.(\ref{Diagram0}).
Set $N_{\uparrow}=\sum_{i=1}^{I+1}n_{i}$, $N_{\downarrow}=\sum_{i=1}^{J+1}n'_{i}$, 
$M=\sum_{i=1}^{I}m_{i}$ and $M'=\sum_{i=1}^{I}m'_{i}$.

We enumerate all arrows from left to right by $1, 2,\ldots, N$.
From Eqn.(\ref{XonSB}), a diagram $D'$ satisfying $X_{D,D'}\neq0$ can 
be obtained from $D$ by reversing an up (resp. down) arrow to a down 
(resp. up) arrow.
The number of reversed arrows in $D'$ is at most one.
Let $N_{\uparrow}$ (resp. $N_{\downarrow}$) be the number of unpaired
up (down) arrows in $D$. 
Since we have $X_{D,D}=q^{N_{\uparrow}-N_{\downarrow}}[Q;0]$, 
the contribution from the diagonal term is 
$q^{N_{\uparrow}-N_{\downarrow}}[Q;0]\Psi_{D}$. 
 
Firstly, we  reverse an up arrow $a$ of $D$ to obtain $D'$. 
We have two cases: 1) $a$ is an unpaired up arrow and 
2) $a$ is an up arrow forming an arc.

\paragraph{Case 1}
Let $j$ be the position of the arrow $a$. 
Then, $j$ satisfies 
$\sum_{k=1}^{i-1}n_k+2\sum_{k=1}^{i-1}m_{k}+1\le j\le 
\sum_{k=1}^{i}n_k+2\sum_{k=1}^{i-1}m_{k}$ for some 
$1\le i\le I+1$. 
We have $X_{D,D'}=q^{j-2\sum_{k=1}^{i-1}m_{k}-1}$ and 
$\Psi_{D'}=q^{-(N-j)}Q^{-1}\Psi_{D}$. 
Thus the contribution is given by 
\begin{eqnarray}
\label{XonSB-1}
q^{-N+2j-2\sum_{k=1}^{i-1}m_{k}-1}Q^{-1}\Psi_{D}.
\end{eqnarray}

\paragraph{Case 2}
We have two cases for the arc $b$ containing the arrow $a$: 
a) $b$ is in the region inside of an outer arc of size $m_i$ and 
b) $b$ is in the region inside of an outer arc of size $m'_i$.

\paragraph{Case 2-a}
Suppose that there are $l_1$ arcs outside of $a$ (including $a$ 
itself), $n_{\uparrow}$ up arrows and $l_{2}$ arcs left to 
the arrow $a$.
Then, $a$ is the $(n_{\uparrow}+l_1+2l_{2}+1)$-th arrow from 
left.
We have $X_{D,D'}=q^{n_{\uparrow}-l_1}$ and 
$\Psi_{D'}=q^{-(N-n_{\uparrow}-l_1-2l_2-1)}Q^{-1}\Psi_{D}$. 
Thus the contribution is given by 
\begin{eqnarray}
\label{XonSB-2}
q^{-N+2n_{\uparrow}+2l_2+1}Q^{-1}\Psi_{D}.
\end{eqnarray}
Since $b$ is in the region inside of an outer arc of size $m_i$, 
we have $n_{\uparrow}=\sum_{k=1}^{i}n_{k}$ and $l_2$ takes 
the values $0, 1,\ldots, m_{i}-1$ once.

\paragraph{Case 2-b}
By a similar argument to Case 2-a, the contribution of $D'$'s is 
given by 
\begin{eqnarray}
\label{XonSB-3}
q^{-N+2N_{\uparrow}+2l_2+1}Q^{-1}\Psi_{D}
\end{eqnarray}
where $l_2$ takes $0, 1,\ldots, M'-1$ once.

The sum of contributions from Eqns.(\ref{XonSB-1}) to (\ref{XonSB-3})
is given by 
\begin{eqnarray}
\label{XonSB-4}
q^{-N+N_{\uparrow}+M+M'}[N_{\uparrow}+M+M']Q^{-1}\Psi_{D}.
\end{eqnarray}

Secondly, we reverse an down arrow in $D$ to obtain $D'$.
By a similar argument to Case 1 and 2, the sum of contributions 
is given by 
\begin{eqnarray}
\label{XonSB-5}
q^{N_{\uparrow}+M+M'}[N_{\downarrow}+M+M']Q\Psi_{D}.
\end{eqnarray}
The sum of Eqns.(\ref{XonSB-4}), (\ref{XonSB-5}) and the diagonal 
contribution is $[Q;N]\Psi_{D}$. 
This completes the proof.
\end{proof}

\section{\texorpdfstring{Action of Hamiltonian on $\Psi$}
{Action of Hamiltonian on Psi}}
\label{sec-eigenH}

Let $a$ be an arc and $d$ be the number of arcs and dashed arc 
outside of $a$ (including $a$ itself).  
We call the number $d$ the {\it depth} of the arc $a$.

In this section, we will show $e_i\Psi=0$ for $1\le i\le N$ for 
arbitrary $q, Q, Q_{0}$  and $e_0\Psi=0$ under the integrable 
condition: 
\begin{eqnarray}
\label{def-icond}
q^{N-1}QQ_{0}-1=0.
\end{eqnarray}
Since we have an explicit action of $e_i$, $0\le i\le N$,  
on a diagram $D$, we compute explicitly the $D$-component 
of $e_{i}\Psi$. 

\subsection{Type A}

\begin{prop}
\label{prop-HA}
We have
\begin{eqnarray}
\label{prop-epsiA}
e_{i}\Psi=0, \qquad 1\le i\le N-1.
\end{eqnarray}
\end{prop}
\begin{proof}
Suppose that $D$ does not have an arc connecting 
$i$-th and $(i+1)$-th arrows.
There is no diagram $D'$ such that $D$ appears in the 
expansion of $e_i(D)$.
Thus the $D$-component of Eqn.(\ref{prop-epsiA}) is 
obviously zero.

Below, we consider the case where $D$ has an arc connecting 
the $i$-th and the $(i+1)$-th arrows. 
We denote by $e$ this small arc. 
There are two cases: 1) the depth of $e$ is greater than one  
and 2) the depth of $e$ is one.
We will show that the left hand side of (\ref{prop-epsiA}) 
is actually zero.

\paragraph{\bf Case 1}
Let $d$ be the depth of $e$. There exists a unique arc $e'$ 
of depth $d-1$ such that $e$ is inside of $e'$.
There may be several arcs of depth $d$ inside of $e'$.
The diagram $D$ locally looks like 
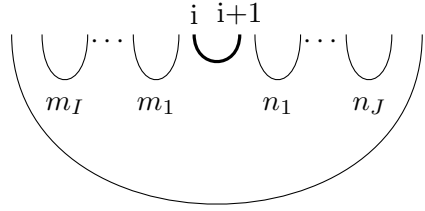
\begin{eqnarray*}
\begin{tikzpicture}
\draw(-0.2,0)..controls(-0.2,-3)and(5.2,-3)..(5.2,0);
\draw[very thick](2.2,0)..controls(2.2,-0.48)and(2.8,-0.48)..(2.8,0);
\draw(2.2,0)node[anchor=south]{i}(2.8,0)node[anchor=south]{i+1};
\draw(0.2,0)..controls(0.2,-0.8)and(0.8,-0.8)..(0.8,0)
(0.5,-0.7)node[anchor=north]{$m_I$};
\draw(0.7,-0.1)node[anchor=west]{$\cdots$};
\draw(1.4,0)..controls(1.4,-0.8)and(2.0,-0.8)..(2.0,0)
(1.7,-0.7)node[anchor=north]{$m_1$};
\draw(3.0,0)..controls(3,-0.8)and(3.6,-0.8)..(3.6,0)
(3.3,-0.7)node[anchor=north]{$n_1$};
\draw(3.5,-0.1)node[anchor=west]{$\cdots$};
\draw(4.2,0)..controls(4.2,-0.8)and(4.8,-0.8)..(4.8,0)
(4.5,-0.7)node[anchor=north]{$n_J$};
\end{tikzpicture}
\end{eqnarray*}
The arcs which are left to $e$ are of size $m_i$, $1\le i\le I$ and 
the arcs which are right to $e$ are of size $n_j$, $1\le j\le J$. 
Set $M'=\sum_{i}^{I}m_i$ and $N'=\sum_{j=1}^{J}n_j$. 
The arc $e'$ is the one of size $M'+N'+2$ and of depth $d-1$. 
Let $D'$ be a diagram such that $e_i(D')$ contains the term 
$D$.
Suppose that the arc of depth $d$ and of length $m_i$ 
connects the $k$-th and the $l$-th ($k<l$) arrows.
We denote by $D'$ a diagram obtained from $D$ by connecting 
the $l$-th and the $i$-th arrows via an arc and also 
the $k$-th and $(i+1)$-th arrows via an arc.  
Since $e_{i}(D')=D$, the contribution of such $D'$'s
to the left hand side of Eqn.(\ref{prop-epsiA}) is 
\begin{eqnarray}
\label{epA1}
\sum_{i=1}^{I}\frac{[m_i]}
{[1+\sum_{j=1}^{i-1}m_i][1+\sum_{j=1}^{i}m_i]}\Psi_{D}.
\end{eqnarray}
Similarly, suppose that the arc of depth $d$ and of 
length $n_j$ connects the $k$-th and the $l$-th 
arrows. 
Let $D'$ be a diagram obtained from $D$ by connecting
the $i$-th and the $l$-th arrows via an arc and 
the $k$-th and the $(i+1)$-th arrows via an arc.
The contribution of such $D'$'s is 
\begin{eqnarray}
\label{epA2}
\sum_{i=1}^{J}\frac{[n_i]}
{[1+\sum_{j=1}^{i-1}n_i][1+\sum_{j=1}^{i}n_i]}\Psi_{D}.
\end{eqnarray}
Suppose that the arc of depth $d-1$ and of size $M+N+2$ 
connects the $k$-th and $l$-th arrows. 
Let $D'$ be a diagram obtained from $D$ by connecting 
the $k$-th and $i$-th arrows via an arc and $(i+1)$-th 
and $l$-th arrows via an arc.
The contribution of this $D'$ is 
\begin{eqnarray}
\label{epA3}
\frac{[M+N+2]}{[M+1][N+1]}\Psi_{D}.
\end{eqnarray}
From Lemma~\ref{lemma-app2}, the sum of Eqn.(\ref{epA1}) to Eqn.(\ref{epA3})
is $[2]\Psi_{D}$.

Since $e_i(D)=-[2]D$, the contribution of $D$ 
to the left hand side of Eqn.(\ref{prop-epsiA}) is 
$-[2]\Psi_{D}$.
Thus the left hand side of Eqn.(\ref{prop-epsiA}) is zero.

\paragraph{\bf Case 2}
The diagram $D$ is locally depicted as 
\begin{eqnarray}
\label{diagramA}
\cdots\alpha  
\raisebox{-0.8\totalheight}{
\begin{tikzpicture}
\draw(0,0)..controls(0,-0.8)and(0.8,-0.8)..(0.8,0)
(0.4,-0.8)node{$m_{I+1}$};
\end{tikzpicture}}
\raisebox{-0.8\totalheight}{
\begin{tikzpicture}
\draw(0,0)..controls(0,-0.8)and(0.8,-0.8)..(0.8,0)
(0.4,-0.8)node{$m_{I}$};
\end{tikzpicture}}\  
\ldots
\raisebox{-0.8\totalheight}{
\begin{tikzpicture}
\draw(0,0)..controls(0,-0.8)and(0.8,-0.8)..(0.8,0)
(0.4,-0.8)node{$m_{1}$};
\end{tikzpicture}}
\raisebox{-0.4\totalheight}{
\begin{tikzpicture}
\draw(0,0)..controls(0,-0.6)and(0.6,-0.6)..(0.6,0)
(0,0.2)node{$i$}(0.6,0.2)node{$i+1$};
\end{tikzpicture}}
\raisebox{-0.8\totalheight}{
\begin{tikzpicture}
\draw(0,0)..controls(0,-0.8)and(0.8,-0.8)..(0.8,0)
(0.4,-0.8)node{$n_{1}$};
\end{tikzpicture}}
\ldots
\raisebox{-0.8\totalheight}{
\begin{tikzpicture}
\draw(0,0)..controls(0,-0.8)and(0.8,-0.8)..(0.8,0)
(0.4,-0.8)node{$n_{J}$};
\end{tikzpicture}}
\beta\cdots
\end{eqnarray}
where $\alpha$ are $\beta$ either $\uparrow$, $\downarrow$ or empty. 
By empty we mean that there are no arrows.
The inside of the arc of size $m_i$ or $n_i$ is filled with arcs.
Let $M'=\sum_{i=1}^{I}m_i$ and $N'=\sum_{i=1}^{J}n_i$.
Since $e_i(D)=-[2]D$, the contribution of $D$ to Eqn.(\ref{prop-epsiA})
is $-[2]\Psi_{D}$.
By a similar argument to Case 1, we have $D'$'s which changes
the arc $e$ and the arc of size $m_i$ and of depth one 
to two arcs of size $1+\sum_{j=1}^{i-1}m_i$ and $1+\sum_{j=1}^{i}m_i$.
We also have similar $D'$'s regarding the arc of size $n_i$.
These $D'$'s contribution to the left hand side of
Eqn.(\ref{prop-epsiA}) is 
\begin{eqnarray}
\label{epA4}
\sum_{i=1}^{I}\frac{[m_i]\Psi_{D}}
{[1+\sum_{j=1}^{i-1}m_i][1+\sum_{j=1}^{i}m_i]}
+
\sum_{i=1}^{J}
\frac{[n_i]\Psi_{D}}{[1+\sum_{j=1}^{i-1}n_i][1+\sum_{j=1}^{i}n_i]}
\end{eqnarray}
We have eight cases for the diagram $D$: 
a) $(\alpha,\beta)=(\uparrow,\uparrow)$, 
b) $(\alpha,\beta)=(\uparrow,\downarrow)$,  
c) $(\alpha,\beta)=(\downarrow,\downarrow)$, 
d) $(\alpha,\beta)=(\uparrow,\emptyset)$, 
e) $(\alpha,\beta)=(\downarrow,\emptyset)$
f) $(\alpha,\beta)=(\emptyset,\uparrow)$,
g) $(\alpha,\beta)=(\emptyset,\downarrow)$, and 
h) $(\alpha,\beta)=(\emptyset,\emptyset)$
where $\alpha=\emptyset$ (resp. $\beta=\emptyset$) 
means that there are no arrows left to (resp. right to)
$\alpha$ (resp. $\beta$). 

\paragraph{Case 2-a}
Let $d$ be the sum of the numbers of arcs and up arrows left to 
the arrow $\alpha$.
Let $D'$ be a diagram obtained from $D$ by connecting the arrow $\alpha$
and the $i$-th arrow via an arc and  putting an up arrow at the $(i+1)$-th 
site.
Thus the contribution to Eqn.(\ref{prop-epsiA}) is 
\begin{eqnarray}
\label{epA5}
\frac{[d]}{[M'+1][d+M'+1]}\Psi_{D}.
\end{eqnarray}
Similarly, let $D'$ be a diagram obtained from $D$ by connecting $(i+1)$-th site 
and the arrow $\beta$ via an arc and putting an up arrow at 
$i$-th site. 
The contribution of this $D'$ is 
\begin{eqnarray}
\label{epA6}
\frac{[d+M'+N'+2]}{[N'+1][d+M'+1]}\Psi_{D}.
\end{eqnarray}
The sum of contributions from Eqns.(\ref{epA4}) to (\ref{epA6}) is 
$[2]\Psi_{D}$ by applying Lemma~\ref{lemma-app2} to Eqn.(\ref{epA4}).
This implies that the contributions of Eqn.(\ref{prop-epsiA}) is zero.

\paragraph{Case 2-b to 2-h}
By a similar argument to Case 2-a, one can show that the sum of 
contributions is zero.
In Case 2-b, 2-d and 2-h, we have a contribution from $D'$ which 
is obtained from $D$ by putting an up arrow at $i$-th site and a 
down arrow at $(i+1)$-th site.
The contribution of this $D'$ is 
\begin{eqnarray*}
\frac{[d+M'+N'+d'+2]}{[d+M'+1][d'+N'+1]}\Psi_{D},
\end{eqnarray*}
where $d$ is the number of up arrows and arcs left to the arrow $\alpha$ and 
$d'$ is the number of down arrows and arcs right to the arrow $\beta$.
This completes the proof.
\end{proof}

\begin{prop}
\label{prop-eNAPsi}
We have $e_N\Psi=0$.
\end{prop}
\begin{proof}
We have three cases for $D$: 1) the rightmost arrow is an
up arrow, 2) the rightmost arrow forms an arc and 
3) the rightmost arrow is a down arrow.

\paragraph{\bf Case 1}
The diagram $D$ is depicted as 
\begin{eqnarray}
\label{pic-A-1}
\tikzpic{-0.5}{
\draw(0,0)node[anchor=north]
{$\underbrace{\uparrow\cdots\uparrow}_{n_1}$};
\draw(0.8,-0.2)..controls(0.8,-0.8)and(1.6,-0.8)..(1.6,-0.2);
\draw(1.2,-0.8)node[anchor=north]{$m_1$};
\draw(2.3,0)node[anchor=north]{$\uparrow\cdots\uparrow$};
\draw(3.1,-0.2)..controls(3.1,-0.8)and(3.7,-0.8)..(3.7,-0.2);
\draw(3.4,-0.8)node[anchor=north]{$m_I$};
\draw(4.4,0)node[anchor=north]
{$\underbrace{\uparrow\cdots\uparrow}_{n_{I+1}}$};
},
\end{eqnarray}
where the region inside of the arc of size $m_i$, $1\le i\le I$, 
is filled with arcs.
Set $d=\sum_{i=1}^{I+1}n_i+\sum_{i=1}^{I}m_i$.
We have two cases for $n_{I+1}$: 
a) $n_{I+1}\ge2$, and 
b) $n_{I+1}=1$.

We have two common contributions to the $D$-component of
$e_N\Psi$ for both case a) and b).
Since $e_N(D)=-Q^{-1}D+\cdots$, the contribution from $D$ 
itself is $-Q^{-1}\Psi_D$. 
Let $D'$ be a diagram obtained from $D$ by changing 
the rightmost up arrow to a down arrow.
We have $e_N(D')=D+\cdots$. Thus the contribution is 
given by $q^{-(d-1)}Q^{-1}[d]\Psi_{D}$.

\paragraph{Case 1-a}
Let $D'$ be a diagram obtained from $D$ by connecting 
the $(N-1)$-th site and the $N$-th site via an arc.
We have $e_N(D')=-q^{-1}D+\cdots$.
The contribution is given by $-q^{-d}Q^{-1}[d-1]\Psi_D$.
The sum of three contributions is zero, which implies 
$e_N\Psi=0$.

\paragraph{Case 1-b}
Let $D'$ be a diagram obtained from $D$ by connecting 
the $(N-2m_I-1)$-th site and the $N$-th site via 
an arc. 
We have $e_N(D')=-q^{-1}D+\cdots$. 
The contribution of $D'$ is 
\begin{eqnarray}
-q^{-d}Q^{-1}\frac{[d-m_I-1]}{[m_I+1]}\Psi_D.
\end{eqnarray}
Let $D'$ be a diagram obtained from $D$ by connecting 
the $(N-1)$-th site and the $N$-th site via an arc 
and changing the $(N-2m_I)$-th up arrow to an down 
arrow. We have $e_N(D')=-q^{-1}D+\cdots$.
The contribution is given by 
\begin{eqnarray}
-q^{-d}Q^{-1}\frac{[d][m_I]}{[m_I+1]}.
\end{eqnarray}
Thus the sum of four contributions is zero, which implies 
$e_N\Psi=0$.

\paragraph{\bf Case 2}
We have two cases for $D$: a) $D$ has no down arrows,
b) $D$ has down arrows.

\paragraph{Case 2-a}
The diagram $D$ is depicted as Eqn.(\ref{pic-A-1}) with 
$n_{I+1}=0$.
Set $N'=\sum_{i=1}^{I}n_i$, $M=\sum_{i=1}^{I}m_i$ and 
$d=N'+M$.
We have four types of contributions to the $D$-component of
$e_N\Psi$.
Since $e_N(D)=-Q^{-1}D+\cdots$, we have the contribution
$-Q^{-1}\Psi_D$. 

Let $D'$ be a diagram obtained from $D$ by changing the 
outer arc of size $m_I$ to two down arrows. 
We have $e_N(D')=D+\cdots$. 
Thus the contribution is 
\begin{eqnarray}
q^{-(d-1)}Q^{-1}\frac{[N'+M+1][m_I]}{[m_I+1]}\Psi_D.
\end{eqnarray}

Suppose that $D$ has arcs of depth two inside of the arc 
of size $m_I$. We enumerate these arcs of depth two from 
right to left by $1,2,\ldots,J$ where $J$ is the number 
of arcs of depth two.
We denote by $\widetilde{m_{j}}$, $1\le j\le J$, the size 
of the $j$-th arc of depth two.
Suppose that the $i_1$-th arrow and the $i_2$-th arrow 
form the arc of size $\widetilde{m_j}$. 
Let $D'$ be a diagram obtained from $D$ by connecting 
the $i_2$-th site and the $N$-th site via an arc and 
putting two down arrows at the $(N-2m_I+1)$-th site 
and the $i_1$-th site.
We have $e_N(D')=-q^{-1}D+\cdots$. 
The contribution is given by 
\begin{eqnarray}
-q^{-1}\sum_{j=1}^{J}
\frac{q^{-(d-1)}Q^{-1}[d+1][m_I][\widetilde{m_j}]}
{[m_I+1][\sum_{k=1}^{j}\widetilde{m_k}+1][\sum_{k=1}^{j-1}\widetilde{m_k}+1]}\Psi_D
=-q^{-d}Q^{-1}\frac{[d+1][m_I-1]}{[m_I+1]}\Psi_D.
\end{eqnarray}
where we have used Lemma~\ref{lemma-app2} and 
$\sum_{j=1}^{J}\widetilde{m_j}=m_{I}-1$.

Let $D'$ be a diagram obtained from $D$ by changing the rightmost up 
arrow to a down arrow.
$e_N(D')=-q^{-2}D+\cdots$.
The contribution is given by $-q^{-(d+1)}Q^{-1}[d-m_I][m_I+1]^{-1}$.

By a straightforward calculation, the sum of four contributions 
is zero, which implies $e_N\Psi=0$.

\paragraph{Case 2-b}
Let $D$ be a diagram which starts with $n'_1$ up arrow, followed by an 
outer arc of size $m'_1$, followed by $n'_2$ up arrows, followed 
by an outer arc of size $m'_2$, $\cdots$, 
followed by $n'_{J+1}$ up arrows, followed by $n_{I}$ down 
arrows, followed by an outer arc of size $m_I$, followed by 
$n_{I-1}$ down arrows, $\cdots$, and ends with an outer arc 
of size $m_1$. 
We set $N'=\sum_{i=1}^{J+1}n'_i$, $M'=\sum_{i=1}^{J}m'_i$, 
$N_{\downarrow}=\sum_{i=1}^{I}n_i$, $M=\sum_{i=1}^{I}m_i$, 
$L=N'+M'+N_{\downarrow}+M$, $d=N'+M'+M$ 
and $v_i=\sum_{j=1}^{i-1}n_j+m_j$.

We have five types of contributions to the $D$-component of 
$e_N\Psi$.
Since $e_N(D)=-Q^{-1}D+\cdots$, we have a contribution $-Q^{-1}\Psi_D$.

Let $D'$ be a diagram obtained from $D$ by changing the outer arc of 
size $m_1$ to two down arrows.
We have $e_N(D')=D+\cdots$.
The contribution is 
\begin{eqnarray}
\label{eNA-2-1}
q^{-(d-1)}Q^{-1}\Psi_D[L+1]\frac{[m_1]}{[1+m_1]}
\prod_{i=1}^{I}\frac{[1+m_i+v_i]}{[1+v_{i+1}]}.
\end{eqnarray}
Suppose that $D$ has arcs of depth two inside of the arc of size 
$m_1$. We enumerate these arcs from right to left and denote 
by $\widetilde{m_j}$ its size. 
Suppose that the $i_1$-th arrow and the $i_2$-th arrow form 
the arc of size $\widetilde{m_j}$.
Let $D'$ be a diagram obtained from $D$ by connecting 
the $i_2$-th site and the $N$-th site via an arc and 
putting two down arrows at the $(N-2m_{1}+1)$-th 
and the $i_1$-th sites.
We have $e_N(D')=-q^{-1}D+\cdots$.
The contribution is given by 
\begin{multline}
-q^{-d}Q^{-1}\Psi_D\sum_{j}\frac{[L+1][\widetilde{m_j}]
[m_1]}{[1+\sum_{k=1}^{j-1}\widetilde{m_k}][1+\sum_{k=1}^{j}\widetilde{m_k}]
[m_1+1]}\prod_{i=1}^{I}\frac{[1+m_i+v_i]}{[1+v_{i+1}]} \\
=-q^{-d}Q^{-1}\Psi_D\frac{[L+1][m_1-1]}{[m_1+1]}
\prod_{i=1}^{I}\frac{[1+m_i+v_i]}{[1+v_{i+1}]}
\end{multline}
where we have used Lemma~\ref{lemma-app2} and $\sum_{j}\widetilde{m_j}=m_1-1$.

Let $D'$ be a diagram obtained from $D$ by changing the rightmost 
up arrow to a down arrow. 
We have $e_N(D')=-q^{-\sum_{j=1}^{I}n_j-2}D+\cdots$.
The contribution is given by 
\begin{eqnarray}
-q^{-\sum_{j=1}^{I}n_j-1-d}Q^{-1}\frac{[N'+M']}{[N_{\downarrow}+M+1]}
\Psi_D.
\end{eqnarray}
Suppose that the $j_1$-th arrow and the $j_2$-th arrow form an 
outer arc of size $m_j$ for $2\le j\le I$.
Let $D'$ be a diagram obtained from $D$ by changing the outer 
arc of size $m_i$ to two down arrows. 
We have $e_N(D')=-q^{-\sum_{k=1}^{j-1}n_k-2}D+\cdots$.
The contribution is given by 
\begin{multline}
\label{eNA-2-2}
-q^{-(d-1)}Q^{-1}\Psi_D\sum_{j=2}^{I}q^{-\sum_{k=1}^{j-1}n_k-2}
\frac{[L+1][m_j]}{[1+v_{j}][1+v_{j+1}]}
\prod_{k=i+1}^{I}
\frac{[1+m_k+v_{k}]}{[1+v_{k+1}]} \\
=-q^{-d}Q^{-1}\Psi_{D}[L+1]\left\{
\prod_{i=1}^{I}\frac{[1+m_i+v_i]}{[1+v_{i+1}]}
-\frac{q^{M}}{[1+v_{J+1}]}-\frac{q^{-1}[m_1]}{[m_1+1]}
\prod_{i=1}^{I}\frac{[1+m_i+v_i]}{[1+v_{i+1}]}
\right\}
\end{multline}
where we have used Lemma~\ref{lemma-app-eN2}.
By a straightforward calculation, one can show that the 
sum of Eqns.(\ref{eNA-2-1}) to (\ref{eNA-2-2}) 
is $Q^{-1}\Psi_{D}$. 
This cancels the contribution of $D$ itself, which 
implies $e_N\Psi=0$.

\paragraph{\bf Case 3}
Let $D$ be a diagram depicted as Eqn.(\ref{Diagram0}).
The $D$-component of $\Psi$ is explicitly given by 
\begin{eqnarray*}
\Psi_{D}&=&q^{d(d-1)/2}Q^{d}
\prod_{B\in S}[m_B]^{-1}
\prod_{i=1}^{I}\frac{[\sum_{j=1}^{i}(n_j+m_j)]!}{[n_i+\sum_{j=1}^{i-1}(n_j+m_j)]!}
\frac{[N_{\uparrow}+M+N'+M']!}{[N_{\uparrow}+M]!} \\
&&\times\prod_{i=1}^{J+1}\frac{[\sum_{j=1}^{i-1}(n'_j+m'_j)]!}
{[n'_{i}+\sum_{j=1}^{i-1}(n'_j+m'_j)]!}
\end{eqnarray*}
where $N_{\uparrow}=\sum_{i=1}^{I+1}n_i$, $M=\sum_{i=1}^{I}m_i$, 
$N'=\sum_{i=1}^{J+1}n'_i$, $M'=\sum_{i=1}^{J}m'_i$ and $d=N_{\uparrow}+M+M'$.

We have two cases for $D$: i) $n_1\ge2$ and ii) $n_1=1$. 
Below, we consider the case i) only since one can prove Proposition for case ii) 
by a similar argument.

We have ten types of contributions for the $D$-component of 
$e_N\Psi$ as follows.
\begin{enumerate}[a)] 
\item
Since $e_N(D)=-QD+\ldots$, we have a contribution from $D$ itself, 
that is $-Q\Psi_{D}$. 

\item
Let $D'$ be a diagram obtained from $D$ by changing the outer arc 
of size $m'_i$ to two down arrows.
We have $e_N(D')=q^{-\sum_{j=1}^{i}n'_j}D+\ldots$.

\item
Let $D'$ be a diagram obtained from $D$ by changing the $N_{\uparrow}$-th 
(from left) up arrow to a down arrow.
We have $e_N(D')=q^{-N'}D+\ldots$.

\item
Let $D'$ be a diagram obtained from $D$ by changing the first and 
the second (from right) down arrows to an arc.
We have $e_N(D')=D+\ldots$.

\item
Let $D'$ be a diagram obtained from $D$ by changing the first and 
the second (from right) down arrows to an arc and by changing 
the $N_{\uparrow}$-th and the $(N_{\uparrow}-1)$-th (from left) up 
arrows to two down arrows.
We have $e_N(D')=-q^{-2N'-1}D+\ldots$. 

\item
Let $D'$ be a diagram obtained from $D$ by changing the first and 
the second (from right) down arrows to an arc, by changing the outer 
arc of size $m'_i$ to two down arrows and by changing the 
$N_{\uparrow}$-th (from left) up arrow to a down arrow.
We have $e_N(D')=-(1+q^{2})q^{-N'-2-\sum_{j=1}^{i}n_j}D+\ldots$. 

\item
Let $D'$ be a diagram obtained from $D$ by changing the outer 
arc of size $m'_i$ to two down arrows and by changing the outer 
arc of size $m'_j$ ($j<i$) to two down arrows. 
We have 
$e_N(D')=-(1+q^{-2})q^{-\sum_{k=1}^{i}n'_k-\sum_{k=1}^{j}n'_k}D+\ldots$.

\item
Let $D''$ be a diagram obtained from $D$ by changing the first 
and the second (from right) down arrows to an arc and by changing 
the outer arc of size $m'_{i}$ to two down arrows. 
Suppose that $m'_{1,j}$, $1\le j\le r$, be the size of outer arcs of $D''$ 
which is inside of the outer arc of the size $m_1$ in $D$. 
Let $D'$ be a diagram obtained from $D''$ by changing the outer arc 
of size $m'_{1,j}$ to two down arrows.
Then, we have $e_N(D')=-q^{-1-2\sum_{k=1}^{i}n'_k}D+\ldots$.

\item
Let $D'$ be a diagram obtained from $D$ by changing the first and 
the second (from right) down arrows to an arc and by changing the 
$N_{\uparrow}$-th (from left) up arrow to a down arrow.
We have $e_N(D')=q^{-N'}(Q-Q^{-1})D+\ldots$.

\item
Let $D'$ be a diagram obtained from $D$ by changing the first and 
the second (from right) down arrows to an arc and by changing 
the outer arc of size $m_i$ to two down arrows.
We have $e_N(D')=q^{-\sum_{j=1}^{i}n'_{j}}(Q-Q^{-1})D+\ldots$.
\end{enumerate}

The sum of contributions from a) to j) is written as 
$A_1Q\Psi_{D}+A_{-1}Q^{-1}\Psi_D$. 
We will show that $A_1=A_{-1}=0$.
Set $v_i:=\sum_{j=1}^{i-1}(n'_j+m'_j)$, $w_i:=1+n'_{i}+v_{i}$, 
$d_{j}:=1-N_{\uparrow}-M-M'-\sum_{k=1}^{j}n'_k$
and $L:=1+N_{\uparrow}+M+n'_{I+1}+v_{I+1}$. 

The contribution to $A_{1}$ is summarized as follows.
From a), we have $-1$. 
From d) we have 
\begin{eqnarray}
\label{PsieNAQ1}
\frac{q^{d}}{[L-1]}
\prod_{i=1}^{J+1}\frac{[n'_{j}+v_{j}]}{[v_j]}.
\end{eqnarray}
From i), we have
\begin{eqnarray}
\label{PsieNAQ2}
q^{-N'}\frac{[N_{\uparrow}+M]}{[L-1][M'+N']}
\prod_{i=1}^{J+1}\frac{[n'_{j}+v_{j}]}{[v_j]}.
\end{eqnarray}
From j), we have 
\begin{eqnarray}
\label{PsieNAQ3}
\sum_{i=1}^{J}q^{-\sum_{j=1}^{i}n'_j}
\frac{[m'_i]}{[v_{i+1}][v_{i}]}
\prod_{j=1}^{i-1}\frac{[n'_{j}+v_{j}]}{[v_j]}.
\end{eqnarray}
Applying Lemma~\ref{lemma-app-eNA} to the sum of Eqns.(\ref{PsieNAQ1}),  
(\ref{PsieNAQ2}) and (\ref{PsieNAQ3}), we obtain 
$1$ which cancels the contribution from a). 
Thus we have $A_1=0$.

We have eight types of contributions to $A_{-1}$.
From b), we have 
\begin{eqnarray}
\label{eNPsiA-1}
\sum_{i=1}^{J}q^{d_i}\frac{[L][m'_i]}{[w_{i}][w_{i+1}]}
\prod_{j=i+2}^{J+1}\frac{[1+v_{j}]}{[w_j]}.
\end{eqnarray}
From c), we have 
\begin{eqnarray}
\label{eNPsiA-2}
q^{d_{J+1}}\frac{[N_{\uparrow}+M]}{[w_{J+1}]}.
\end{eqnarray}
From e), we have
\begin{eqnarray}
\label{eNPsiA-3}
-q^{d'}\frac{[N_{\uparrow}+M][N_{\uparrow}+M-1]}{[L-1][w_{J+1}][v_{J+1}]}
\prod_{i=1}^{J}\frac{[n'_i+v_{i}]}{[v_{i}]},
\end{eqnarray}
where $d':=d_{I+1}-N'-1$.
From f), we have 
\begin{eqnarray}
\label{eNPsiA-4}
-(1+q^{-2})\sum_{i=1}^{J}
q^{d_i-N'}
\frac{[N_{\uparrow}+M][m'_i]}{[v_{i}][v_{i+1}][w_{J+1}]}
\prod_{j=1}^{i-1}\frac{[n'_j+v_{j}]}{[v_j]}.
\end{eqnarray}
From g), we have 
\begin{multline}
\label{eNPsiA-5}
-(1+q^{-2})\sum_{i=2}^{J}q^{d_i}\frac{[L][m'_i]}{[w_{i}][w_{i+1}]}
\prod_{k=i+2}^{J+1}\frac{[1+v_{k}]}{[w_{k}]} 
\sum_{j=1}^{i-1}\frac{q^{-\sum_{k=1}^{j}n'_k}[m'_j]}{[v_{j}][v_{j+1}]}
\prod_{k=1}^{j-1}\frac{[n'_k+v_{k}]}{[v_k]} \\
=-\sum_{i=1}^{J}(1+q^{-2})q^{d_i}\frac{[m'_i][L]}{[w_i][w_{i+1}]}
\prod_{j=i+2}^{J+1}\frac{[1+v_{j}]}{[w_j]} \\
+\sum_{i=1}^{J}(1+q^{-2})q^{d_i+\sum_{j=1}^{i-1}m'_j}\frac{[m'_i][L]}{[w_i][w_{i+1}]}
\prod_{j=i+2}^{J+1}\frac{[1+v_{j}]}{[w_{j}]}
\prod_{k=1}^{i-1}\frac{[n'_k+v_k]}{[v_{k+1}]}
\end{multline}
where we have used Lemma~\ref{lemma-app-eNA}.
From h), we have
\begin{multline}
\label{eNPsiA-6}
\sum_{i=1}^{J}
\frac{q^{d_{i}-\sum_{j=1}^{i}n'_j-1}[L][m'_{i}]}
{[n'_i+v_i][1+n_{i+1}]}
\prod_{k=1}^{i}\frac{[n'_k+v_k]}{[v_k]}
\prod_{k=i+1}^{J+1}\frac{[1+v_k]}{[w_k]}
\left(
\sum_{l}
\frac{[m_{1,l}]}{[w_{i,j-1}][w_{i,j}]}
\right) \\
=
-\sum_{i=1}^{J}\frac{q^{d_{i}-\sum_{j=1}^{i}n'_j-1}[L][m'_{i}][m'_i-1]}
{[n'_i+v_{i}][1+v_{i+1}][w_i][v_{i+1}]}
\prod_{k=1}^{i}\frac{[n'_k+v_k]}{[v_{k}]}
\prod_{k=i+1}^{J+1}\frac{[1+v_k]}{[w_k]}
\end{multline}
where $w_{i,j}:=w_i+\sum_{k\le j}m_{1,k}$, we have used 
Lemma~\ref{lemma-app2} and $\sum_{l}m_{i,l}=m_{i}-1$.
From i), we have
\begin{eqnarray}
\label{eNPsiA-7}
-q^{-N'}\frac{[N_{\uparrow}+M]}{[L-1][v_{I+1}]}
\prod_{j=1}^{I}\frac{[n'_j+v_j]}{[v_j]}.
\end{eqnarray}
From j), we have 
\begin{eqnarray}
\label{eNPsiA-8}
-\sum_{i=1}^{J}q^{-\sum_{j=1}^{i}n'_j}\frac{[m'_i]}{[v_i][v_{i+1}]}
\prod_{k=1}^{i-1}\frac{[n'_k+v_k]}{[v_k]}.
\end{eqnarray}
We apply Lemma~\ref{lemma-app-eNA} to $q^2/(1+q^{2})$ times 
Eqn.(\ref{eNPsiA-4}) and the sum of 
Eqns.(\ref{eNPsiA-3}) and (\ref{eNPsiA-7}). 
The result cancels the contribution of Eqn.(\ref{eNPsiA-2}).  

The sum of Eqn.(\ref{eNPsiA-8}) and $1/(1+q^{2})$ times 
Eqn.(\ref{eNPsiA-4}) becomes 
\begin{multline}
\label{eNPsiA-9}
-\sum_{i=1}^{J}\frac{q^{-(M+N_{\uparrow}+\sum_{j=1}^{i}n'_j)}[m'_i][L]}
{[v_i][v_{i+1}][w_{J+1}]}
\prod_{k=1}^{i-1}\frac{[n'_k+v_k]}{[v_k]} \\
=-q^{-N_{\uparrow}-M}\frac{[L]}{[w_{J+1}]}
+q^{-(N_{\uparrow}+M-\sum_{j=1}^{I}m'_j)}\frac{[L]}{[w_{J+1}]}
\prod_{k=1}^{J}\frac{[n'_k+v_k]}{[v_{k+1}]}.
\end{multline}
where we have used Lemma~\ref{lemma-app-eNA}.
The sum of Eqn.(\ref{eNPsiA-1}) and the first term of 
the right hand side of Eqn.(\ref{eNPsiA-5}) is 
\begin{eqnarray}
\label{eNPsiA-10}
-q^{-N_{\uparrow}-M-\sum_{j=1}^{J}m'_j}[L]
\prod_{i=1}^{I+1}\frac{[1+v_j]}{[w_{j}]}
+q^{-N_{\uparrow}-M}\frac{[L]}{[w_{J+1}]}.
\end{eqnarray}
where we have used Lemma~\ref{lemma-app-eNA4}. 
The sum of the second term of the right hand side of 
Eqn.(\ref{eNPsiA-5}) and Eqn.(\ref{eNPsiA-6}) 
is given by 
\begin{multline}
\label{eNPsiA-11}
q^{-N_{\uparrow}-M-\sum_{i=1}^{J}m'_j+1}\sum_{i=1}^{J}
q^{-\sum_{k=1}^{i}n'_k}
\frac{[L][m'_i]}{[w_i][w_{i+1}]}
\prod_{k=i+2}^{J+1}\frac{[1+v_k]}{[w_k]}
\prod_{k=1}^{i-1}\frac{[n'_k+v_k]}{[v_{k+1}]} \\
\times\left\{
(1+q^{-2})q^{\sum_{j=1}^{i-1}m'_j}-q^{-\sum_{j=1}^{i}n'_j-1}
\frac{[m'_i-1]}{[v_{i+1}]}
\right\}
\end{multline}
By Lemma~\ref{lemma-app-eNA3}, the sum of Eqns.(\ref{eNPsiA-9}),
(\ref{eNPsiA-10}) and (\ref{eNPsiA-11}) is zero, which 
implies $A_{-1}=0$. 
This completes the proof.
\end{proof}

\begin{prop}
\label{prop-e0APsi}
We have $e_0\Psi=0$ at $qQQ_{0}=1$.
\end{prop}
\begin{proof}
We compute the $D$-component of $e_0\Psi$ at the specialization 
(\ref{def-icond}). 
We have three cases for the leftmost arrow $a$ of a diagram $D$:
1) $a$ is a down arrow, that is, $N_{\uparrow}=0$ 2) $a$ is an up arrow, and 
3) $a$ is a down arrow forming an arc.

\paragraph{\bf Case 1}
Let $D$ be a diagram depicted as 
\begin{eqnarray*}
\tikzpic{-0.5}{
\draw(0,0)node[anchor=north]
{$\underbrace{\downarrow\cdots\downarrow}_{n_1}$};
\draw(0.8,-0.2)..controls(0.8,-0.8)and(1.6,-0.8)..(1.6,-0.2);
\draw(1.2,-0.8)node[anchor=north]{$m_1$};
\draw(2.3,0)node[anchor=north]{$\downarrow\cdots\downarrow$};
\draw(3.1,-0.2)..controls(3.1,-0.8)and(3.7,-0.8)..(3.7,-0.2);
\draw(3.4,-0.8)node[anchor=north]{$m_I$};
\draw(4.4,0)node[anchor=north]
{$\underbrace{\downarrow\cdots\downarrow}_{n_{I+1}}$};
}.
\end{eqnarray*}
We have two cases for $D$: a) $n_1=1$, and b) $n_2\ge2$.
We consider only the Case a since one can apply a similar 
argument to Case b.

\paragraph{Case 1-a}
We have four types of contributions to the 
$D$-component of $e_{0}\Psi$.
Since we have $e_0(D)=-Q_{0}^{-1}D+\ldots$, the contribution 
to the $D$-component of $e_{0}\Psi$ is $-Q_{0}^{-1}\Psi_{D}$. 

Let $D'$ be a diagram obtained from $D$ by reversing the leftmost 
down arrow to an up arrow.
We have $e_0(D)=D+\ldots$. 
The contribution is $q^{|S|}Q[|S|+N_{\downarrow}]\Psi_{D}$.

We denote by $D'$ a diagram obtained from $D$ by connecting the 
first and the second (from left) up arrows via an arc.
We have $e_{0}(D)=-q^{-1}D+\cdots$.
The contribution is 
\begin{eqnarray*}
-q^{|S|-1}Q\frac{[|S|+N_{\downarrow}-m_1-1]}{[m_1+1]}\Psi_{D}.
\end{eqnarray*}

Let $D'$ be a diagram obtained from $D$ by changing 
the outer arc of size $m_1$ to two up arrows and 
by connecting the first and the second (from left) 
arrows via an arc.
We have $e_0(D)=-q^{-1}D+\cdots$.
The contribution is given by 
\begin{eqnarray*}
-q^{|S|-1}Q\frac{[m_1][|S|+N_{\downarrow}]}{[m_1+1]}.
\end{eqnarray*}
The sum of four contributions above is 
\begin{eqnarray*}
Q_{0}^{-1}(q^{N-1}QQ_{0}-1)\Psi_{D},
\end{eqnarray*}
which implies the $D$-component of $e_0\Psi$ is zero at 
the specialization (\ref{def-icond}).

\paragraph{\bf Case 2}
Let $D$ be a diagram depicted as Eqn.(\ref{Diagram0}).
Set $N_{\uparrow}:=\sum_{i=1}^{I+1}n_{i}$, 
$N_{\downarrow}:=\sum_{i=1}^{J+1}n'_{i}$,  
$M:=\sum_{i=1}^{I}m_i$, $M':=\sum_{i=1}^{J}m'_i$, 
$d:=|S|+N_{\uparrow}$ and $v_{i}:=\sum_{j=1}^{i}(n_j+m_j)$. 
We will prove Proposition in the case of $n_1\ge2$ since one 
can apply a similar argument to the case of $n_1=1$.

We have ten types of contributions to the $D$-component of 
$e_0\Psi$ as follows:
\begin{enumerate}[a)]
\item 
Since $e_0(D)=-Q_{0}D+\cdots$, we have a contribution from $D$ itself, 
that is, $-Q_{0}\Psi_{D}$.

\item
Let $D'$ be a diagram obtained from $D$ by connecting the first and 
the second (from left) up arrows via an arc.
We have $e_0(D')=D+\cdots$. 
The contribution is 
\begin{eqnarray}
\label{e0APsi2}
q^{-(d-1)}Q^{-1}\frac{[N_{\uparrow}+M]}{[N_{\uparrow}+N_{\downarrow}+M+M']}
\prod_{i=1}^{I}\frac{[n_i+v_{i-1}]}{[v_{i}]}\Psi_{D}.
\end{eqnarray}

\item
Let $D'$ be a diagram obtained from $D$ by connecting the first and second 
(from left) up arrows to an arc and by flipping the $N_{\downarrow}$-th 
(from right) down arrow to an up arrow.
We have $e_{0}(D')=q^{-N_{\uparrow}}(Q_{0}-Q_{0}^{-1})+\cdots$.
The contribution is given by 
\begin{eqnarray}
\label{e0APsi3}
q^{-N_{\uparrow}}(Q_{0}-Q_{0}^{-1})
\frac{[N_{\downarrow}+M']}{[N_{\uparrow}+N_{\downarrow}+M+M']}
\prod_{i=1}^{I}\frac{[n_i+v_{i-1}]}{[v_{i}]}\Psi_{D}.
\end{eqnarray}

\item
Let $D'$ be a diagram obtained from $D$ by connecting the first and second 
(from left) up arrows to an arc and by flipping the $N_{\downarrow}$-th 
and the $(N_{\downarrow}-1)$ (from right) down arrows to two up arrows.
We have $e_{0}(D')=-q^{-2N_{\uparrow}-1}D+\cdots$. 
The contribution is 
\begin{eqnarray}
\label{e0APsi4}
-q^{d-2N_{\uparrow}-1}Q
\frac{[N_{\downarrow}+M'][N_{\downarrow}+M'-1]}
{[N_{\uparrow}+M+1][d+N_{\downarrow}]}
\prod_{i=1}^{I}\frac{[n_i+v_{i-1}]}{[v_{i}]}\Psi_{D}.
\end{eqnarray}

\item
Let $D'$ be a diagram obtained from $D$ by changing the outer
arc of size $m_{i}$ to two up arrows. 
We have $e_0(D')=q^{-\sum_{j=1}^{i}n_{j}}D+\cdots$.
The contribution is 
\begin{multline}
\label{e0APsi5}
\sum_{i=1}^{I}q^{d-\sum_{j=1}^{i}n_{j}}Q
\frac{[m_i]}{[1+n_{i}+v_{i-1}]}
\frac{[d+N_{\downarrow}+1]}{[N_{\uparrow}+M+1]}
\prod_{i=1}^{I}\frac{[1+v_{i}]}{[1+n_{j}+v_{j-1}]}\Psi_{D} \\
=q^{d}Q\Psi_{D}\frac{[d+N_{\downarrow}+1]}{[N_{\uparrow}+M+1]}
\left(q\prod_{j=1}^{I}\frac{[1+v_{j}]}{[1+n_j+v_{j-1}]}
-q^{1+\sum_{i=1}^{I}m_{i}}\right),
\end{multline}
where we have used Lemma~\ref{lemma-app-16}.

\item 
Let $D'$ be a diagram obtained from $D$ by changing the 
$N_{\downarrow}$-th (from right) down arrow to an up 
arrow.
We have $e_{0}(D')=q^{-N_{\uparrow}}D+\cdots$.
The contribution is 
\begin{eqnarray}
\label{e0APsi6}
q^{d-N_{\uparrow}}Q\frac{[N_{\downarrow}+M']}{[N_{\uparrow}+M+1]}\Psi_{D}.
\end{eqnarray}

\item
Let $D'$ be a diagram obtained from $D$ by connecting the first and 
the second (from left) up arrows via an arc and flipping the 
$N_{\downarrow}$-th (from right) down arrow to an up arrow and by 
changing the outer arc of size $m_i$ to two up arrows.
We have $e_{0}(D')=-(1+q^{2})q^{-N_{\uparrow}-2-\sum_{j=1}^{i}n_{j}}D+\cdots$.
The contribution is given by 
\begin{multline}
\label{e0APsi7}
-(1+q^{2})q^{d-N_{\uparrow}-2}Q\Psi_{D}
\sum_{i=1}^{I}q^{-\sum_{j=1}^{i}n_{j}}
\frac{[m_i][N_{\downarrow}+M']}{[v_{i}][N_{\uparrow}+M+1]}
\prod_{j=1}^{i-1}\frac{[n_j+v_{j-1}]}{[v_{j}]} \\
=-(1+q^{2})q^{d-N_{\uparrow}-2}Q\Psi_{D}
\frac{[N_{\downarrow}+M']}{[N_{\uparrow}+M+1]}
\left(1-q^{\sum_{i=1}^{I}m_{i}}\prod_{i=1}^{I}\frac{[n_i+v_{i-1}]}{[v_{i}]}\right),
\end{multline}
where we have used Lemma~\ref{lemma-app-15}.

\item 
Let $D'$ be a diagram obtained from $D$ by connecting the first and the 
second (from left) up arrows via an arc and changing the outer arcs of 
size $m_i$ and $m_{j}$ ($i<j$) to four up arrows.
We have $e_{0}(D')=-(1+q^{-2})q^{-\sum_{k=1}^{i}n_i-\sum_{k=1}^{j}n_j}D+\cdots$. 
The contribution is given by 
\begin{multline}
\label{e0APsi8}
-(1+q^{-2})q^{d}Q\Psi_{D}\frac{[d+N_{\downarrow}+1]}{[N_{\uparrow}+M+1]}
\sum_{j=2}^{I}q^{-\sum_{k=1}^{j}n_k}\frac{[m_j]}{[1+n_j+v_{j-1}]}
\prod_{l=j+1}^{I}\frac{[1+v_{l}]}{[1+n_{l}+v_{l-1}]} \\
\times \sum_{i=1}^{j-1}q^{-\sum_{k=1}^{i}n_k}\frac{[m_i]}{[v_{i}]}
\prod_{l=1}^{i-1}\frac{[n_l+v_{l-1}]}{[v_l]} \\
=-(1+q^{-2})q^{d}Q\Psi_{D}\frac{[d+N_{\downarrow}+1]}{[N_{\uparrow}+M+1]}
\left\{q\prod_{j=1}^{I}\frac{[1+v_{j}]}{[1+n_j+v_{j-1}]}
-q^{1+\sum_{i=1}^{I}m_{i}}
\right. \\
\left.
-\sum_{j=1}^{I}q^{-\sum_{k=1}^{j}n_k+\sum_{k=1}^{j-1}m_k}\frac{[m_j]}{[1+n_j+v_{j-1}]}
\prod_{l=j+1}^{I}\frac{[1+v_{l}]}{[1+n_{l}+v_{l-1}]}
\prod_{l=1}^{i-1}\frac{[n_l+v_{l-1}]}{[v_{l}]}\right\}.
\end{multline}
where we have used Lemma~\ref{lemma-app-15} and Lemma~\ref{lemma-app-16}.

\item
Let $D''$ be a diagram obtained from $D$ by connecting the first and the second 
(from left) up arrows via an arc and by changing the outer arc of size $m_i$ 
to two up arrows. 
Suppose that $\tilde{m}_{i,j}$, $1\le j\le r$ be the size of outer arcs of 
$D''$ which are inside of the outer arc of size $m_i$ in $D$.
Let $D'$ be a diagram obtained from $D''$ by changing the outer arc of size 
$\tilde{m}_{i,j}$ to two up arrows.
We have $e_0(D')=-q^{-1-2\sum_{k=1}^{i}n_k}D+\cdots$.
The contribution is given by 
\begin{multline}
\label{e0APsi9}
-q^{d-1}Q\Psi_{D}\frac{[d+N_{\downarrow}+1]}{[N_{\uparrow}+M+1]}
\sum_{i=1}^{I}q^{-2\sum_{k=1}^{i}n_k}[m_i]
\prod_{j=1}^{i-1}\frac{[n_j+v_{j-1}]}{[v_{j}]}
\prod_{j=i+1}^{I}\frac{[1+v_{j}]}{[1+n_j+v_{j-1}]} \\
\times\sum_{p=1}^{r}\frac{[\tilde{m}_{i,p}]}
{[1+n_p+v_{p-1}+\sum_{k=1}^{p-1}\tilde{m}_{i,k}][1+n_p+v_{p-1}+\sum_{k=1}^{p}\tilde{m}_{i,k}]}\\
=-q^{d-1}Q\Psi_{D}\frac{[d+N_{\downarrow}+1]}{[N_{\uparrow}+M+1]}
\sum_{i=1}^{I}\frac{q^{-2\sum_{k=1}^{i}n_k}[m_i][m_i-1]}{[1+n_i+v_{i-1}][v_{i}]}
\prod_{j=1}^{i-1}\frac{[n_j+v_{j-1}]}{[v_{j}]}
\prod_{j=i+1}^{I}\frac{[1+v_{j}]}{[1+n_j+v_{j-1}]},
\end{multline}
where we have used Lemma~\ref{lemma-app2}.

\item
Let $D'$ be a diagram obtained from $D$ by connecting the first and 
the second (from left) up arrows via an arc and by changing the 
outer arc of size $m_i$ to two up arrows.
We have $e_0(D')=q^{-\sum_{k=1}^{i}n_k}(Q_0-Q_{0}^{-1})D+\cdots$.
The contribution is given by 
\begin{multline}
\label{e0APsi10}
(Q_0-Q_{0}^{-1})\Psi_{D}\sum_{i=1}^{I}
q^{-\sum_{k=1}^{i}n_k}\frac{[m_i]}{[v_i]}
\prod_{j=1}^{i-1}\frac{[n_j+v_{j-1}]}{[v_j]} \\
=(Q_0-Q_0^{-1})\Psi_{D}\left(
1-q^{\sum_{i=1}^{I}m_{i}}\prod_{i=1}^{I}\frac{[n_i+v_{i-1}]}{[v_{i}]}
\right)
\end{multline}
\end{enumerate}
Note that one can apply Lemma~\ref{lemma-app-17} to the third term of the right 
hand side of Eqn.(\ref{e0APsi8}) and the right hand side of 
Eqn.(\ref{e0APsi9}).
The sum of contributions from a) to j) is 
\begin{eqnarray*}
(q^{N-1}QQ_{0}-1)
\left(
Q_0^{-1}-(q^{d'}Q_0^{-1}+q^{1-d}Q^{-1})\frac{[N_{\uparrow}+M]}{[d+N_{\downarrow}]}
\prod_{i=1}^{I}\frac{[n_{i}+v_{i-1}]}{[v_{i}]}
\right)\Psi_{D},
\end{eqnarray*}
where $d'=N_{\downarrow}+M+M'$ and $N=N_{\uparrow}+N_{\downarrow}+2M+2M'$. 
The sum becomes zero under the specialization (\ref{def-icond}).

\paragraph{\bf Case 3}
Let $D$ be a diagram depicted as 
\begin{eqnarray}
\raisebox{-0.8\totalheight}{
\begin{tikzpicture}
\draw(0,0)..controls(0,-0.8)and(1,-0.8)..(1,0);
\draw(0.5,-0.8)node{size $m_1$};
\end{tikzpicture}}\! \! 
\underbrace{\uparrow\ldots\uparrow}_{n_1}\! \!
\ldots\uparrow \! \! \! 
\raisebox{-0.8\totalheight}{
\begin{tikzpicture}
\draw(0,0)..controls(0,-0.8)and(1,-0.8)..(1,0);
\draw(0.5,-0.8)node{size $m_I$};
\end{tikzpicture}}\! \! 
\underbrace{\uparrow\ldots\uparrow}_{n_{I}}
\underbrace{\downarrow\ldots\downarrow}_{n'_{J+1}}
\! \!\!
\raisebox{-0.8\totalheight}{
\begin{tikzpicture}
\draw(0,0)..controls(0,-0.8)and(1,-0.8)..(1,0);
\draw(0.5,-0.8)node{size $m'_J$};
\end{tikzpicture}} \! \! \!
\downarrow\ldots\downarrow \!\!\!
\raisebox{-0.8\totalheight}{
\begin{tikzpicture}
\draw(0,0)..controls(0,-0.8)and(1,-0.8)..(1,0);
\draw(0.5,-0.8)node{size $m'_1$};
\end{tikzpicture}} \!\!
\underbrace{\downarrow\ldots\downarrow}_{n'_1}.
\end{eqnarray}
Set $N_{\uparrow}:=\sum_{i=1}^{I}n_i$, $M:=\sum_{i=1}^{I}m_{i}$, 
$N_{\downarrow}:=\sum_{i=1}^{J+1}n'_i$, $M':=\sum_{i=1}^{J}m'_{i}$, 
$d=N_{\uparrow}+M+M'$ and $v_{i}:=\sum_{j=1}^{i}(n_j+m_j)$.
We have five types of contributions to the $D$-component of $e_0\Psi$.
\begin{enumerate}[a)]
\item
Since $e_0(D)=-Q_{0}^{-1}D+\cdots$, we have a contribution from $D$, 
which is $-Q_{0}^{-1}\Psi_{D}$. 

\item
Let $D'$ be a diagram obtained from $D$ by changing the outer arc 
of size $m_1$ to two up arrows.
We have $e_0(D')=D+\cdots$.
The contribution is given by 
\begin{eqnarray*}
q^{d}Q\Psi_{D}\frac{[m_i][d+N_{\downarrow}+1]}{[N_{\uparrow}+M+1]}
\prod_{i=2}^{I}\frac{[1+m_i+v_{i-1}]}{[1+v_{i-1}]}.
\end{eqnarray*}

\item
Let $D'$ be a diagram obtained from $D$ by changing the outer arc 
of size $m_i$ to two up arrows.
We have $e_0(D')=-q^{-\sum_{k=1}^{i-1}-2}D+\cdots$.
The contribution is given by
\begin{multline}
-q^{d-2}Q\Psi_{D}\frac{[d+N_{\downarrow}+1]}{[N_{\uparrow}+M+1]}
\sum_{i=2}^{I}\frac{q^{-\sum_{k=1}^{i-1}n_k}[m_i]}{[1+v_{i-1}]}
\prod_{j=i+1}^{I}\frac{[1+m_j+v_{j-1}]}{[1+v_{j-1}]} \\
=-q^{d-2}Q\Psi_{D}\frac{[d+N_{\downarrow}+1]}{[N_{\uparrow}+M+1]}
\left(
q^{m_1+1}\prod_{i=2}^{I}\frac{[1+m_j+v_{j-1}]}{[1+v_{j-1}]}
-q^{1+\sum_{i=1}^{I}m_{i}}
\right)
\end{multline}
where we have used Lemma~\ref{lemma-app-16}.

\item
Let $D'$ be a diagram obtained from $D$ by connecting the 
first and the second (from left) up arrows via an arc and 
by flipping the $N_{\downarrow}$-th (from right) down arrow 
to an up arrow.
We have $e_{0}(D')=-q^{-N_{\uparrow}-2}D+\cdots$.
The contribution is 
\begin{eqnarray*}
-q^{d-N_{\uparrow}-2}Q\Psi_{D}
\frac{[N_{\downarrow}+M']}{[N_{\uparrow}+M+1]}.
\end{eqnarray*}

\item
Suppose that there are several arcs of depth two inside 
of the outer arc of size $m_1$. 
We denote by $\tilde{m}_{j}$, $1\le j\le r$, their sizes 
from left to right.
Suppose that the $i_1$-th and the $i_2$-th (from left) arrows
form the arc of size $\tilde{m}_{j}$. 
Let $D'$ be a diagram obtained by connecting the first and 
the $i_1$-th arrows via an arc and by putting two up arrows
at the $i_2$-th and the $2m_1$-th sites. 
We have $e_{0}(D')=-q^{-1}D+\cdots$.
The contribution is 
\begin{multline*}
-q^{d-1}Q\Psi_{D}\frac{[d+N_{\downarrow}+1][m_1]}{[N_{\uparrow}+M+1]}
\prod_{i=2}^{I}\frac{[1+m_{i}+v_{i-1}]}{[1+v_{i-1}]}
\sum_{j=1}^{r}\frac{[\tilde{m}_{j}]}
{[1+\sum_{k=1}^{j-1}\tilde{m}_{k}][1+\sum_{k=1}^{j}\tilde{m}_{k}]}\\
=-q^{d-1}Q\Psi_{D}\frac{[d+N_{\downarrow}+1][m_1-1]}{[N_{\uparrow}+M+1]}
\prod_{i=2}^{I}\frac{[1+m_{i}+v_{i-1}]}{[1+v_{i-1}]}
\end{multline*}
where we have used Lemma~\ref{lemma-app2} and $\sum_{j=1}^{r}\tilde{m}_{j}=m_1-1$.
\end{enumerate}
By a straightforward calculation, the sum of the contributions form a) to e)
is given by $(-Q_{0}^{-1}+Qq^{N-1})\Psi_{D}$ which vanishes 
at the specialization (\ref{def-icond}).
\end{proof}

\subsection{Type BI, BII, BIII and standard basis}

\begin{theorem}
\label{eonPsi-all}
For type BI, BII, BIII and standard basis, we have 
\begin{eqnarray}
\label{eonPsi-all-1}
e_{i}\Psi&=&0,\qquad 1\le i\le N, \\
\label{eonPsi-all-2}
e_0\Psi&=&0,\qquad \text{at } q^{N-1}QQ_{0}-1=0.
\end{eqnarray}
where $Q=q^{M}$ for type BI.
\end{theorem}
\begin{proof}
Let $\Psi^{\mathrm{Y}}$ be the eigenfunction $\Psi$ for 
type Y. 
Let $T^{\mathrm{Z}\leftarrow\mathrm{Y}}$ be the transition matrix from 
the Kazhdan--Lusztig basis of type Y to type Z. 
Since we have proved that $X\Psi^{\mathrm{Y}}=[Q;N]\Psi^{\mathrm{Y}}$ 
($Q=q^{M}$ for type BI) and the multiplicity is one in Section~\ref{sec-eigenX}, 
we have $\Psi^{\mathrm{Z}}=T^{\mathrm{Z\leftarrow Y}}\Psi^{\mathrm{Y}}$. 
From Proposition~\ref{prop-HA} and Proposition~\ref{prop-eNAPsi}, 
we have $e_i\Psi^{\mathrm{A}}=0$ for $1\le i\le N$. 
Multiplying $T^{\mathrm{Y\leftarrow A}}$ from left and plugging 
$\Psi^{A}=T^{\mathrm{A\leftarrow Y}}\Psi^{Y}$, we obtain 
$T^{\mathrm{Y\leftarrow A}}e_iT^{\mathrm{A\leftarrow Y}}\Psi^{Y}=0$. 
Since $T^{\mathrm{Y\leftarrow A}}e_iT^{\mathrm{A\leftarrow Y}}$ is the matrix 
expression of $e_i$ on the Kazhdan--Lusztig basis of type Y, 
we have $e_i\Psi^{Y}=0$.
Similarly, from Proposition~\ref{prop-e0APsi}, we have $e_0\Psi^{Y}=0$ at 
the specialization (\ref{def-icond}).

In the case of the standard basis, we define the transition matrix from 
a standard basis to a Kazhdan--Lusztig basis of type Z by 
$T^{\mathrm{Z\leftarrow0}}$. 
By a similar argument to the case of Kazhdan--Lusztig bases, we obtain
Eqns.(\ref{eonPsi-all-1}) and (\ref{eonPsi-all-2}). 
\end{proof}

\begin{remark}
\label{remark-KL}
In the proof of Theorem~\ref{eonPsi-all}, we do not need an explicit 
expression of the transition matrix. 
The entries of the transition matrix from the Kazhdan--Lusztig basis 
to the standard basis are nothing but Kazhdan--Lusztig polynomials.
Therefore, the relation $\Psi^{0}=T^{0\leftarrow\mathrm{Y}}\Psi^{Y}$ gives
highly non-trivial relations regarding Kazhdan--Lusztig polynomials.
\end{remark}

\subsection{\texorpdfstring{$\Psi$ as the ground state of the Hamiltonian}
{Psi as the ground state of the Hamiltonian}}
Let $A$ be a non-negative $N\times N$ square matrix. 
The matrix $A$ is called {\it irreducible} if for any 
$i,j$ there is a $k=k(i,j)$ such that $(A^{k})_{ij}>0$.
Let $\rho(A)$ denote the spectral radius of $A$.
Then, Perron--Frobenius Theorem for a non-negative and irreducible
matrix $A$ states that the eigenspace 
associated with $\rho(A)$ is one-dimensional, there exists
a unique eigenvector $\mathbf{x}=(x_1,\ldots,x_N)^{T}$
such that the entries of $\mathbf{x}$ are positive and 
$A\mathbf{x}=\rho(A)\mathbf{x}$.
For a general non-negative matrix $A$, we have 
\begin{lemma}[Lemma 6.2 in \cite{NacNgSta12}]
\label{lemma-PF}
Suppose that $\mathbf{x}$ be a positive vector such 
that $A\mathbf{x}=\lambda\mathbf{x}$ with some scalar $\lambda$.
Then, we have $\rho(A)=\lambda$.
\end{lemma}

Below, we set $q,Q>0$ and consider the Kazhdan--Lusztig bases 
for type BI and BIII.
From the explicit expression of the action of Temperley--Lieb 
algebra on the Kazhdan--Lusztig bases (see Section~\ref{sec-reps}),
the matrix $H'=-H^{1B}+t\mathbf{1}$ is a non-negative matrix for $t$ sufficiently
large.
Since $\Psi$ is a positive eigenvector of $H^{1B}$ with the eigenvalue zero, 
$\Psi$ is also the positive eigenvector of $H'$ with the eigenvalue $t$.
We apply Lemma~\ref{lemma-PF} to $H'$ and $\Psi$.
Therefore, the vector $\Psi$ is the eigenvector of $H'$ with the largest 
eigenvalue, which implies that $\Psi$ is the ground state of $H$.

In the case of type A and BII, the Hamiltonian $-H^{1B}+t\mathbf{1}$ can not
be a non-negative matrix. 
However, the bases of type A and BII can be obtained from type BIII
by the change of bases.  
Note that the spectrum of $H$ is invariant under the change of bases.
Therefore, the vector $\Psi$ for type A or BII is also the ground 
state of the Hamiltonian $H^{1B}$.

Recall that the two-boundary Temperley--Lieb Hamiltonian is given by 
$H^{1B}-a_{0}e_{0}$.
At $a_{0}=0$, $H^{2B}$ becomes $H^{1B}$.
The eigenfunction $\Psi$ is an eigenvector of $H^{2B}$ 
with an eigenvalue zero for arbitrary $a_{0}$.
We can regard the Hamiltonian $H^{2B}$ as a perturbation of $H^{1B}$.
Since $\Psi$ is the ground state of $H^{1B}$, $\Psi$ is also 
the ground state of $H^{2B}$ for $a_{0}$ sufficiently small.

\section{Integral structure and conjectures}
\label{sec-sum}

\subsection{Correlation functions}
Let $\Psi=|\Psi\rangle$ be the ground state of $H^{2B}$ and 
$\mathcal{O}$ be an observable which acts on $V_{1}^{\otimes N}$.
A correlation function is defined by
\begin{eqnarray*}
\langle\mathcal{O}\rangle:=
\frac{\langle\Psi|\mathcal{O}|\Psi\rangle}{\langle\Psi|\Psi\rangle}.
\end{eqnarray*}
We consider the case where $\mathcal{O}$ is a product 
of $\alpha_{i}:=(\sigma_{z}+1)/2$ and $\sigma^{\pm}_{i}$.
Let $I, J$ be the subsets of $\{1,2,\ldots N\}$ satisfying 
$I\cap J=\emptyset$.
In general, an observable is written as 
\begin{eqnarray*}
\mathcal{O}_{I,J}:=\prod_{i\in I}\alpha_{i}\prod_{j\in J}\beta_{j},
\end{eqnarray*}
where $\beta_{j}$ is $\sigma^{+}_{j}$ or $\sigma^{-}_{j}$.
We compute a correlation function in the standard basis.
From Definition~\ref{defPsi-SB}, $|\Psi^{0}\rangle$ is written 
as 
\begin{eqnarray}
|\Psi^{0}\rangle=w_{1}\otimes w_{2}\otimes \ldots\otimes w_{N},
\end{eqnarray}
where $w_{i}:=v_{-1}+q^{N-i}Qv_{1}$.
We have 
\begin{eqnarray*}
&&\alpha_i w_{i}=q^{N-i}Qv_{1}, \\
&&\sigma^{+}_{i}w_{i}=v_1, \quad \sigma_{i}^{-}w_i=q^{N-i}Qv_{-1}.
\end{eqnarray*}
Therefore, we obtain 
\begin{eqnarray}
\langle\mathcal{O}_{I,J}\rangle
=
\prod_{i\in I}\frac{q^{2(N-i)}Q^{2}}{1+q^{2(N-i)}Q^{2}}
\prod_{j\in J}\frac{q^{N-j}Q}{1+q^{2(N-j)}Q^{2}}.
\end{eqnarray}

\subsection{Positive integral structure}

\subsubsection{Type A} 
We have
\begin{lemma}
All components of $\Psi$ belong to $\mathbb{N}[q,q^{-1},Q]$. 
\end{lemma}
\begin{proof}
Let $D$ be a diagram depicted as Eqn.(\ref{Diagram0}). 
In the notation used in the proof of Theorem~\ref{thrm-PsiA}. 
The explicit expression (\ref{PsiA2}) is rewritten 
as 
\begin{eqnarray*}
\Psi_{D}=
q^{d(d-1)/2}Q^{d}
\prod_{i=1}^{I}
\genfrac{[}{]}{0pt}{}{\sum_{j=1}^{i}n_{i}+m_{i}}{m_i}
\prod_{i=1}^{J}
\genfrac{[}{]}{0pt}{}{\sum_{j=1}^{i}n'_{i}+m'_{i}}{m'_i}
\genfrac{[}{]}{0pt}{}{N_{\uparrow}+M+N_{\downarrow}+M'}{N_{\uparrow}+M}.
\end{eqnarray*}
Since a quantum binomial belongs to $\mathbb{N}[q,q^{-1}]$, we obtain 
$\Psi_{D}\in\mathbb{N}[q,q^{-1},Q]$.
\end{proof}

\subsubsection{Type BI}
We have 
\begin{prop}
All components of $\Psi$ belong to $\mathbb{N}[q,q^{-1}]$ and invariant 
under $q\rightarrow q^{-1}$.
\end{prop}
\begin{proof}
Since a component $\Psi_{D}$ contains only quantum integers and terms 
$(q^{i}+q^{-i})$ for some $i$ (see Definition~\ref{PsiBI}), 
it is invariant under $q\rightarrow q^{-1}$. 

Recall that we have two types of parabolic Kazhdan--Lusztig 
polynomials according to the choice of projection map (see e.g. 
\cite{Deo87,Shi14-1}).
As in Remark~\ref{remark-KL}, the transition matrix 
$T^{0\leftarrow\mathrm{BI}}$ from the 
Kazhdan--Lusztig basis to the standard basis is written in 
terms of parabolic Kazhdan--Lusztig polynomials. 
The inverse of $T^{0\leftarrow\mathrm{BI}}$, that is, 
$T^{\mathrm{BI}\leftarrow0}$, is also written 
in terms of another parabolic Kazhdan--Lusztig polynomials 
(see Theorem 6 in \cite{Shi14-1}).
The diagonal entries of the matrix $T^{\mathrm{BI}\leftarrow0}$ 
are one and it is an upper triangular matrices whose non-zero
entries are in $q^{-1}\mathbb{N}[q^{-1}]$. 
At $Q=q^{M}$, we have $\Psi^{0}_{\epsilon}\in\mathbb{N}[q]$.
Since $\Psi^{\mathrm{BI}}=T^{\mathrm{BI}\leftarrow0}\Psi^{0}$, 
we obtain $\Psi^{\mathrm{BI}}_{D}\in\mathbb{N}[q,q^{-1}]$.
\end{proof}

Recall that a diagram $D$ is characterized by a binary string
$b\in\{\pm\}^{N}$ of length $N$ (see Section~\ref{sec-Rep-KL}).
Given a binary string $b$, let $J_{D}$ be the set of 
positions of $+$ from right. 
We define 
\begin{eqnarray*}
d_{D}:=\sum_{j\in J_{D}}(j+M-1).
\end{eqnarray*}
For example, when $b=(+-+-)$ with $M=2$, we have 
$d_{D}=8$.
\begin{cor}
The component $\Psi_{D}$ has the leading term $q^{d_{D}}$ 
with the leading coefficient one. 
\end{cor}
\begin{proof}
Let $b=b_1\ldots b_{N}$ and $b'=b'_1\ldots b'_{N}$ be two binary 
strings in $\{\pm\}^{N}$.
We introduce the reversed lexicographic order, which is 
$b<b'$ if and only if $b_{j}=b'_{j}$ for $1\le j\le i-1$,  
$b_{i}=+$ and $b'_{i}=-$. 
This lexicographic order of binary strings induces a natural order 
of diagrams. 
If $D<D'$, then we have $d_{D}\ge d_{D'}$. 
The eigenvector satisfies $\Psi_{D}^{0}=q^{d_{D}}$.
Recall that the diagonal entries are one and other non-zero 
entries are in $q^{-1}\mathbb{N}[q^{-1}]$.
Since $\Psi^{\mathrm{BI}}=T^{\mathrm{BI}\leftarrow0}\Psi^{0}$, 
one easily show that the leading term is $q^{d_{D}}$ with 
the coefficient one. 
\end{proof}

\subsubsection{Type BII and BIII} 
We have
\begin{lemma}
All components of $\Psi$ belong to $\mathbb{N}[q,q^{-1},Q,Q^{-1}]$. 
\end{lemma}
\begin{proof}
We prove Lemma for type BII and type BIII separately.

\paragraph{Type BII}
Let $D$ be a diagram starting with $n_1$ up arrows, followed by 
an outer arc of size $m_1$, followed by $n_{2}$ up arrows, 
followed by an outer arc of size $m_{2}$, $\cdots$, 
followed by $n_{I+1}$ up arrows, followed by $p_{J+1}$ e-unpaired 
down arrows (and $p_{J+1}-1, p_{J+1}$ or $p_{J+1}+1$ o-unpaired 
down arrows), followed by an outer arc of size $m'_{J+1}$, 
followed by $p_{J}$ e-unpaired down arrows, followed by 
an outer arc of size $m'_{J}$, $\cdots$, and ending with  
$p_{1}$ e-unpaired down arrows.
Set $N_{\uparrow}=\sum_{i=1}^{I+1}n_{i}$, 
$M=\sum_{i=1}^{I}m_i$, $P=\sum_{i=1}^{J+1}p_{i}$ and 
$M'=\sum_{i=1}^{J+1}m'_i$.  
From Definition~\ref{PsiBII}, the product $N_{1}N_{2}N_{3}$ is 
written as 
\begin{eqnarray*}
\prod_{i=1}^{I}
\genfrac{[}{]}{0pt}{}{\sum_{j=1}^{i}n_{j}+m_{j}}{m_{i}}
\prod_{i=1}^{J}
\genfrac{[}{]}{0pt}{}{\sum_{j=1}^{i}m'_{j}+p_{j}}{m'_{i}}
\genfrac{[}{]}{0pt}{}{N_{\uparrow}+M+P+M'}{N_{\uparrow}+M}
\in\mathbb{N}[q,q^{-1}].
\end{eqnarray*}
Since $N_{4}\in\mathbb{N}[q,q^{-1},Q,Q^{-1}]$, we have 
$\Psi_{D}\in\mathbb{N}[q,q^{-1},Q,Q^{-1}]$.

\paragraph{Type BIII}
Let $D$ be a diagram starting with $n_1$ up arrows, followed by 
an outer arc of size $m_1$, followed by $n_{2}$ up arrows, 
followed by an outer arc of size $m_{2}$, $\cdots$, 
followed by $n_{I+1}$ up arrows, followed by $n'_{J+1}$ down 
arrows, followed by an outer arc of size $m'_J$, 
followed by $n'_{J}$ down arrows (with a circled integer), 
followed by an outer arc of size $m'_{J-1}$, $\cdots$, 
ending with $n'_1$ down arrows.  
Set $N_{\uparrow}=\sum_{i=1}^{I+1}n_{i}$, $M=\sum_{i=1}^{I}m_i$, 
$N_{\downarrow}=\sum_{i=1}^{J+1}n'_{i}$ and 
$M'=\sum_{i=1}^{J}m'_{i}$. 
From Definition~\ref{PsiBIII}, the product $N_1N_2N_3$ is 
\begin{eqnarray*}
\prod_{i=1}^{I}
\genfrac{[}{]}{0pt}{}{\sum_{j=1}^{i}n_{j}+m_{j}}{m_{i}}
\prod_{i=1}^{J}
\genfrac{[}{]}{0pt}{}{\sum_{j=1}^{i}m'_{j}+n'_{j}}{m'_{i}}
\genfrac{[}{]}{0pt}{}{N_{\uparrow}+M+N_{\downarrow}+M'}{N_{\uparrow}+M}
\in\mathbb{N}[q,q^{-1}].
\end{eqnarray*}
Since $N_{4}\in\mathbb{N}[q,q^{-1},Q,Q^{-1}]$, we have 
$\Psi_{D}\in\mathbb{N}[q,q^{-1},Q,Q^{-1}]$.

\end{proof}

\subsection{Sum rule}
Let $S^{\mathrm{X}}_{N}(q,Q)$ be the sum of all components of $\Psi$ on the 
Kazhdan--Lusztig bases of Type X (X=A, BI, BII or BIII), {\it i.e.,} 
$S^{X}_{N}(q,Q)=\sum_{D}\Psi_{D}$.

Let $A_{n}$ be the number of $n\times n$ symmetric binary matrix
with no row sum greater than one.
$A_n$ satisfies the recurrence relation 
\begin{eqnarray*}
A_n=2A_{n-1}+(n-1)A_{n-2},
\end{eqnarray*}
with $A_0=1$ and $A_{1}=2$.
The sequence $A_n$ is A005425 in \cite{OEIS} (see also \cite{Rob76}).

Let $B_n$ be the number of $n\times n$ bisymmetric binary matrix
with a row sum equal to one, that is, 
$B_n$ is the set of permutation matrices with symmetric about two diagonals 
and modulo rotation by $\pi/2$ radians. 
$B_n$ satisfies the recurrence relation
\begin{eqnarray*}
B_n=2B_{n-1}+(2n-2)B_{n-2},
\end{eqnarray*}
with $B_1=1$ and $B_2=3$.
The sequence $B_{n}$ is A000902 in \cite{OEIS}.

Let $C_n$ be the sequence A083886 in \cite{OEIS}, which satisfies
\begin{eqnarray*}
C_{n+1}=3C_{n}+2(n-1)C_{n-1}
\end{eqnarray*}
with $C_1=1$ and $C_{2}=3$. 
The sequence $C_n$ is the total number of signed permutations 
of size $2(n-1)\times 2(n-1)$ which are  
invariant under both diagonal and anti-diagonal reflections 
and avoiding a pattern $(-2,-1)$ \cite{HarTro12}.

\begin{conj}
At $q=1$ and $Q=1$, we have 
\begin{eqnarray*}
S^{\mathrm{A}}_{N}&=&A_N, \\
S^{\mathrm{BI}}_{N}&=&B_{N+1}, \quad \text{for } M=1,\\
S^{\mathrm{BI}}_{N}&=&C_{N+1}, \quad \text{for } M=\infty, \\
S^{\mathrm{BIII}}_{N}&=&C_{N+1}.
\end{eqnarray*}
\end{conj}
At $q=Q=1$, we have $S^{BIII}_{N}=S^{BI}_{N}$ for $M=\infty$.
This coincidence comes from the fact that the diagram $D$ of 
type BI for $M=\infty$ is the same as the diagram of type 
BIII with $Q=q^{M}$ ($M$ large enough).
The first few values of $S^{X}_{N}$ are in Table~\ref{Table-sum}.
\begin{table}
\begin{center}
\begin{tabular}{c|rrrrrrrrr}
 & \multicolumn{9}{c}{Size $N$} \\ 
 & 1 & 2 & 3 & 4 & 5 & 6  & 7 & 8 & 9 \\ \hline
$S^{\mathrm{A}}_{N}$ & 2 & 5 & 14 & 43 & 142 & 499 & 1850 &
 7193 & 29186\\
$S^{\mathrm{BI}}_{N} (M=1) $ & 3 & 10 & 38 & 156 & 692 & 3256 & 
16200 & 84496 & 460592\\
$S^{\mathrm{BI}}_{N} (M=2)$ & 3 & 11 & 44 & 192 & 892 & 4396 & 
22752 & 123248 & 695024 \\
$S^{\mathrm{BI}}_{N} (M=3)$ & 3 & 11 & 45 & 200 & 952 &  
4796 & 25412 & 140720 & 811280 \\
$S^{\mathrm{BI}}_{N} (M=\infty)$ & 3 & 11 & 45 & 201 & 963 & 
4899 & 26253 & 147345 & 862083 \\
$S^{\mathrm{BII}}_{N}$ & 3 & 9 & 33 & 129 & 555 & 2529 & 12273 & 
62481 & 333603\\
$S^{\mathrm{BIII}}_{N}$ & 3 & 11 & 45 & 201 & 963 & 4899 & 26253 &
 147345 & 862083 
\end{tabular}
\end{center}
\caption{The first few values of $S^{X}_{N}$}
\label{Table-sum}
\end{table}

\subsubsection{Type A and Type BIII}
At $q=1$, the sum $S^{A}_{N}$ is uniquely written as 
\begin{eqnarray*}
S^{A}_{N}=\sum_{i=0}^{N}S_{N,i}Q^{i}
\end{eqnarray*}
where $S_{N,i}\in\mathbb{N}$.
 
Since a component $\Psi_{D}$ is invariant under the bar involution, 
the sum $S^{III}_{N}:=S^{\mathrm{BIII}}_{N}$ is also invariant.
At $q=1$, $S^{III}_{N}$ is uniquely written as 
\begin{eqnarray*}
S^{III}_{N}=\sum_{i=0}^{N}S'_{N,i}(Q+Q^{-1})^{i}.
\end{eqnarray*}
with $S'_{N,i}\in\mathbb{N}$.
\begin{conj}
We have $S_{N,i}=S_{N,N-i}=S'_{N,i}=S'_{N,N-i}$ and 
\begin{eqnarray}
S_{N,i}&=&\prod_{k=0}^{i-1}\frac{N-k}{2k+2}\cdot P_{i}(N), \\
P_{i}(N)&=&N^{i}+\sum_{j=0}^{i-1}p_{i,j}N^{j}
\end{eqnarray}
where $p_{i,j}\in\mathbb{Z}$.
\end{conj}
The first few polynomials $P_{i}(N)$'s are 
\begin{eqnarray*}
P_{1}(N)&=&N+1, \\
P_{2}(N)&=&N^{2}-N+2, \\
P_{3}(N)&=&N^{3}-6N^{2}+17N-16, \\
P_{4}(N)&=&N^{4}-14N^{3}+83N^{2}-230N+248. 
\end{eqnarray*}

Let $\mathcal{A}_n$ be a set of symmetric binary matrices of size 
$n\times n$ with no row sum greater than one.
For $a=(a_{i,j})_{1\le i,j\le n}\in\mathcal{A}_{n}$, we define the 
weight of $a$ by $\mathrm{wt}(a):=\#\{a_{i,j}=1| i\le j \}$. 
Then,  
\begin{conj}
We have 
\begin{eqnarray*}
S_{N,i}=\#\{a\in\mathcal{A}_{N}|\mathrm{wt}(a)=i\}.
\end{eqnarray*}
\end{conj}

Let $\mathcal{C}_{n}$ be a set of signed permutation matrices of size 
$2(n-1)\times2(n-1)$ which are invariant under the diagonal and anti-diagonal 
reflections and avoids the pattern $(-2,-1)$. 
For $c=(c_{i,j})_{1\le i,j\le 2n}\in\mathcal{C}_{n}$, we define the weight of $c$ by 
\begin{eqnarray*}
\mathrm{wt}(c):=n_{+}-n_{-}, 
\end{eqnarray*}
where
\begin{eqnarray*}
n_{+}&:=&\#\{c_{i,j}=1| 1\le i\le n, i\le j\le n \}, \\
n_{-}&:=&\#\{c_{i,j}=1| 1\le i\le n, n+1\le j\le 2n+1-i \}.  
\end{eqnarray*}
At $q=1$, the sum $S_{N}$ is rewritten as 
\begin{eqnarray*}
S_{N}=\sum_{i=0}^{N}\widetilde{S}_{N,N-2i}Q^{N-2i}. 
\end{eqnarray*}
where $\widetilde{S}_{N,i}\in\mathbb{N}$. 
Then, 
\begin{conj}
we have 
\begin{eqnarray*}
\widetilde{S}_{N,i}=\#\{c\in\mathcal{C}| \mathrm{wt}(c)=i\}.
\end{eqnarray*}
\end{conj}

\subsubsection{Type BII}
At $q=1$, the sum $S_{N}:=S^{BII}_{N}$ is uniquely written as 
\begin{eqnarray*}
S_{N}=\sum_{j=0}^{N}S_{N,j}(Q+Q^{-1})^{j}. 
\end{eqnarray*}
where $S_{N,j}\in\mathbb{N}$.
Then, we have two conjectures:
\begin{conj}
We have 
\begin{eqnarray*}
S_{N,1}=\frac{1}{8}(2N^{2}+4N+1-(-1)^{N}).
\end{eqnarray*}
\end{conj}
We have checked the conjecture up to $N=20$.
\begin{conj}
We have 
\begin{eqnarray*}
S_{N,N-i}&=&\prod_{j=1}^{i}(2j)^{-1}\cdot P_{i}(N), \\
P_{i}(N)&=&N^{2i}-iN^{2i-1}+\sum_{k=0}^{2i-2}p_{i,k}N^{k}, 
\end{eqnarray*}
where $p_{i,k}\in\mathbb{Z}$.
\end{conj}
We have checked the conjecture up to $j=7$ and $N=20$.
The first few polynomials $P_{j}(N)$'s are 
\begin{eqnarray*}
P_{1}(N)&=&N^{2}-N+2, \\
P_{2}(N)&=&N^{4}-2N^{3}+3N^{2}+14N-8, \\
P_{3}(N)&=&N^{6}-3N^{5}+N^{4}+51N^{3}-2N^{2}-96N+96. 
\end{eqnarray*}

\subsection{\texorpdfstring{Components of $\Psi$ and enumerations 
of binary/permutation matrices}{Components of Psi and enumerations 
of binary/permutation matrices}}
\subsubsection{Type A}
In the case of Type A, some components are related to 
an enumeration of symmetric binary matrices. 

Let $\mathcal{S}$ be the set of symmetric binary matrices  
with no row sum greater than one. 
Let $s=(s_{ij})_{1\le i,j\le N}\in \mathcal{S}$.
Since $s$ is symmetric, we consider only $s_{ij}$ with $i\le j$.
We denote by $N_1$ the number of one and define 
\begin{eqnarray*}
A(s):=\{(i,j)| s_{ij}=1, i\le j\}.
\end{eqnarray*}
For $l=(i,j)\in A(s)$, we define 
$n_{l}:=N-2i+j$.
We define a map $F:\mathcal{S}\rightarrow\mathbb{N}[q,q^{-1},Q]$ 
by
\begin{eqnarray*}
F(s):=q^{N(N-1)/2}Q^{N-N_1}\prod_{l\in A(s)}q^{-n_{l}}.
\end{eqnarray*}

Let $l=(i,j)$ and $l'=(i',j')$ be elements in $A(s)$. 
We call $s$ admissible if there exists no pair ($l, l'$) such 
that $i<i'<j<j'$. 
For an admissible symmetric binary matrix $s$, we define a map 
$g$ from $s\in\mathcal{S}$ to a binary string $b=(b_1\ldots b_N)$ 
of length $N$: 
for each $l=(i,j)\in A(s)$, we set $b_j=-$ and otherwise we set 
$b_{j}=+$.
Then we consider a set 
\begin{eqnarray*}
\mathrm{Adm}(b):=\{s\in\mathcal{S}|s\text{ is admissible, } g(s)=b\}.
\end{eqnarray*}

\begin{example}
Let $s$ be 
\begin{eqnarray*}
\begin{pmatrix}
0 & 0 & 1 & 0 & 0 \\
  & 1 & 0 & 0 & 0 \\
  &   & 0 & 0 & 0 \\
  &   &   & 0 & 1 \\
  &   &   &   & 0 \\  
\end{pmatrix}.
\end{eqnarray*}
Then $s$ is admissible, $F(s)=q^{-1}Q^{2}$ and $g(s)=+--+-$.
\end{example}

Let $b=\underbrace{-\cdots-}_{i}\underbrace{+\cdots+}_{j}
\underbrace{-\cdots-}_{k}\underbrace{+\cdots+}_{N-i-j-k}$ with $i,j,k\ge0$ 
and $D(b)$ be a diagram of type A associated with $b$. Then 
we have 
\begin{conj}
\begin{eqnarray*}
\Psi_{D(b)}=\sum_{s\in \mathrm{Adm}(v)}F(s)
\end{eqnarray*}
\end{conj}

\subsubsection{Type BIII}
Let $\mathcal{C}_{N}$ be the set of signed permutation matrices 
of size $2(N-1)\times 2(N-1)$ which are symmetric both diagonal 
and anti-diagonal reflections and avoid the pattern $(-2,-1)$.

Let $c=(c_{i,j})_{1\le i,j\le N}\in\mathcal{C}_{N+1}$. 
Since $c$ is bisymmetric, we consider only $c_{i,j}$ with 
$1\le i\le N$ and $i\le j\le 2N+1-i$.
Define 
\begin{eqnarray*}
A(c):=\{(i,j) | c_{i,j}=1, 1\le i\le N, i\le j\le 2N\}.
\end{eqnarray*}
We call an element of $A(c)$ a {\it link}.
A link $(i,j)\in A(c)$ is said to be {\it fundamental} 
if and only if $1\le i\le N$ and $1\le j\le 2N+1-i$.
We denote by $A^{+}(c)$ the set of fundamental links.
Given two links $l=(i_1,j_1)$ and $l'=(i_2,l_2)$, we define 
a {\it cross point} $p(l,l')=(i_2,j_1)$ if and only if 
$i_1<i_2<j_1<j_2$, $i_1+j_1\neq 2N+1$ and $i_2+j_2\neq 2N+1$. 
We denote by $B(c)$ the set of cross points $p(l,l')$ 
for $l,l'\in A(c)$.
We consider the coordinate system where $x$-direction is rightward
and $y$-direction is downward.
We define a down-left path from a link $l=(i_1,j_1)\in A^{+}(c)$  
to a diagonal point $(i_2,i_2)$ as follows.
Note that a link $l$ does not belong to $B(c)$. 
First, we go down from $l$ until it reaches to a cross point or 
a diagonal point. 
If a path reaches a cross point form up (resp. right), we make a 
turn on the cross point and go left (resp. down). 
We continue this procedure until a path reaches to a diagonal 
point. 
We define a binary string $b=b_1\ldots b_{N}\in\{\pm\}^{N}$ of 
length $N$ by $b_i=-$ if there exists a path from 
$l\in A^{+}(c)$ to the diagonal point $(i,i)$, and 
$b_i=+$ otherwise.
Therefore, by composing the above procedure, we have a map 
$B:\mathcal{C}_{N}\rightarrow \{\pm\}^{N}$. 

Suppose that $b=\underbrace{-\cdots-}_{i}\underbrace{+\cdots+}_{j}
\underbrace{-\cdots-}_{k}\underbrace{+\cdots+}_{N-i-j-k}$ with 
$i,j,k\ge0$ and $D(b)$ be a diagram of type BIII associated with $b$.
Then,  
\begin{conj}
At $q=Q=1$, we have 
\begin{eqnarray*}
\Psi_{D(b)}=\#\{c\in\mathcal{C}_{N+1}|B(c)=b\}.
\end{eqnarray*} 
\end{conj}

\appendix
\section{}
\label{sec-appendix}
\begin{lemma}
\label{lemma-app-A}
Let $A=(a_{i,j})_{1\le i,j\le N+1}$ be a tridiagonal matrix whose entries are 
\begin{eqnarray*}
a_{i,i}=q^{N+2-2i}\frac{Q-Q^{-1}}{q-q^{-1}}, \qquad
a_{i,i-1}=[N+1-i],\qquad
a_{i,i+1}=q^{N-2i}[i].
\end{eqnarray*}
Then, the eigenvalues of $A$ are 
\begin{eqnarray}
\label{app-eigenvalue}
\frac{Qq^{N-2\lambda}-Q^{-1}q^{-N+2\lambda}}{q-q^{-1}}, \qquad
\lambda=0,1,\ldots,N.
\end{eqnarray}
\end{lemma}
\begin{proof}
Let $x_{\lambda}$ be the expression (\ref{app-eigenvalue}).
To show that $x_{\lambda}$ is the eigenvalue of $A$, it is enough 
to show that the determinant of $A^{(1)}=A-x_{\lambda}\mathbf{1}$ is equal 
to zero.
We diagonalize the tridiagonal matrix $A^{(1)}=(a^{(1)}_{i,j})$ from the right bottom corner. 
First, we subtract $a^{(1)}_{N+1,N}/a^{(1)}_{N+1,N+1}$ times the $(N+1)$-th column vector 
from the $N$-th column and obtain a matrix $A^{(2)}=(a^{(2)}_{i,j})$. 
Then, we subtract $a^{(2)}_{N,N-1}/a^{(2)}_{N,N}$ times the $N$-th column vector from the 
$(N-1)$-th column. 
We continue this procedure until we obtain an upper triangular matrix $L=(l_{i,j})$.
From a direct computation, 
the diagonal entries of the matrix $L$ is written in terms of a set of Laurent polynomials 
$\{v_1,\ldots,v_{N+2}\}$ as $l_{i,i}=v_i/v_{i+1}$ where 
$v_i=a^{(1)}_{i,i}v_{i+1}-a^{(1)}_{i+1,i}a^{(1)}_{i,i+1}v_{i+2}$ with the initial conditions 
$v_{N+2}=1$ and $v_{N+1}=a^{(1)}_{N+1,N+1}$.
Note that $a^{(1)}_{i,i+1}=a_{i,i+1}$, $a^{(1)}_{i,i-1}=a_{i,i-1}$ and 
\begin{eqnarray*}
a^{(1)}_{i,i}=q^{N+1-i-\lambda}[\lambda+1-i]Q+q^{1-i+\lambda}[\lambda-N-1+i]Q^{-1}.
\end{eqnarray*}

We will prove that for $1\le n\le N$ 
\begin{eqnarray}
\label{app-v}
v_{n}=\sum_{j=0}^{m(n)}Q^{m(n)-2j}\alpha(n,j),
\end{eqnarray}
where
\begin{eqnarray*}
\alpha(n,j)
&:=&q^{d(n,j)}\genfrac{[}{]}{0pt}{}{m(n)}{j}
\prod_{i=0}^{m(n)-j-1}[\lambda-N+i]\prod_{i=0}^{j-1}[\lambda-i], \\
m(n)&:=&N+2-n, \\
d(n,j)&:=&m(n)(m(n)-1)/2-\lambda m(n)+j(-N+2\lambda).
\end{eqnarray*}
For $n=N+1, N+2$, Eqn.(\ref{app-v}) holds true.
We assume that Eqn.(\ref{app-v}) is true up to some $n\le N+1$.
We have
\begin{eqnarray*}
v_{n-1}&=&a^{(1)}_{n-1,n-1}v_{n}-a_{n,n-1}a_{n-1,n}v_{n+1} \\
&=&\sum_{j=0}^{m(n)+1}Q^{m(n-1)-2j}
\left(
q^{N+2-n-\lambda}[\lambda+2-n]\alpha(n,j)+
q^{2-n+\lambda}[\lambda-N-2+n]\alpha(n,j-1)
\right) \\
&&-\sum_{j=0}^{m(n)+1}Q^{m(n-1)-2j}
q^{N-2n+3}[n-1][N+2-n]\alpha(n+1,j-1),
\end{eqnarray*}
where $\alpha(n,m(n)+1)=\alpha(n,-1)=0$.
By a direct computation, the above expression is equal to 
Eqn.(\ref{app-v}).

The determinant of $A^{(1)}$ is equal to the one of $L$, that is,
$\prod_{i=1}^{N+1}l_{i,i}=v_{1}$.
The explicit expression of $v_1$ is 
\begin{eqnarray*}
v_1=\prod_{i=1}^{N}[\lambda-i]
\sum_{j=0}^{N+1}Q^{N+1-2j}q^{d(1,j)}\genfrac{[}{]}{0pt}{}{N+1}{j}.
\end{eqnarray*}
Note that $0\le\lambda\le N$.
We have $v_1=0$, {\it i.e.}, the eigenvalues of $A$ is $x_{\lambda}$
for $0\le\lambda\le N$.
\end{proof}

\begin{lemma}[Lemma~A.1 in \cite{Shi14-2}]
\label{lemma-app0}
Set $M:=\sum_{i=1}^{I}m_i$. 
We have 
\begin{eqnarray}
\label{app-0}
\sum_{i=1}^{I}[m_i]
\left[1+\sum_{j=1}^{i}n_j\right]
\frac{\prod_{j=i+1}^{I}[1+\sum_{k=1}^{j}(n_k+m_k)]}
{\prod_{j=i}^{I}[1+n_j+\sum_{k=1}^{j-1}(n_k+m_k)]}
=[M]
\end{eqnarray}
\end{lemma}
\begin{proof}
We prove Lemma by induction.
Let $f(I)$ be the left hand side of Eqn.(\ref{app-0})
When $I=1$, we have $f(I)=m_1$ by a straightforward 
calculation.
Set $N=\sum_{i=1}^{I+1}n_i$. 
We assume that Lemma holds true up to $I$.
We have 
\begin{eqnarray*}
f(I+1)&=&f(I)\frac{[1+\sum_{k=1}^{I+1}(n_k+m_k)]}
{[1+n_{I+1}+\sum_{k=1}^{I}(n_k+m_k)]}
+\frac{[m_{I+1}][1+N']}{[1+N+M]} \\
&=&\frac{[M][1+M+N+m_{I+1}]}{[1+M+N]}
+\frac{[m_{I+1}][1+N']}{[1+N+M]} \\
&=&[M+m_{I+1}].
\end{eqnarray*}
\end{proof}

\begin{lemma}
\label{lemma-app1}
Set $M=\sum_{i=1}^{J}m_i$ and $N=\sum_{i=1}^{J+1}n_j$. 
We have 
\begin{eqnarray}
\label{app-1}
\sum_{i=1}^{J}\left[1+\sum_{j=1}^{i}n_j\right][m_i]
\frac{\prod_{j\ge i+2}^{J+1}[1+\sum_{k=1}^{j-1}(n_k+m_k)]}
{\prod_{j\ge i}^{J+1}[1+n_j+\sum_{k=1}^{j-1}(n_k+m_k)]}
=
\frac{[M]}{[M+N+1]}
\end{eqnarray}
\end{lemma}
\begin{proof}
We prove Lemma by induction. When $J=1$, Lemma holds true by a direct 
computation.
We assume that Lemma is true up to some $J-1\ge1$.
Let $f(J)$ be the left hand side of Eqn.(\ref{app-1}) and 
$w_i=1+n_{i}+\sum_{i=1}^{i-1}(n_i+m_i)$.
We have 
\begin{eqnarray*}
f(J)&=&f(J-1)\frac{[1+\sum_{i=1}^{J}(n_i+m_i)]}
{[w_{J+1}]}+
\frac{[1+\sum_{i=1}^{J}n_i][m_{J}]}
{[w_J][w_{J+1}]}  \\
&=&\frac{[M-m_{J}][1+M+N-n_{J+1}]+[N+1-n_{J+1}][m_{J}]}
{[M+N-m_{J}-n_{J+1}+1][1+M+N]} \\
&=&\frac{[M]}{[M+N+1]}.
\end{eqnarray*}
This completes the proof.
\end{proof}

\begin{lemma}
\label{lemma-app2}
Set $M_i=\sum_{j=1}^{i}m_j$.
\begin{eqnarray}
\label{app-2}
\sum_{i=1}^{I}
\frac{[m_i]}{[x+M_{i-1}][x+M_{i}]}
=\frac{[M_{I}]}{[x][x+M_{I}]}
\end{eqnarray}
\end{lemma}
\begin{proof}
We prove Lemma by induction. 
Let $f(I)$ be the left hand side of Eqn.(\ref{app-2}). 
Lemma is true when $I=1$.
We assume that Lemma holds true up to $I$.
We have 
\begin{eqnarray*}
f(I+1)&=&f(I)+\frac{[m_{I+1}]}{[x+M_{I+1}][x+M_{I}]} \\
&=&\frac{[M_{I+1}]}{[x][x+M_{I+1}]}
\end{eqnarray*}
\end{proof}

\begin{lemma}
\label{lemma-app-eNA}
Set $v_{i}:=\sum_{j=1}^{i-1}n_j+m_j$.
We have 
\begin{eqnarray}
\label{app-8}
\frac{q^{\sum_{i=1}^{I}m_i}}{[v_{I+1}]}
\prod_{j=1}^{I}\frac{[n_j+v_j]}{[v_j]}
+
\sum_{i=1}^{I}\frac{q^{-\sum_{j=1}^{i}n_j}[m_i]}
{[v_i][v_{i+1}]}
\prod_{j=1}^{i-1}\frac{[n_j+v_j]}{[v_j]}
=1.
\end{eqnarray}
\end{lemma}
\begin{proof}
We prove Lemma by induction.
Let $f(I)$ be the left hand side of Eqn.(\ref{app-8}). 
By a straightforward calculation, we have $f(1)=1$.
We assume Lemma holds true up to $I$.
We have 
\begin{eqnarray*}
f(I+1)&=&f(I)-\frac{q^{\sum_{i=1}^{I}m_i}}{[v_{I+1}]}
\prod_{j=1}^{I}\frac{[n_j+v_j]}{[v_j]}
+\frac{q^{-\sum_{i=1}^{I+1}n_i}[m_{I+1}]}{[v_{I+1}][v_{I+2}]}
\prod_{i=1}^{I}\frac{[n_i+v_i]}{[v_i]} \\
&&+\frac{q^{\sum_{i=1}^{I+1}m_i}}{[v_{I+2}]}
\prod_{j=1}^{I+1}\frac{[n_j+v_j]}{[v_j]} \\
&=&1.
\end{eqnarray*}
\end{proof}

\begin{lemma}
\label{lemma-app-eN2}
Set $v_i:=\sum_{j=1}^{i-1}m_j+n_j$. 
We have 
\begin{eqnarray}
\label{app-9}
\sum_{i=1}^{I}\frac{q^{-\sum_{j=1}^{i-1}n_j-1}[m_i]}{[1+v_i][1+v_{i+1}]}
\prod_{j=i+1}^{I}\frac{[1+m_j+v_{j}]}{[1+v_{j+1}]}
=
\prod_{j=1}^{I}\frac{[1+m_j+v_{j}]}{[1+v_{j+1}]}
-\frac{q^{\sum_{i=1}^{I}m_i}}{[1+v_{J+1}]}.
\end{eqnarray}
\end{lemma}
\begin{proof}
We prove Lemma by induction.
Let $f(I)$ be the left hand side of Eqn.(\ref{app-9}).
By a straightforward calculation, Lemma is true for $I=1$.
We assume that Lemma holds true up to $I$.
We have 
\begin{eqnarray*}
f(I+1)&=&f(I)\frac{[1+m_{I+1}+v_{I+1}]}{[1+v_{I+2}]}
+\frac{q^{-\sum_{i=1}^{I}n_i-1}[m_{I+1}]}{[1+v_{I+1}][1+v_{I+2}]} \\
&=&\prod_{i=1}^{I+1}\frac{[1+m_i+v_{i}]}{[1+v_{i+1}]}
-\frac{q^{\sum_{j=1}^{I}m_j}[1+m_{I+1}+v_{I+1}]}{[1+v_{I+1}][1+v_{I+2}]}
+\frac{q^{-\sum_{i=1}^{I}n_i-1}[m_{I+1}]}{[1+v_{I+1}][1+v_{I+2}]} \\
&=&\prod_{i=1}^{I+1}\frac{[1+m_i+v_{i}]}{[1+v_{i+1}]}
-\frac{q^{\sum_{i=1}^{I+1}m_i}}{[1+v_{J+2}]}.
\end{eqnarray*}
\end{proof}

\begin{lemma}
\label{lemma-app-eNA3}
Set $v_i:=\sum_{j=1}^{i-1}n_j+m_j$ and $w_i:=1+n_i+v_i$. 
We have 
\begin{multline}
\label{app-10}
q^{-\sum_{j=1}^{I}m_j+1}
\sum_{i=1}^{I}\frac{q^{-\sum_{j=1}^{i}n_j}[m_i]}{[w_i][w_{i+1}]}
\prod_{k=i+2}^{I+1}\frac{[1+v_{k}]}{[w_k]}
\prod_{k=1}^{i-1}\frac{[n_k+v_k]}{[v_{k+1}]}  \\
\times\left\{(1+q^{-2})q^{\sum_{j=1}^{i-1}m_j}
-\frac{q^{-\sum_{j=1}^{i}n_j-1}[m_{i}-1]}{[v_{i+1}]}
\right\}    \\
=q^{-\sum_{i=1}^{I}m_i}\prod_{j=1}^{I+1}\frac{[1+v_j]}{[w_{j}]}
-\frac{q^{\sum_{i=1}^{I}m_i}}{[w_{I+1}]}
\prod_{j=1}^{I}\frac{[n_j+v_j]}{[v_{j+1}]}.
\end{multline}
\end{lemma}
\begin{proof}
We prove Lemma by induction.
Let $f(I)$ be the left hand side of Eqn.(\ref{app-10}). 
By a straightforward calculation, Lemma is true for $I=1$.
We assume that Lemma holds true up to $I$.
We have 
\begin{eqnarray*}
f(I+1)&=&q^{-m_{I+1}}\frac{[1+v_{I+2}]}{[w_{I+2}]}f(I)
+\frac{q^{1-v_{I+2}}[m_{I+1}]}{[w_{I+1}][w_{I+2}]}
\prod_{j=1}^{I}\frac{[n_j+v_j]}{[v_{j+1}]} \\
&&\times\left\{
(1+q^{-2})q^{\sum_{j=1}^{I}m_j}
-\frac{q^{-\sum_{j=1}^{I+1}n_j-1}[m_{I+1}-1]}{[v_{I+2}]}
\right\} \\
&=&q^{-\sum_{i=1}^{I+1}m_j}\prod_{i=1}^{I+2}\frac{[1+v_i]}{[w_i]}
-\frac{q^{\sum_{i=1}^{I}m_i-m_{I+1}}[1+v_{I+2}]}{[w_{I+1}][w_{I+2}]}
\prod_{i=1}^{I}\frac{[n_i+v_i]}{[v_{i+1}]}   \\
&&+\frac{q^{1-v_{I+2}}[m_{I+1}]}{[w_{I+1}][w_{I+2}]}
\prod_{i=1}^{I}\frac{[n_i+v_i]}{[v_{i+1}]}
\left\{
(1+q^{-2})q^{\sum_{j=1}^{I}m_j}
-\frac{q^{-\sum_{j=1}^{I+1}n_j-1}[m_{I+1}-1]}{[v_{I+2}]}
\right\} \\
&=&
q^{-\sum_{i=1}^{I+1}m_i}\prod_{j=1}^{I+2}\frac{[1+v_j]}{[w_{j}]}
-\frac{q^{\sum_{i=1}^{I+1}m_i}}{[w_{I+2}]}
\prod_{j=1}^{I+1}\frac{[n_j+v_j]}{[v_{j+1}]}.
\end{eqnarray*}
\end{proof}

\begin{lemma}
\label{lemma-app-eNA4}
Set $v_{i}:=\sum_{j=1}^{i}(n_j+m_j)$ and $w_i=1+n_{i}+v_{i}$. 
We have 
\begin{eqnarray}
\label{app-11}
\sum_{i=1}^{I}q^{-1-\sum_{j=1}^{i}n_{j}}
\frac{[m_i]}{[w_{i}][w_{i+1}]}
\prod_{j=i+2}^{I+1}\frac{[1+v_{j}]}{[w_{j}]}
=
\prod_{i=1}^{I+1}\frac{[1+v_{i}]}{[w_{i}]}
-\frac{q^{\sum_{i=1}^{I}m_{i}}}{[w_{J+1}]}
\end{eqnarray}
\end{lemma}
\begin{proof}
We prove Lemma by induction. 
Let $f(I)$ be the left hand side of Eqn.(\ref{app-11}). 
By a straightforward calculation, Lemma holds true when 
$I=1$.
We have 
\begin{eqnarray*}
f(I+1)&=&f(I)\frac{[1+v_{I+2}]}{[w_{I+2}]}
+\frac{q^{-1-\sum_{j=1}^{I+1}n_{j}}[m_{I+1}]}{[w_{I+1}][w_{I+2}]} \\
&=&\prod_{i=1}^{I+2}\frac{[1+v_{i}]}{[w_{i}]}
-\frac{1}{[w_{I+1}][w_{I+2}]}\left(
q^{\sum_{i=1}^{I}m_i}[1+v_{I+2}]-q^{-1-\sum_{i=1}^{I+1}n_i}[m_{I+1}]
\right) \\
&=&\prod_{i=1}^{I+2}\frac{[1+v_{i}]}{[w_{i}]}
-\frac{q^{\sum_{i=1}^{I+1}m_i}}{[w_{J+2}]}.
\end{eqnarray*}

\end{proof}

\begin{lemma}[{\cite[Lemma A.2]{Shi14-2}}]
\label{lemma-app-13}
We have 
\begin{eqnarray*}
\sum_{i=1}^{K}I_{i}\cdot J_{i}
+
\frac{[2z+2]}{[1+x+\sum_{i=1}^{K}m_i]}I_{K}
=
[1+x]^{-1}
\left[2+2z+2\sum_{i=1}^{K}m_i\right]
\end{eqnarray*}
where 
\begin{eqnarray*}
I_{i}&:=&\prod_{j=1}^{i}
\frac{[2+2z+2\sum_{k=j}^{K}m_{k}]}{[2+2z+m_j+2\sum_{k=j+1}^{K}m_{k}]}, \\
J_{i}&:=&
\frac{[m_i][x+3+2z+\sum_{j=1}^{i}m_{j}+2\sum_{j=i+1}^{K}m_{j}]}
{[1+x+\sum_{j=1}^{i-1}m_{j}][1+x+\sum_{j=1}^{i}m_{j}]}.
\end{eqnarray*}
\end{lemma}

\begin{lemma}
\label{lemma-app-15}
Set $v_{i}:=\sum_{j=1}^{i}(n_j+m_j)$.
We have 
\begin{eqnarray}
\label{app-15}
\sum_{i=1}^{I}q^{-\sum_{j=1}^{i}n_{j}}\frac{[m_i]}{[v_{i}]}
\prod_{l=1}^{i-1}\frac{[n_l+v_{l-1}]}{[v_{l}]}
=
1-q^{\sum_{i=1}^{I}m_i}\prod_{l=1}^{I}\frac{[n_l+v_{l-1}]}{[v_{l}]}
\end{eqnarray}
\end{lemma}
\begin{proof}
We prove Lemma by induction on $I$. 
Let $f(I)$ be the left hand side of Eqn.(\ref{app-15}).
By a straightforward calculation, Lemma holds true 
for $I=1$. 
We have 
\begin{eqnarray*}
f(I+1)&=&f(I)+q^{-\sum_{j=1}^{I+1}n_{j}}\frac{[m_{I+1}]}{[v_{I+1}]}
\prod_{l=1}^{I}\frac{[n_l+v_{l-1}]}{[v_{l}]} \\
&=&1-\prod_{l=1}^{I}\frac{[n_l+v_{l-1}]}{[v_{l}]}
\left\{q^{\sum_{i=1}^{I}m_i}-q^{-\sum_{j=1}^{I+1}n_{j}}
\frac{[m_{I+1}]}{[v_{I+1}]}\right\} \\
&=&1-q^{\sum_{i=1}^{I+1}m_i}\prod_{l=1}^{I+1}\frac{[n_l+v_{l-1}]}{[v_{l}]}
\end{eqnarray*}
\end{proof}

\begin{lemma}
\label{lemma-app-16}
Set $v_{i}:=\sum_{j=1}^{i}(n_j+m_j)$.
We have 
\begin{eqnarray}
\label{app-16}
\sum_{i=1}^{I}
\frac{q^{-\sum_{j=1}^{i}n_{j}}[m_i]}{[1+n_i+v_{i-1}]}
\prod_{j=i+1}^{I}\frac{[1+v_{j}]}{[1+n_j+v_{j-1}]}
=q\prod_{j=1}^{I}\frac{[1+v_{j}]}{[1+n_j+v_{j-1}]}
-q^{1+\sum_{i=1}^{I}m_{i}}
\end{eqnarray}
\end{lemma}
\begin{proof}
We prove Lemma by induction. 
Let $f(I)$ be the left hand side of Eqn.(\ref{app-16}).
By a straightforward calculation, Lemma holds true for $I=1$.
We have 
\begin{eqnarray*}
f(I+1)&=&f(I)\frac{[1+v_{I+1}]}{[1+n_{I+1}+v_{I}]}+
\frac{q^{-\sum_{j=1}^{I+1}n_j}[m_{I+1}]}{[1+n_{I+1}+v_{I}]} \\
&=&q\prod_{j=1}^{I+1}\frac{[1+v_{j}]}{[1+n_j+v_{j-1}]}
+\frac{1}{[1+n_{I+1}+v_{I}]}\left\{
q^{-\sum_{j=1}^{I+1}n_j}[m_{I+1}]
-q^{1+\sum_{i=1}^{I}m_i}[1+v_{I+1}]
\right\} \\
&=&q\prod_{j=1}^{I+1}\frac{[1+v_{j}]}{[1+n_j+v_{j-1}]}
-q^{1+\sum_{i=1}^{I+1}m_{i}}
\end{eqnarray*}
\end{proof}

\begin{lemma}
\label{lemma-app-17}
Set $v_i:=\sum_{j=1}^{i}(n_j+m_j)$. 
We have 
\begin{multline}
\label{app-17}
\sum_{i=1}^{I}q^{-\sum_{k=1}^{i}n_k}
\frac{[m_i]}{[1+n_i+v_{i-1}]}
\prod_{j=1}^{i-1}\frac{[n_j+v_{j-1}]}{[v_{j}]}
\prod_{j=i+1}^{I}\frac{[1+v_{j}]}{[1+n_j+v_{j-1}]} \\
\times\left\{
(1+q^{-2})q^{\sum_{k=1}^{i-1}m_k}
-q^{-1-\sum_{k=1}^{i}n_k}\frac{[m_i-1]}{[v_{i}]}\right\} \\
=q^{-1}\prod_{i=1}^{I}\frac{[1+v_{j}]}{[1+n_j+v_{j-1}]}
-q^{2\sum_{k=1}^{I}m_{k}-1}
\prod_{i=1}^{I}\frac{[n_i+v_{i-1}]}{[v_{i}]}
\end{multline}
\end{lemma}
\begin{proof}
We prove Lemma by induction on $I$. 
Let $f(I)$ be the left hand side of Eqn.(\ref{app-17}).
By a direct calculation, Lemma holds true for $I=1$. 
We have 
\begin{eqnarray*}
f(I+1)&=&f(I)\frac{[1+v_{I+1}]}{[1+n_{I+1}+v_{I}]}
+q^{-\sum_{k=1}^{I+1}n_k}\frac{[m_{I+1}]}{[1+n_{I+1}+v_{I}]}
\prod_{j=1}^{I}\frac{[n_j+v_{j-1}]}{[v_{j}]} \\
&&\times \left\{
(1+q^{-2})q^{\sum_{k=1}^{I}m_k}
-q^{-1-\sum_{k=1}^{I+1}n_k}\frac{[m_{I+1}-1]}{[v_{I+1}]}\right\} \\
&=&q^{-1}\prod_{i=1}^{I+1}\frac{[1+v_{i}]}{[1+n_{i}+v_{i-1}]}
-\prod_{i=1}^{I}\frac{[n_i+v_{i-1}]}{[v_{i}]}
\left\{ q^{2\sum_{k=1}^{I}m_k-1}\frac{[1+v_{I+1}]}{[1+n_{I+1}+v_{I}]}\right.\\
&&\left. -q^{-\sum_{k=1}^{I+1}n_{k}}\frac{[m_{I+1}]}{[1+n_{I+1}+v_{I}]}
\left((1+q^{-2})q^{\sum_{k=1}^{I}m_k}
-q^{-1-\sum_{k=1}^{I+1}n_k}\frac{[m_{I+1}-1]}{[v_{I+1}]}
\right)\right\} \\
&=& \text{right hand side of Eqn.(\ref{app-17})}.
\end{eqnarray*}
\end{proof}

\bibliographystyle{amsplainhyper} 
\bibliography{biblio}

\end{document}